\documentclass[lettersize,onecolumn]{IEEEtran}

\usepackage{authblk}
\usepackage{amsmath,amstext,amsfonts,amssymb,amsthm} 
\usepackage{graphicx,color} 
\usepackage[latin1]{inputenc}   
\usepackage[T1]{fontenc} 
\usepackage[english]{babel} 
\usepackage[caption=false,font=normalsize,labelfont=sf,textfont=sf]{subfig}
\usepackage{textcomp}
\usepackage{stfloats}
\usepackage{url}
\usepackage{verbatim}
\usepackage{cite}
\usepackage{bm}
\usepackage{bbm}
\usepackage{tikz}
\usetikzlibrary{positioning}
\usepackage{geometry}
\usepackage{xcolor}

\usepackage{mathtools}
\usepackage{bm}
\usepackage{dsfont}
\usepackage{caption}
\usepackage{fullpage} 
\usepackage{mathrsfs} 
\usepackage{xcolor}
\usepackage{marginnote} 

\usepackage[T1]{fontenc} 
\usepackage[english]{babel} 
\usepackage[plainpages=false,pdfpagelabels,colorlinks=true,linkcolor=red,urlcolor=red,citecolor=blue]{hyperref}

\let\OLDthebibliography\thebibliography
\renewcommand\thebibliography[1]{
  \OLDthebibliography{#1}
  \setlength{\parskip}{0pt}
  \setlength{\itemsep}{0pt plus 0.6ex}
}

\makeatletter

\usepackage{kantlipsum}

\allowdisplaybreaks

\numberwithin{equation}{section}
\newtheorem{theorem}{Theorem}[section]

\newtheorem{lemma}{Lemma}[section]

\newtheorem{proposition}{Proposition}[section]

\newtheorem{remark}{Remark}[section]

\usepackage[plainpages=false,pdfpagelabels,colorlinks=true,linkcolor=red,urlcolor=red,citecolor=blue]{hyperref}

\usepackage[symbol]{footmisc}

\def\ex{\mathbb{E}}

\def\exp{\text{exp}}
\def\X{\mathbf{x}}
\def\Y{\mathbf{Y}}
\def\Z{\mathbf{Z}}
\def\V{\mathbf{V}}
\def\v{\mathbf{v}}
\def\bA{\mathbf{A}}
\def\bB{\mathbf{B}}
\def\G{\mathbf{G}}
\def\bg{\mathbf{g}}
\def\bz{\mathbf{z}}
\def\M{\mathcal{M}}
\def\q{\mathfrak{q}}
\def\tPhi{\tilde{\Phi}}
\def\LLambda{\boldsymbol{\Lambda}}
\def\bbeta{\boldsymbol{\beta}}
\def\boldeta{\boldsymbol{\eta}}

\def\ymmse{\text{ymmse}}
\def\mmse{\text{mmse}}

\newcommand{\Tr}{{\text{Tr}}}

\title{Information-theoretic limits \\and vector approximate message-passing \\for high-dimensional time series}

\author{Daria Tieplova, Samriddha Lahiry,  Jean Barbier
\thanks{D. Tieplova is with The Abdus Salam International Centre for Theoretical Physics, Italy and B. Verkin Institute for Low Temperature Physics and Engineering of the National Academy of Sciences of Ukraine. Email: \url{dtieplov@ictp.it}}%
\thanks{S. Lahiry is with Department of Statistics and Data Science, National University of Singapore,  Singapore. Email: \url{slahiry@nus.edu.sg}}
\thanks{ J. Barbier is with The Abdus Salam International Centre for Theoretical Physics, Italy. Email: \url{jbarbier@ictp.it}}%
}

\begin{document}

\maketitle


\begin{abstract}
    High-dimensional time series appear in many scientific setups, demanding a nuanced approach to model and analyze the underlying dependence structure. Theoretical advancements so far often rely on stringent assumptions regarding the sparsity of the underlying signal. In non-sparse regimes, analyses have primarily focused on linear regression models with the design matrix having independent rows.
In this paper, we expand the scope by investigating a high-dimensional time series model wherein the number of features grows proportionally to the number of sampling points, without assuming sparsity in the signal. Specifically, we consider the stochastic regression model and derive a single-letter formula for the normalized mutual information between observations and the signal, as well as for minimum mean-square errors. We also empirically study the vector approximate message passing VAMP algorithm and show that, despite the lack of theoretical guarantees, its performance for inference in our time series model is robust and often statistically optimal.
\end{abstract}


\section{Introduction}

Analysis of high-dimensional time series are increasingly important with the advent of new technologies, with examples arising from macroeconomic forecasting \cite{li2014forecasting}, genomics \cite{bar2012studying}, measurement of systemic risk in financial markets \cite{basu2024high}, and fMRI studies \cite{sporns2021dynamic}. Typically, these applications involve large datasets which exhibit both cross-sectional and temporal dependence. In many cases, questions related to the analysis of these structures can be formulated as estimation of a signal in a linear regression problem, as shown in \cite{basu2015regularized}.

In the past decade, techniques from the high-dimensional statistics literature have been extended to deal with such dependent setups. However, the results so far have been entirely based on assumptions of low-dimensional structures, which often translate into the estimation of a very sparse signal in a linear regression task. On the other hand, a new paradigm has emerged in the high-dimensional asymptotics literature that are mainly motivated by techniques in statistical physics and information theory, where such stringent assumptions of sparsity are removed and one can analyze the same model when the number of features $p$ and the number of samples $N$ grow instead proportionally. In that regime, a particular problem of interest is computing the minimum mean-square error (MMSE) on the signal in a Bayesian setup, which is closely related to the computation of the mutual information between the signal and the observations. Specifically, one is interested in establishing a single-letter formula for the mutual information per observation in the large $N$ limit, derive the limiting MMSE, and construct a computationally efficient estimator able to asymptotically match the MMSE performance, at least in some regimes. We note that while single-letter formulas for the aforementioned quantities have been rigorously established for linear regression with a design consisting of independent and identically distributed (i.i.d.) Gaussian entries \cite{barbier2016mutual,barbier2018mutual,reeves_2016,barbier2019optimal}, and later extended to deal with more structured designs drawn from certain dependent random matrix ensembles \cite{barbier2018mutual,dudeja2022universality,dudeja2023universality,dudeja2024spectral,lahiry2024universality,saeed_mont}, such results do not cover dependence structures that arise in time series models. In this article, we make an important step in that direction by rigorously establishing single-letter formulas for information-theoretic quantities in a stochastic regression model, where the rows of the measurement matrix come from an AR(1) process \cite{basu2015regularized}.  

\subsection{Setting}

  To introduce our setup formally, let $(Y_\mu,\X_\mu)_{\mu=1}^N$  be $N$ observations from a time series, with the dependence structure outlined below. Let  $\X_\mu$ be a $p$~dimensional, centered, stationary Gaussian process given by
\begin{equation}\label{eq:def_x_t}
    \X_{\mu+1}=\bA_p\X_\mu+\boldsymbol{\xi}_{\mu}, \quad \mu\in\mathbb{Z},
\end{equation}
 where $\bA_p$ is a $p \times p$ deterministic matrix  and $\boldsymbol{\xi}_\mu$ is a $p$-dimensional Gaussian random vector distributed as $\mathcal{N}(0,\sigma_1^2\mathbf{I}_p)$ independently of everything else for each $\mu$.
The observed time series $Y_\mu$ has the form
\begin{equation}\label{eq_SRM}
    Y_\mu=\dfrac{1}{\sqrt{p}}\X_\mu^\intercal\bbeta_0+Z_\mu, \quad \mu=1,\ldots,N
\end{equation}
where $\bbeta_0$ is a signal vector with i.i.d. entries drawn from the prior distribution $P_0$ with bounded support and $\ex\beta_{0,i}^2=\rho$, while $Z_\mu$ is i.i.d. Gaussian noise with law $\mathcal{N}(0,\sigma_2^2)$. We can see that $Y_\mu$ does not depend separately on $\sigma_1$ or $\sigma_2$, but only on their ratio.  Without loss of generality we can set one of the noise variances ($\sigma_1$ or $\sigma_2$) to $1$. In what follows we consider $\sigma_1=1$ and denote $\sigma_2=\sigma$. 

We are interested in the high-dimensional regime where the number of observations $N$ and the dimension $p$ both tend to infinity at the same rate. That is, we assume that $p=p(N)$ and for $c_N:=N/p$ we have
$$\lim_{N\rightarrow \infty}c_N =\lim_{N\to \infty}\frac Np=c \in (0,\infty).$$ 
 Let $\boldsymbol{\Phi}_N=(\X_1^\intercal,\ldots, \X_N^\intercal)^\intercal$ be a $N\times p$ matrix whose rows are the vectors $(\X_\mu)_{\mu=1}^N$. Further, let $\Y_N=(Y_1,\ldots, Y_N)^\intercal$ and $\Z_N=(Z_1,\ldots,Z_N)^\intercal$ be two vectors of length $N$ containing the individual observations and noises, respectively. Then our model can be written in a vector form:
\begin{align}\label{eq:mod_vec}
    \Y_N=\dfrac{1}{\sqrt{p}}\boldsymbol{\Phi}_N\bbeta_0+\Z_N.
\end{align}

The form of $\bA_p$ plays a major role in the type of resulting linear regression. An important example, addressed in, e.g., \cite{Takeda_2006,tulino2013support,barbier2018mutual,rangan2019vector} (with the final work proposing a near-optimal algorithm for the estimation of signal $\boldsymbol{\beta}_0$), is when $\bA_p$ is a symmetric random matrix with uniformly distributed eigenvectors. Then, one has the eigendecomposition $\bA_p=\boldsymbol{U}\boldsymbol{D}\boldsymbol{U}^{-1}$ with $\boldsymbol{U}$ a Haar distributed random matrix and we may rewrite \eqref{eq:def_x_t} as $\boldsymbol{U}^{-1}\X_{\mu+1}=\boldsymbol{D}\boldsymbol{U}^{-1}\X_\mu+\boldsymbol{\xi}_{\mu}$, noting that $\boldsymbol{U}^{-1}\boldsymbol{\xi}_{\mu}=\boldsymbol{\xi}_{\mu}$ in law due to rotational invariance of the Gaussian noise. Thus the time series $\bz_\mu=\boldsymbol{U}^{-1}\X_\mu$ satisfies $\bz_{\mu+1}=\boldsymbol{D}\bz_\mu+\boldsymbol{\xi}_{\mu}$, which upon rewriting $\boldsymbol{\Phi}_N=\tilde{\boldsymbol{\Phi}}_N\boldsymbol{U}$ with $\tilde{\boldsymbol{\Phi}}_N=(\bz_1^\intercal,\ldots,\bz_N^\intercal)^\intercal$ brings us to the linear regression model with a right-rotationally invariant design matrix. 

Our present interest lies in the extreme opposite setting where the previously Haar matrix $\boldsymbol{U}$ is now proportional to the identity so that $\bA_p$ becomes diagonal, with a rather general choice of eigenvalues. Such a case \emph{cannot} be treated with the aforementioned line of work and requires new techniques, both from the analytical and algorithmic perspective. We thus denote
$\bA_p={\rm diag}(\lambda_1,\lambda_2,\ldots, \lambda_p)$  with $|\lambda_i|<1$ for all $i$. Under this assumption, equation (\ref{eq:def_x_t}) breaks down into $p$ independent equations $$x_{\mu+1,i}=\lambda_i x_{\mu,i}+\xi_{\mu,i},\quad i=1,\ldots,p\,,$$ 
which allows us to easily calculate the variance of $x_{\mu,i}$ and the correlations between $x_{\mu,i}$ and $x_{\mu+s,i}$:
\begin{align}
    \ex(x_{\mu+1,i})^2=\lambda_i^2\ex(x_{\mu,i})^2+1.
\end{align}
Since $\X_\mu$ is a stationary process, its variance does not depend on $\mu$ and the last equality gives us:
\begin{align}\label{eq:var_x_t^i}
    \ex(x_{\mu,i})^2=(1-\lambda_i^2)^{-1}.
\end{align}
From this we immediately obtain 
\begin{align}\label{eq:cor_x_t^i}
    \ex(x_{\mu,i}x_{\mu+s,i})=\ex(x_{\mu,i}(\lambda_i^sx_{\mu,i}+\sum_{j=0}^{s-1}\lambda_i^j\xi_{\mu-j+s-1,i}))=(1-\lambda^2_i)^{-1}\lambda_i^s.
\end{align}

Independence of coordinates also implies that the columns of $\boldsymbol{\Phi}_N$ are independent Gaussian vectors and, due to (\ref{eq:var_x_t^i}) and (\ref{eq:cor_x_t^i}), 
we can easily calculate the covariance matrix $\LLambda_{i,N}$ of the $i$-th column $(x_{1,i},x_{2,i},\ldots,x_{N,i})^\intercal$:
\begin{align}\label{eq:def_Lambda}
    \LLambda_{i,N}:=\dfrac{1}{1-\lambda_i^2}\begin{pmatrix}
    1&\lambda_i &\ldots&\lambda_i^{N-1}\\
    \lambda_i&1&\ldots&\lambda_i^{N-2}\\
    \vdots&\vdots&\ddots&\vdots\\
    \lambda_i^{N-1}&\lambda_i^{N-2}&\ldots&1
\end{pmatrix}.
\end{align}
In what follows, we will often drop indices $N,p$ when it is not crucial. 
With  this, we can present the design matrix as $\boldsymbol{\Phi}_N=(\LLambda_{1,N}^{1/2}\bg_{1},\ldots,\LLambda_{p,N}^{1/2}\bg_{p})$, where the $\bg_i$ are i.i.d. standard Gaussian vectors. In this form it is easy to see the connection to the previously well-studied models.

\textit{Case 1: $\lambda_i=0$ for $i=1,\ldots,p.$} 
In this case all $\LLambda_{i,N}$ are the identity and our model collapses to that of standard linear regression, where the design matrix has i.i.d. standard Gaussian elements. This model was studied in the CDMA \cite{tanaka2002statistical}
 and compressed sensing contexts \cite{DMM09,krzakala2012statistical,barbier2016mutual,barbier2020mutual,reeves_2016}. 

\textit{Case 2: $\lambda_i=\lambda\neq0$ for $i=1,\ldots,p.$} 
Again, all columns of $\boldsymbol{\Phi}_N$ will share the same covariance matrix, denoted by  $\LLambda$. This allows us to rewrite $\boldsymbol{\Phi}_N=\LLambda^{1/2}\G_N$, where $\G_N$ is an $N\times p$ matrix with i.i.d. standard Gaussian elements. Since $\G_N$ is rotationally invariant, this results in the design matrix $\boldsymbol{\Phi}_N$  belonging to the class of right-rotationally invariant matrices. Single-letter formulas for the mutual information of such models was obtained in \cite{Takeda_2006,tulino2013support,barbier2018mutual}.

In this work, we aim for an arbitrary diagonal matrix $\bA_p$ with some natural assumptions. However, in order to address the general case, we  first need to consider matrices 
$\bA_p$ for which the number of eigenvalues remains fixed to $k$ as $p$ goes to infinity. This assumption introduces a "block" structure that can be later "smoothed out" by taking afterwards $k\rightarrow\infty$. 

\textit{Case 3:}\label{Assum:A} We assume that diagonal matrix $\bA_p$ has $k$ different eigenvalues $\lambda_1,\ldots,\lambda_k$ ($k$ fixed) associated to the respective index sets $I_1,\ldots,I_k$. More precisely, there is a partition of  $\{1,\ldots,N\}=\bigcup_{i=1}^k I_i$ such that $A_{jj}=\lambda_i$ for $j\in I_i$. Importantly, we require that the cardinality of each such set grows proportionally with $p$, i.e.,  we have $|I_i|/p\rightarrow l_i>0$ for each $i=1,\ldots,k$.

Without loss of generality, we further assume that  $I_1$ contains the first $|I_1|$ integers, $I_2$ contains the next $|I_2|$ consecutive integers and so on, since it affects only the permutation of the columns of $\boldsymbol{\Phi}_N$. Then,
 $\boldsymbol{\Phi}_N$ can be represented as a block matrix with $k$ blocks  of dimensions $N\times |I_i|$, i.e., $(\LLambda_{1,N}^{1/2}\mathbf{G}_N^{(1)},\ldots,\LLambda_{k,N}^{1/2}\mathbf{G}_N^{(k)})$, where for every $i=1,\ldots,k$,  $\mathbf{G}_N^{(i)}$ is an $N\times |I_i|$ independent matrix with i.i.d. standard Gaussian elements. With this, (\ref{eq:mod_vec}) can be rewritten as
\begin{align}
    \Y_N=\frac{1}{\sqrt{p}}(\LLambda_{1,N}^{1/2}\mathbf{G}_N^{(1)},\ldots,\LLambda_{k,N}^{1/2}\mathbf{G}_N^{(k)})\bbeta_0+\Z_N \label{diagonal_lambda_model}.
\end{align}
The block structure of the design matrix naturally propagates to the signal. In what follows we define   $\bbeta_0^{(i)} := (\beta_{0,j})_{j\in I_i} \in \mathbb{R}^{|I_i|}$ as the $i$-th block of $\bbeta_0$, and similarly for $\bbeta^{(i)}$.
We emphasize that in this setting the design matrix is not right-rotationally invariant, contrary to the Cases 1 and 2 described above. Instead, it exhibits a richer "block-wise" right-rotational invariance. Consequently, for the Cases 1 and 2 the MMSE can be obtained by the vector approximate message passing (VAMP) algorithm \cite{rangan2019vector} (when no statistical-to-computational gap is present) which properly exploits the right rotationally invariant structure; see also the related concurrent orthogonal AMP algorithm (OAMP) \cite{ma2016orthogonal}.

In the "block" model, however, the optimality of VAMP (or OAMP) is not guaranteed because of the absence of right rotational invariance. We will show empirically that in spite of the lack of theoretical guarantees for VAMP in this case, it nevertheless remains very robust and can consistently yield Bayes-optimal performance.


\subsection{Related literature}
We now discuss known results on high-dimensional time series. In particular, we compare ours to those from the closely related information theory literature, with a special focus on asymptotic formulas for mutual informations and MMSEs within the realm of (random) regression tasks. As we are going to also explore the algorithmic implications of our work using VAMP, we recall its original applications in scenarios involving right rotationally invariant ensembles of covariates matrices.

Analysis of high-dimensional time series is pervasive in the literature and a 
 non-exhaustive account of the data generating models and estimation methodologies can be found in \cite{banbura2010large,nardi2011autoregressive,chen2013covariance,basu2015regularized,medeiros2016,wu2016performance,davis2016sparse,wong2020lasso}; see also the survey \cite{basu2021survey}. A basic prototypical model of high-dimensional time series is the stochastic regression model (see \eqref{eq:def_x_t} and \eqref{eq_SRM}) which has been analyzed in \cite{basu2015regularized} from a penalized regression perspective. In \cite{basu2015regularized}, the authors consider this model with a sparse signal vector $\beta$ and analyze the estimation risk of the Lasso estimator. We will instead analyze the very same model but from a purely Bayesian perspective, in the regime where the number of features grows proportionally to the number of observations.  

On the other hand the standard linear regression model, i.e.,  \eqref{eq_SRM} with $(Y_\mu,\X_\mu)_{\mu=1}^N$ conditionally i.i.d. observations, has been thoroughly analyzed, as well as its non-linear generalisation \cite{barbier2019optimal}. Originally due to the pioneering work by Tanaka on CDMA \cite{tanaka2002statistical} (where each element
of the design is taken i.i.d. from $\pm1$), the single-letter formula for the
normalized mutual information between measurements and signal has been generalized to many regression tasks using statistical physics, in particular the replica method \cite{mezard2009information} as, e.g., in \cite{krzakala2012statistical}, and later on rigorously proved using independent methods \cite{barbier2016mutual,barbier2019adaptive,barbier2019optimal,reeves_2016}. Since then, there has been a considerable amount of work in the i.i.d. design set-up \cite{DMM09,bayati_mont,LASSO_Gauss, rangan2011generalized,lnk_flra_survey,donoho2016high,Sur_2,maleki20bridge,elkaroui_pnas_1,elkaroui_pnas_2,elkaroui_ptrf,Sur_1,Sur_2,Sur_3,Sur_4,stojnic_1,stojnic_2,stojnic_3,CT_1,CT_3,CT_5}. In the dependent (or structured) design setup less is known. Yet, a single-letter formula for a large class of correlated matrices forming a subclass of the right-rotationally invariant ensemble has been proved in \cite{barbier2018mutual} and later extended in \cite{li2023random}.  

From the algorithmic perspective, for the measurement matrices with i.i.d. Gaussian entries, the 
approximate message-passing (AMP) algorithm \cite{DMM09} has been proved to potentially be Bayes-optimal. Indeed, AMP achieves the performance of the MMSE estimator for a large set of
parameters \cite{barbier2020mutual}. 
For rotationally invariant matrices, different but related
approaches were proposed \cite{cakmak2014s,minka2013expectation,fletcher2016expectation,ma2016orthogonal,rangan2019vector}. In the present paper we use the standard implementation of VAMP \cite{rangan2019vector}. We note that for Cases 1 and 2 our model is right rotationally invariant and VAMP should thus be optimal (again, if not stuck by a computational gap). However, in the more general \textit{Case 3}, the assumption of right rotational invariance breaks down and VAMP is not expected to produce the optimal MSE \emph{a-priori}. However, in Section \ref{sec:sim} we will observe that VAMP works surprisingly well for certain choices of $\lambda_i$s.


\subsection{Notations and outline of the paper}
 For all vectors and matrices we will be using bold notations $\Y_N,\boldsymbol{\Phi},$ etc. and their elements will be $Y_{N,i}, \Phi_{ij},$ etc. $\mathbf{I}_n$ represents the identity matrix of size $n\times n$. We use the standard notations $O(.)$ and $o(.)$ to denote ``asymptotically bounded'' and ``asymptotically vanishing'' respectively, i.e., we say $a_N=O(\gamma_N)$ if $\limsup|a_N/\gamma_N|\leq C$ and $a_N=o(\gamma_N)$ if $\limsup|a_N/\gamma_N|=0$. 
 Expectation with respect to a group of random variables (conditioned on the rest of the random variables) will be denoted as $\ex_X[.]$  where $X$ would denote the relevant collection of random variables. On the other hand $\ex[.]$ and $\langle .\rangle$ will be reserved for unconditional expectation and expectation with respect to the posterior measure respectively. Norm $\|.\|$ is the $\ell_2$ norm of a vector, unless otherwise mentioned. For matrices the $\|\bA\|$ will always denote the largest singular value of $\bA$ (operator norm) and hence the same notation $\|.\|$ will be reserved for the $\ell_2$ norm of a vector and the operator norm of a matrix, the difference being understood from the context. Throughout the notation $C$ will denote a generic constant unless otherwise specified, that may change from place to place. The Kronecker delta function will be denoted as  
$\delta_{\mu,\nu}=1_{\mu=\nu}$. For a Gaussian standard random variable $z$ we denote $Dz$ the integration over its measure.

In Section~\ref{main_res}, we describe the main results of our paper. In Section \ref{sec:sim}, we provide simulation results showcasing the effectiveness of VAMP for recovering the signal $\bbeta_0$ and compare the mean-square error reached by the VAMP estimator to our theoretical predictions for the MMSE. The rigorous proof of the single-letter formula for the mutual information by the adaptive interpolation method \cite{barbier2019adaptive,barbier2019adaptive_2} is provided in Section \ref{int_prf}, while the replica calculations which first allowed us to form a conjecture are presented in Appendix \ref{rep_cal} for completeness. In Section~\ref{sec:mmse} we prove the relation between MMSE and measurement MMSE, and, finally, in Section~\ref{sec:k_infty} we generalize the results for an arbitrary diagonal matrix $\bA_p$.  Appendices~\ref{apx:KMS}-\ref{apx:Nish} contain some known results concerning KMS matrices and Nishimori identity, while in Appendices~\ref{apx:concen_z}-\ref{apx:lemmas}  we prove some auxiliary results used in Section~\ref{sec:proofs}.

\section{Main results\label{main_res}}

In this section we state our main theorem on the asymptotic expression for the normalized mutual information. Consider the stochastic regression model given by equations \eqref{eq_SRM} and \eqref{eq:def_x_t}. We assume that we are in the Bayesian optimal setting, meaning that the prior $P_0$ is known together with the noise distribution as well as all hyperparameters of the problem. This means that we can write the true posterior distribution associated to (\ref{eq:mod_vec}):
\begin{align}
   dP(\bbeta|\Y_N):=\mathbb{P}(\bbeta_0\in (\bbeta,\bbeta+d\bbeta)|\Y_N)=\frac{P_0(\bbeta)\,d\bbeta}{\mathcal{Z}(\Y_N)} \exp\Big\{-\frac{1}{2\sigma^2}\|\Y_N-p^{-1/2}\boldsymbol{\Phi}_N\bbeta\|^2\Big\},\label{posterior_0}
\end{align}
where $\mathcal{Z}(\Y_N)$ is the partition function (i.e., normalisation constant) which reads, when expressing the observations in terms of the independent variables, as
\begin{align}
    \mathcal{Z}(\Y_N):=\int P_0(\bbeta) \,\exp\Big\{-\frac{1}{2\sigma^2}\|p^{-1/2}\boldsymbol{\Phi}_N(\bbeta_0-\bbeta)+\Z_N\|^2\Big\}\,d\bbeta,
\end{align}
and recall that $\bbeta_0$ is the ground truth. We interpret the posterior as a Boltzmann-Gibbs measure as in statistical mechanics and thus use the standard Gibbs bracket notation  $\langle\cdot\rangle$ to denote the expectation w.r.t. the posterior $dP(\bbeta|\Y)$ (which thus remains a function of the data $\Y_N$ even if not apparent from notation):
\begin{align}
 \langle g(\bbeta)\rangle :=\mathbb{E}[ g(\bbeta) | \Y_N]  = \int dP(\bbeta|\Y_N) g(\bbeta).
\end{align}

\subsection{Mutual information}

One of the key objects of interest is the so-called \emph{free energy} $\bar{f}_p$, which is simply the standard Shannon entropy of the data distribution $H(\Y_N)$ given the design divided by the number of samples:
\begin{equation}
   \bar{f}_p=-\frac{1}{p}\ex_{\Y_N}\ln\mathcal{Z}(\Y_N)=\frac{1}{p}H(\Y_N|\boldsymbol{\Phi}_N) -\frac{N}{2}\log 2\pi\sigma\label{free_energy_def}.
\end{equation}
It is connected to the mutual information through a simple relation
\begin{align}\label{eq:mut_inf}
    i_p:=\frac{1}{p}I(\bbeta_0;\Y_N|\boldsymbol{\Phi}_N)=\frac{1}{p}H(\Y_N|\boldsymbol{\Phi}_N)-\frac{1}{p}H(\Y_N| \bbeta_0,\boldsymbol{\Phi}_N)=\bar{f}_p-\frac{c}{2}+o(1).
\end{align}

The aim of this article is to obtain a tractable low-dimensional variational formula for the normalized mutual information $i_p$ between the observations $\Y_N$ and the signal $\bbeta$ that is asymptotically correct as $N,p\rightarrow \infty$. This will be expressed thanks to a \emph{replica symmetric potential} (this name comes from the derivation of such results based on the replica method from statistical physics \cite{mezard2009information}, see the Appendix) which in the present context reads
\begin{align}
    i_{\rm RS}(\mathbf{r}_1,\mathbf{r}_2):=\frac{c}{2\pi}\int_{0}^{\pi}\ln \Big(\frac1{\sigma^2 }\sum_{i=1}^kl_i\delta_i(\theta)r_{2,i}+1\Big)d\theta
   -\frac{1}{2}\sum_{i=1}^kl_ir_{1,i}r_{2,i}
   +\sum_{i=1}^kl_iI(\beta;\sqrt{r_{1,i}}\beta+Z), \label{replica_potential}
\end{align}
 where $\mathbf{r}_1=(r_{1,1},\ldots,r_{1,k})$ and $\mathbf{r}_2=(r_{2,1},\ldots,r_{2,k})$ are two $k$-dimensional vectors of parameters. 
  Here $\delta_i(\theta)$ is the generating function (see Appendix \ref{apx:KMS}) of matrix $\LLambda_{i,N}$  (\ref{eq:def_Lambda}) and given by
\begin{align}
        \delta_i(\theta):=\dfrac{1}{1-2\lambda_i\cos \theta+\lambda_i^2},
    \end{align}
and $I(\beta;\sqrt{r_{1,i}}\beta+Z)$ is the mutual information of the scalar channel $Y=\sqrt{r_{1,i}}\beta+Z$ with $Z\sim \mathcal{N}(0,1)$, $\beta\sim P_0$. We can now state our main theorem:
 \begin{theorem}[Replica formula for the stochastic regression model]{\label{main_thm}}
 Assume that the signal $\bbeta$ has i.i.d. entries with prior  $P_0$ with compact support and matrix $\bA_p$ is diagonal with a number of different eigenvalues $(\lambda_1,\ldots,\lambda_k)$ independent of $N,p$ (see \textbf{Case 3} or \eqref{diagonal_lambda_model}). Then we have 
  $$\lim_{p\rightarrow \infty}i_p= \inf_{\mathbf{r}_1\in[0,+\infty)^k}\sup_{\mathbf{r}_2\in[0,\rho]^k} i_{\rm RS}(\mathbf{r}_1,\mathbf{r}_2)\label{info_limit}.$$
\end{theorem}


\begin{remark}\label{rem:equiv_gamma}
    The final replica formula for the limiting normalized mutual information can also be written as
\begin{align}
    \lim_{p\rightarrow \infty}i_p=\inf_{\Gamma}i_{\rm RS}(\mathbf{r}_1,\mathbf{r}_2),
\end{align}
whenever the infimum of $\Gamma$, the set of critical points of the potential, is unique: 
\begin{align}
    \Gamma=\Big\{&(\mathbf{r}_1, \mathbf{r}_2)\in\mathbb{R}_+^k\times[0,\rho]^k \ |\ \forall\, i=1,\ldots,k: \nonumber \\
    &\qquad r_{2,i}=\text{mmse}(\beta|\sqrt{r_{1,i}}\beta+Z),\quad r_{1,i}=\frac{c}{\pi}\int_{0}^\pi\frac{\delta_i(\theta)d\theta}{\sum_{j\le k}l_jr_{2,j}\delta_j(\theta)+\sigma^2}\Big\}\label{fixed_point}.
\end{align}
Above, $\text{mmse}(\beta| f(\beta))$ is the MMSE on $\beta$ given data $f(\beta)$. The equivalence of the two formulations has been shown in \cite{barbier2019optimal} and the equivalence in our current setup follows along similar lines.
\end{remark}

\paragraph{General $\mathbf{A}_p$}
 Theorem~\ref{main_thm} can be generalized to a broader class of matrices $\mathbf{A}_p$. Let us consider the sequence of diagonal matrices $\{\mathbf{A}_p\}_{p=1}^{\infty}$  with eigenvalues $(\lambda_{1,p},\ldots,\lambda_{p,p})$ and their respective empirical eigenvalue measure $\eta_p=\frac{1}{p}\sum_{i=1}^p\delta_{\lambda_{i,p}}$. We assume the following natural hypotheses:

\begin{description}
    \item[h1)\label{Assum:supp}] Independently on $p$ there exist an interval $[a,b]\in[0,1)$ such that $\lambda_{i,p}\in[a,b]$ for all $i$ and $p$. 
    \item[h2)\label{Assum:eta}]  As $p\rightarrow\infty$ the empirical measures  $\eta_p$ converge in Wasserstein-2 distance to some measure $\eta$ defined on $[a,b]$.
\end{description}
Let us also define the sets of Lipschitz functions $\mathcal{C}_1=\{f:[a,b]\mapsto[0,C]\}$ (where $C$ is some constant that depends only on $a,b$) and $ \mathcal{C}_2=\{f:[a,b]\mapsto[0,\rho]\}$. Then, Theorem~\ref{main_thm} can be generalized as follows.
\begin{theorem}[Replica formula with general eigenvalue distribution]\label{th:gen_A}
Assume that the signal $\bbeta$ has i.i.d. entries with prior $P_0$ with compact support and matrix $\bA_p$ satisfies \textbf{h1}, \textbf{h2} above. Then we have 
    \begin{align}
        \lim_{p\rightarrow \infty}i_p= \inf_{r_1\in \mathcal{C}_1}\sup_{r_2\in \mathcal{C}_2} i_{\rm RS}[r_1,r_2],\label{mutual_info_generalAp}
    \end{align}
    where, letting $\delta(\theta,\lambda):=(1-2\lambda\cos \theta+\lambda^2)^{-1}$, the corresponding replica symmetric potential is
\begin{multline}
    i_{\rm RS}[r_1,r_2]:=\frac{c}{2\pi}\int_0^\pi\ln\Big(\frac1 {\sigma^2}\int_a^b\delta(\theta,\lambda)r_2(\lambda)d\eta(\lambda)+1\Big)d\theta-\frac{1}{2}\int_a^br_1(\lambda)r_2(\lambda)d\eta(\lambda)\\    +\int_a^bI(\beta;\sqrt{r_1(\lambda)}\beta+Z)d\eta(\lambda).\label{irs[r]}
\end{multline}
\end{theorem}
Also in this case  Remark~\ref{rem:equiv_gamma} applies, with now 
\begin{align}
    \lim_{p\rightarrow \infty}i_p=\inf_{\Gamma_\infty}i_{\rm RS}[{r}_1,{r}_2],\label{ip_critical_2}
\end{align}
where
\begin{align}
    \Gamma_\infty=\Big\{&({r}_1, {r}_2)\in\mathcal{C}_1 \times \mathcal{C}_2 \ |\ \forall\, \lambda\in[a,b]: \nonumber \\
    &\qquad r_{2}(\lambda)=\text{mmse}(\beta|\sqrt{r_{1}(\lambda)}\beta+Z),\quad r_{1}(\lambda)=\frac{c}{\pi}\int_{0}^\pi\frac{\delta(\theta,\lambda)d\theta}{\int_a^b r_{2}(x)\delta(\theta,x)d\eta(x)+\sigma^2}\Big\}\label{fixed_point_2}.
\end{align}

\paragraph{Special Cases}
In the known special cases our result \eqref{replica_potential} reduces to familiar expressions. First, we consider the standard setting $\mathbf{A}_p=\mathbf{0}$. In this case $\X_\mu$s are independent Gaussian vectors and the design matrix $\boldsymbol{\Phi}_N$ will be a matrix with i.i.d. Gaussian entries. We note that $\delta_i(\theta)=1$ and using \eqref{fixed_point} it can be easily shown that \eqref{replica_potential} and \eqref{info_limit} reduce to 

\begin{align}
   \lim_{p\rightarrow \infty}i_p&= \inf_{E\in[0,\rho]} i_{\rm RS}(E) \quad \text{with}\quad i_{\rm RS}(E)=\frac{c}{2}\ln\Big(\frac{E}{\sigma^2}+1\Big)
   -\frac{Ec}{2(E+\sigma^2)}
   +I\Big(\beta;\sqrt{\frac{c}{E+\sigma^2}}\beta+Z\Big),
\end{align}
which is the same (up to rescaling) as the equations (11-15) obtained in \cite{barbier2016mutual}.

Next we consider the case when $\mathbf{A}_p=\lambda \mathbf{I}_p$. In this case $\delta_i(\theta)=\delta(\theta)=(1-2\lambda \cos \theta+\lambda^2)^{-1}$ and using \eqref{fixed_point} it can be easily shown that \eqref{replica_potential} and \eqref{info_limit} reduce this time to 
\begin{align}
\lim_{p\rightarrow \infty}i_p&= \inf_{r \in (0,\infty)}\sup_{E\in[0,\rho]} i_{\rm RS}(r,E),\\
   \text{with} \quad i_{\rm RS}(r,E)&=\frac{c}{2\pi}\int_{0}^{\pi}\ln \Big(\frac{\delta(\theta)E}{\sigma^2}+1\Big)d\theta
   -\frac{1}{2}Er
+I(\beta;\sqrt{r}\beta+Z)\label{replica_lambda}.
\end{align}
This recovers the result in Theorem 1.1 of \cite{barbier2018mutual} (originally derived using the replica method \cite{tulino2013support}) when the first term in the right-hand side of \eqref{replica_lambda} is written in terms of the R-transform (the usual transform from random matrix theory \cite{potters2020first}) of the matrix $\boldsymbol{\Phi}_N^\intercal\boldsymbol{\Phi}_N/p$.

\subsection{Minimum mean-square errors}

An important criterion to measure the performance of an estimator is comparing it with the Bayes risk or the MMSE. In particular the estimator which minimizes the mean square error is the mean of the posterior distribution \eqref{posterior_0}, $\langle\bbeta\rangle$, and the MMSE is given by
\begin{align}
    {\rm mmse}(\bbeta | \Y_N,\boldsymbol{\Phi}_N):=\frac1p\ex\|\bbeta_0-\langle\bbeta\rangle\|^2 =\sum_{i=1}^k l_i \,{\rm mmse}^{(i)}:=\sum_{i=1}^kl_i\frac{1}{|I_i|}\ex\|\bbeta_0^{(i)}-\langle\bbeta^{(i)}\rangle\|^2,\label{blockMMSE}
\end{align}
where $\bbeta_0^{(i)} = (\beta_{0j})_{j\in I_i} \in \mathbb{R}^{|I_i|}$ is the $i$th block of $\bbeta_0$ in its natural decomposition induced by the structure \eqref{diagonal_lambda_model}, and similarly for $\bbeta^{(i)}$.
As it is not easy to access $\text{mmse}(\bbeta| \Y_N,\boldsymbol{\Phi}_N)$ directly from $i_p$, we introduce the \emph{measurement MMSE}, which is the MMSE on the responses:
\begin{align}
    \text{ymmse}:=\frac{1}{N}\ex\|\frac{1}{\sqrt{p}}\boldsymbol{\Phi}_N\bbeta_0-\langle\frac{1}{\sqrt{p}}\boldsymbol{\Phi}_N\bbeta\rangle\|^2.\
\end{align}
It is easy to find an expression for $\ymmse$ through the I-MMSE formula \cite{guo2005mutual} 
\begin{align}
    \frac{d i_p}{d \sigma^{-2}}=\frac{c_N}{2}\text{ymmse}.
\end{align}
From Theorem~\ref{main_thm} we obtain the expression for the limit of $\text{ymmse}$
\begin{align}
    \lim_{p\rightarrow\infty}\text{ymmse}=\frac{2}{c}\lim_{p\rightarrow\infty}\frac{d i_{p}}{d \sigma^{-2}}=\frac{2}{c}\frac{d i_{\rm RS}(\tilde{\mathbf{r}})}{d \sigma^{-2}}, \label{ymmse_formula}
\end{align}
where $\tilde{\mathbf{r}}=(\tilde{\mathbf{r}}_1,\tilde{\mathbf{r}}_2)$ is the global minimum of $i_{\rm RS}$ on $\Gamma$. The commutation of the large $p$ limit and the derivative used for the last equality above is justified by concavity of the mutual information in $\sigma^{-2}$. We can compute explicitly the RHS of \eqref{ymmse_formula} using Theorem~\ref{info_limit}. This proves the first formula for the measurement MMSE in the theorem below. In addition, we prove in Section~\ref{sec:mmse} a connection between the measurement MMSE and the block MMSEs defined by \eqref{blockMMSE} (see the second formula below). 
\begin{theorem}[Measurement MMSE, and connection to the block MMSEs]\label{th:ymmse_mmse}
Let $\tilde{\mathbf{r}}=(\tilde{\mathbf{r}}_1,\tilde{\mathbf{r}}_2)$ be the global minimum of $i_{\rm RS}$ on the set of critical points $\Gamma$ given by \eqref{fixed_point}. The limiting measurement MMSE is given by
\begin{align}
\lim_{p\rightarrow\infty}\text{ymmse} =\frac{\sigma^2}{\pi}\int_0^{\pi}\frac{\sum_{i=1}^{k}l_i\delta_i(\theta)\tilde{r}_{2,i}}{\sum_{i=1}^{k}l_i\delta_i(\theta)\tilde{r}_{2,i}+\sigma^2}d\theta.\label{ymmse_explicit}
\end{align}
Moreover, the measurement MMSE is related to the MMSE on the blocks of the regressor ${\rm mmse}^{(i)}:=\frac{1}{|I_i|}\ex\|\bbeta_0^{(i)}-\langle\bbeta^{(i)}\rangle\|^2$ as follows:
\begin{align}    \text{ymmse}=\frac{\sigma^2}{\pi}\int_0^{\pi}\frac{\sum_{i=1}^{k}l_i\delta_i(\theta)\text{mmse}^{(i)}}{\sum_{i=1}^{k}l_i\delta_i(\theta)\text{mmse}^{(i)}+\sigma^2}d\theta .
\end{align}
In the case of a non-discrete asymptotic spectral distribution for $\bA_p$ verifying hypotheses \textbf{h1}, \textbf{h2} above, we  have 
\begin{align*}
\lim_{p\rightarrow\infty}\text{ymmse} =\frac{\sigma^2}{\pi}\int_0^{\pi}\frac{\int_a^b\delta(\theta,\lambda)\tilde{r}_{2}(\lambda) d\eta(\lambda)}{\int_a^b\delta(\theta,\lambda)\tilde{r}_{2}(\lambda)d\eta(\lambda)+\sigma^2}d\theta.
\end{align*}
\end{theorem}

Based on this theorem, we  formulate a natural conjecture for the discrete eigenvalues profile case:
\begin{align}
&\lim_{p\rightarrow\infty}\text{mmse}^{(i)}=\tilde{r}_{2,i} \ \ \text{for all $i \in \{1,\ldots,k\}$,}\nonumber\\
&\lim_{p\rightarrow\infty} {\rm mmse}(\bbeta | \Y_N,\boldsymbol{\Phi}_N)=\sum_{i=1}^k l_i\,\tilde{r}_{2,i},\label{eq:mmse_conj}
\end{align}
where $\tilde{\mathbf{r}}=(\tilde{\mathbf{r}}_1,\tilde{\mathbf{r}}_2)$ is again the unique global minimum of $i_{\rm RS}$ on $\Gamma$. For the non-discrete case, the conjecture is as follows: 
\begin{align*}
    \lim_{p\rightarrow\infty} {\rm mmse}(\bbeta | \Y_N,\boldsymbol{\Phi}_N)=\int_a^b\,\tilde{r}_{2}(\lambda)d\eta(\lambda),
\end{align*}
where $\tilde r_2$ attains the supremum in \eqref{mutual_info_generalAp}, or equivalently corresponds to the infimum in $\Gamma_{\infty}$.

\section{Numerical experiments with VAMP\label{sec:sim}}
In this section, we provide some numerical experiments to compare the predicted measurement MMSE obtained from Theorem~\ref{th:ymmse_mmse}, along with the MMSE on the regressor conjectured to be given by \eqref{eq:mmse_conj}, to their counterparts obtained using the vector approximate message passing algorithm (VAMP) proposed in \cite{rangan2019vector} (which is closely linked to the orthogonal AMP \cite{ma2017orthogonal}). We consider various  choices of matrices $\mathbf{A}_p$.  
\begin{figure*}[!t]
\centering
\subfloat[]{\includegraphics[width=8cm]{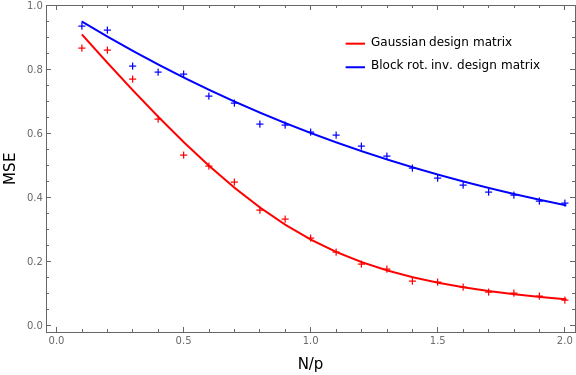}%
\label{fig:gaussian_2case}}
\hfill
\subfloat[]{\includegraphics[width=8cm]{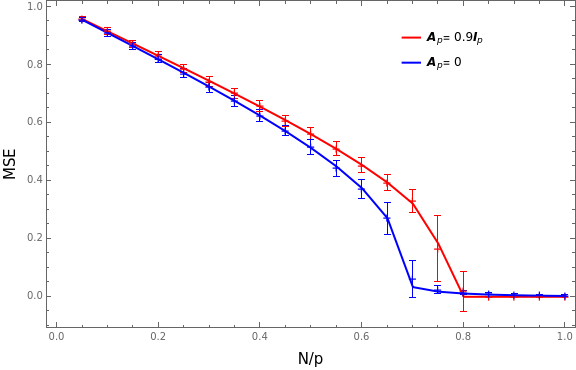}%
\label{fig:old_cases}}
    \caption{On the left MMSE versus $N/p$ for Gaussian prior. We take  $p=2100$, and $\mathbf{A}_p=0$ (red)  and $\mathbf{A}_p={\rm diag}(0.9,0.7,0.5,0.3,0.1)$ (blue). On the right we see models with $\bbeta$ drawn from Rademacher prior with $p=2100$, $\sigma^2=0.1$. Blue line represents theoretically obtained MMSE for the case with i.i.d. Gaussian design matrix ($\mathbf{A}_p=0$) and on the red - the case of  right rotation invariant design matrix ($\mathbf{A}_p=0.9\mathbf{I}_p$). Bars represents the span of MSE obtained through the VAMP algorithm for 50 instances.} 
\end{figure*}

First, let us remark that in the case of Gaussian prior, due to the rotational invariance of standard Gaussian vectors, we can replace the initial signal $\bbeta_0$ by $\mathbf{U}_N\bbeta_0$, where $\mathbf{U}_N$ is a Haar distributed random matrix. Then, the "new" design matrix $\boldsymbol{\Phi}_N^\prime=\boldsymbol{\Phi}_N\mathbf{U}_N$ is right rotationally invariant itself, which makes this case eligible for VAMP. And indeed
Figure~\ref{fig:gaussian_2case} shows two realizations of two models with Gaussian signal $\bbeta$ and equal SNR. The red curve  corresponds to the model with $\mathbf{A}_p=0$ and the blue one to the  model where $\bA_p$ has eigenvalues $\{0.9,0.7,0.5,0.1\}$ with equal multiplicity each.  We can observe that  MSE obtained with VAMP (marked with "+") matches replica prediction (continuous lines).

As was mentioned before, in \textit{Case 1} and \textit{Case 2}, that is when $\mathbf{A}_p=\lambda \mathbf{I}_p$, the design matrix $\boldsymbol{\Phi}_N$ can be decomposed as $\boldsymbol{\Phi}_N=\LLambda_N\mathbf{G}_N$, where $\LLambda_N$ is of the form (\ref{eq:def_Lambda}) for $\lambda_i=\lambda$ and $\mathbf{G}_N $ is $N\times p$ matrix with i.i.d. standard Gaussian 
 entries. The presence of $\mathbf{G}_N$ makes $\boldsymbol{\Phi}_N$ right rotationally invariant. This must imply a match between the MSE curves obtained by VAMP and the MMSE curves obtained through our results, up to finite size corrections guaranteed to vanish in the large size limit.  Figure~\ref{fig:old_cases} compares theoretically obtained MSE, represented by the continuous lines) with the MSE obtained with VAMP (bars represent span of MSE obtained from 50 realisations of the problem). The signal is drawn from Rademacher prior, noise variance is $\sigma^2=0.1$, dimension $p=2100$, and $\mathbf{A}_p=0$ (blue) and $\mathbf{A}_p=0.9 \mathbf{I}_p$ (red). In both cases we see slightly instability around the phase transition point. 
 \begin{figure*}[!t]
\centering
\subfloat[]{\includegraphics[width=8cm]{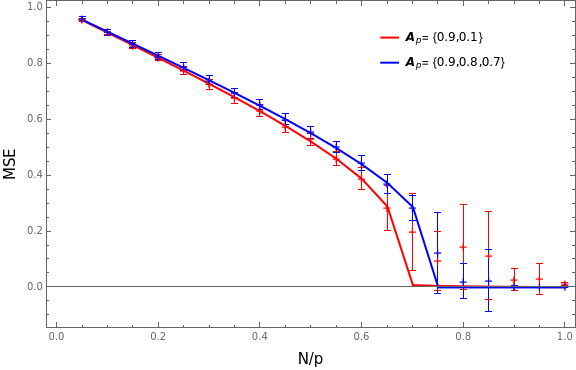}%
\label{fig:aver_big_gap}}
\hfil
\subfloat[]{\includegraphics[width=8cm]{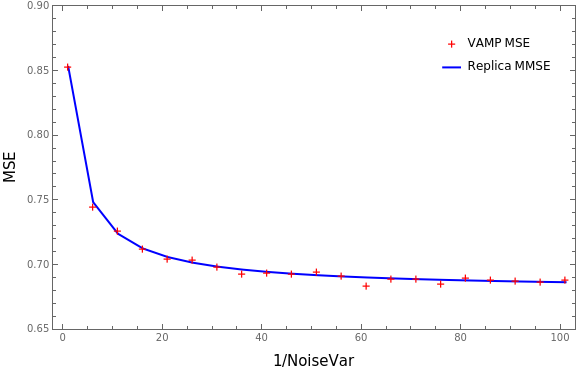}%
\label{fig:mse-snr}}
\caption{On the right MMSE versus $c_N=N/p$ for models with Rademacher prior, $\sigma^2=0.1$,  $p=2100$, and two different choices of $\bA_p$. Continuous lines represent the theoretically obtained MMSE, \eqref{eq:mmse_conj}, while the bars show the span of MSE obtained from the VAMP algorithm for 50 instances of the problem. On the left we see MMSE given by \eqref{eq:mmse_conj} versus $1/\sigma^2$ (blue), and MSE of VAMP averaged over 50 instances with signal with entries drawn from the Rademacher prior with $p=2100$ and parameters $N/p=0.3$, with $\bA_p=\{0.9,0.1\}$.}
\end{figure*}

\textit{Case 3} is more interesting. The block structure in \eqref{diagonal_lambda_model}, where
for all $i=1,\ldots,k$ the matrix $\mathbf{G}_N^{(i)}$ is $N\times |I_i|$ with i.i.d. standard Gaussian elements, implies that despite each block is individually right rotationally invariant, the overall design is not. Consequently the state evolution analysis of \cite{rangan2019vector} does not hold so we cannot guarantee a good performance of VAMP, nor being able to track its performance. 

 Figure~\ref{fig:aver_big_gap} depicts two such cases. For both models, we took the signal $\bbeta_0$ of dimension $p=2100$ with binary $\pm 1$ elements drawn from a Rademacher distribution, and noise $\Z$  taken with variance $\sigma^2=0.1$. The red curves correspond to the model with $\bA_p=\{0.9,0.1\}$, i.e., diagonal matrix which consists of $0.9$ and $0.1$ eigenvalues with equal multiplicity, while for the blue curves, $\bA_p$ is again a diagonal matrix with eigenvalues $\{0.9,0.8,0.7\}$ of equal multiplicity. The continuous lines correspond to theoretical results. As we can see, the MMSE curves display a phase transition. Apart from a region rather close to the transition, VAMP seems essentially optimal for both models despite the lack of rotational invariance. However, around the transition point we observe that VAMP is highly unstable. Furthermore, the greater the dispersion of the eigenvalues of
$\bA_p$, the more unstable the algorithm becomes. This is likely due to the fact that, naturally, when the eigenvalues of $\bA_p$ are grouped together, the resulting design matrix is "closer" to being right-rotationally invariant.

Figure~\ref{fig:mse-snr} shows the MMSE versus the inverse $\sigma^{-2}$ of the noise variance with ratio $c_N=0.3$, far from the transition. Remarkably, for $\bA_p=\{0.9,0.1\}$ the MSE obtained from VAMP averaged over 50 instances matches almost perfectly the theoretical curve, suggesting again a strong robustness to the rotational-invariance hypothesis, and demonstrating its potential applicability in the context of times series.

In Figure~\ref{fig:ymmse} we compare the measurement MMSE (YMMSE) obtained theoretically from \eqref{ymmse_explicit} and the one calculated from the VAMP estimator of $\bbeta_0$. In this case, due to the matrix multiplication with the design, the important estimation error on the regressor close to the transition point in the first figure is amplified even more. But, again, away from it the match with the MMSE theory is excellent. 
\begin{figure*}[!t]
\centering
    \subfloat[]{\includegraphics[width=8cm]{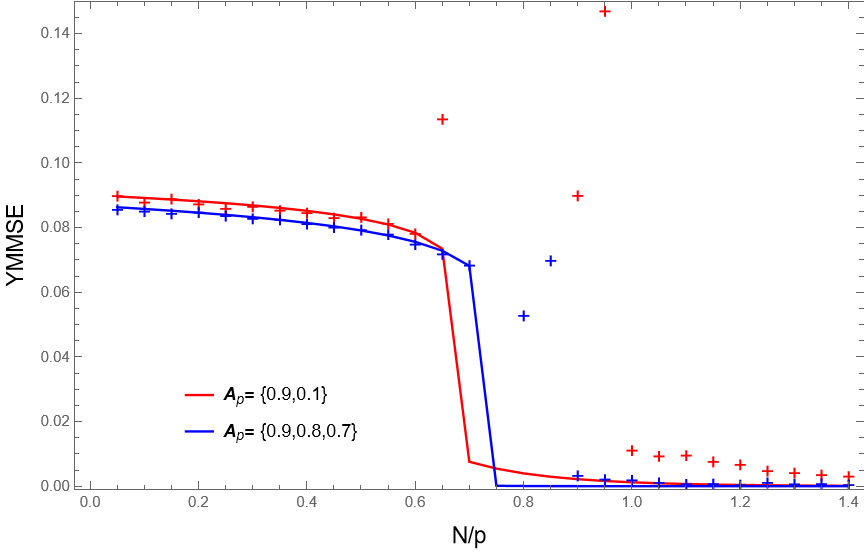}
  \label{fig:ymmsse_new}}
\hfill
       \subfloat[]{ \includegraphics[width=8cm]{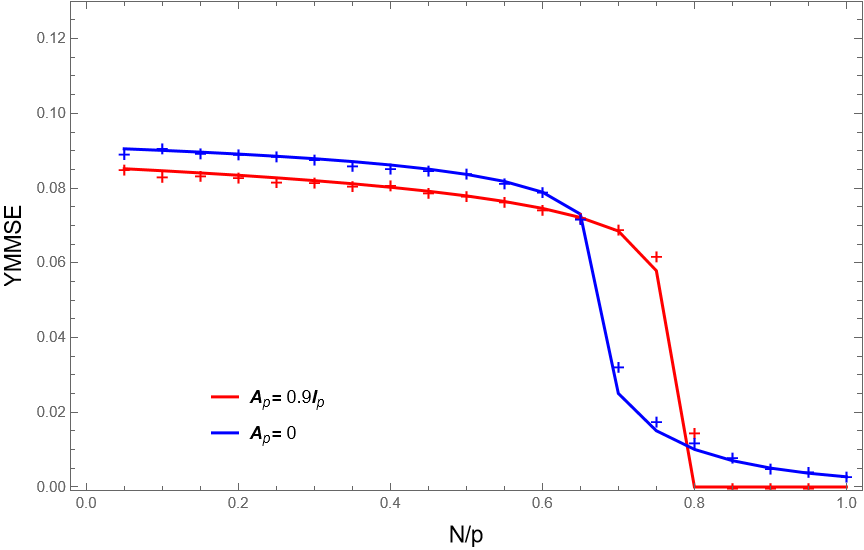}
\label{fig:ymmse_old}}
    \caption{YMMSE versus $c_N=N/p$ for models with Rademacher prior, $\sigma^2=0.1$,  $p=2100$, and  different choices of $\bA_p$. Continuous lines represent the theoretically obtained YMMSE, while "+" show the YMMSE averaged over 50 instances obtained from the VAMP algorithm. The plot on the right shows two models with right rotationally invariant design matrices ($\bA_p=0$ and $\bA_p=0.9 \mathbf{I}_p$ and on the left we see two models with block right rotationally invariant design matrices ($\bA_p=\{0.9,0.1\}$ and $\bA_p=\{0.9,0.8,0.7\}$).}
  \label{fig:ymmse}
\end{figure*}

\section{Proofs}\label{sec:proofs}
\subsection{Proof of the replica formula by the adaptive interpolation method\label{int_prf}}

In this section we prove Theorem \ref{main_thm} using the adaptive interpolation method \cite{barbier2019adaptive,barbier2019optimal}. First we set up some notations.

From now on, we drop some subscripts indicating the vectors and matrix sizes. Recall that our linear model \eqref{eq_SRM}  with $\X_\mu$ given by \eqref{eq:def_x_t} can be reformulated in the form \eqref{diagonal_lambda_model}. For brevity we write   $\LLambda^{1/2}_{i,N}$ as $\LLambda^{(i)}$ and $\bbeta_0^{(i)}$ ($\bbeta^{(i)}$) as the   block  of length $|I_i|$ of vector $\bbeta_0$ ($\bbeta$) with  elements on positions $I_i$, each block of entries being non-overlapping with the others.
\begin{align}
    Y_{\mu}=\sqrt{\frac{1}{p}}(\boldsymbol{\Phi} \bbeta_0)_{\mu}+Z_\mu=\sqrt{\frac{1}{p}}\sum_{i=1}^k(\LLambda^{(i)} \mathbf{G}^{(i)} \bbeta_0^{(i)})_{\mu}+Z_\mu.
\end{align}
   We fix a sequence $s_p\in (0,1/2] $ such that $s_p\rightarrow 0$ as $p$ goes to  infinity. Let $\boldsymbol{\epsilon}=(\epsilon_{1,1},\ldots,\epsilon_{1,k}, \epsilon_{2,1},\ldots,\epsilon_{2,k})\in [s_p,2s_p]^{2k}$, and   $r_{1,i}: [0, 1] \rightarrow \mathbb{R}$ and $r_{2,i} : [0, 1] \rightarrow [0,\infty)$ be $2k$
continuous ``interpolation functions'' (that will later depend on $\boldsymbol{\epsilon}$). We define 
\begin{align}
 R_{1,i}(t):=\epsilon_{1,i}+\int_0^t r_{1,i}(u)du, \qquad R_{2,i}(t):=\epsilon_{2,i}+\int_0^t r_{2,i}(u)du
\end{align}
for $t\in[0,1]$. We will denote $\mathbf{R}_1(t)=(R_{1,1}(t),\ldots,R_{1,k}(t))$ and $\mathbf{R}_2(t)=(R_{2,1}(t),\ldots,R_{2,k}(t))$.
Consider the following $t$-dependent observation
channels:
\begin{align}
    \Y_{t}&=\sqrt{\frac{1-t}{p}}\boldsymbol{\Phi} \bbeta_0+\sum_{i=1}^{k}\sqrt{R_{2,i}(t)l_i}\LLambda^{(i)} \V^{(i)}+\Z, \\
    \tilde{\Y}_{t}^{(i)}&=\sqrt{R_{1,i}(t)}\bbeta_{0}^{(i)}+\tilde{\Z}^{(i)},\; i=1,\ldots,k.\label{scalar_int}
\end{align}
Here $\V^{(i)}\sim \mathcal{N}(0,\mathbf{I}_N)$, the vector channel $\tilde{\Y}_{t}^{(i)}$ is of dimension $|I_i|$, and $\tilde{\Z}^{(i)}\sim\mathcal{N}(0,\mathbf{I}_{|I_i|})$. We also define $N$ dimensional vector $\tilde{\Y}_{t}=(\tilde{\Y}_{t}^{(1)\intercal},\ldots,\tilde{\Y}_{t}^{(k)\intercal})^\intercal$. The surrogate inference problem defined above for the sake of the proof is to jointly recover $\bbeta_0$ and $\V^{(i)}$ from the observations $(\Y_{t},\tilde{\Y}_{t})$ and given $\boldsymbol{\Phi}$. In the Bayesian setting the posterior distribution of $\bbeta$ and $\v$ is given in Gibbs-Boltzmann form
\begin{align}\label{eq:post_inter}
    dP(\bbeta,\v|\Y_t,\tilde{\Y}_t,\boldsymbol{\Phi})=\frac{dP_0(\bbeta)D\v e^{-\mathcal{H}(t,\bbeta,\v,\Y_t,\tilde{\Y}_t,\boldsymbol{\Phi})}}{\int dP_0(\bbeta)D\v e^{-\mathcal{H}(t,\bbeta,\v,\Y_t,\tilde{\Y}_t,\boldsymbol{\Phi})}},
\end{align}
where the Hamiltonian $ \mathcal{H}=\mathcal{H}(t,\bbeta,\v^{(i)},\Y_t,\tilde{\Y}_t^{(i)},\boldsymbol{\Phi})$ (i.e., negative of the log-likelihood and ignoring the constants) and associated partition function $\mathcal{Z}_t(\Y_t,\tilde{\Y}_t)$
are given by
\begin{align}
    \mathcal{H}&:=\frac{1}{2\sigma^2}\|\Y_t-\sqrt{\frac{1-t}{p}}\boldsymbol{\Phi} \bbeta-\sum_{i=1}^k\sqrt{R_{2,i}(t)l_i}\LLambda^{(i)}\v^{(i)}\|^2+\frac{1}{2}\sum_{i=1}^k\|\tilde{\Y}^{(i)}_t-\sqrt{R_{1,i}(t)}\bbeta^{(i)}\|^2,\label{interpol_hamiltonian}\\
    \mathcal{Z}_t(\Y_t,\tilde{\Y}_t)&:=\int dP_0(\bbeta)D\v\,\exp\big\{-\mathcal{H}(t,\bbeta,\v,\Y_t,\tilde{\Y}_t,\boldsymbol{\Phi})\big\}.\label{eq:part_inter}
\end{align}
The normalized mutual information for the interpolating model  is given by
\begin{equation}
    i_{p,\epsilon}(t):=-\frac{1}{p}\ex_{\boldsymbol{\Phi},\V,\bbeta_0}\ln \int dP_0(\bbeta)D\v\, \exp\big\{-\mathcal{H}(t,\bbeta,\v,\Y_t,\tilde{\Y}_t,\boldsymbol{\Phi})\big\}-\frac{c+1}{2}. \label{interpolating_info}
\end{equation}
First note that $i_p=i_{p,\epsilon}(0)+o(1)$. 
Indeed, denote the interpolating free energy, i.e. the first term on the RHS of \eqref{interpolating_info}) by $f_{p,\epsilon}(t)$. Define the Hamiltonian of the original problem by
$$\mathcal{H}(\bbeta,\Y,\boldsymbol{\Phi}):=\frac{1}{2\sigma^2}\|\Y-\sqrt{\frac{1}{p}}\boldsymbol{\Phi} \bbeta\|^2.$$
We note that $f_{p,\epsilon}(0)$ converges to $f_{p,0}(0)$ (see proof of Lemma 1 of \cite{barbier2019adaptive} for example), where
\begin{align*}
f_{p,0}(0)&=-\frac{1}{p}\ex_{\boldsymbol{\Phi},\bbeta_0}\ln \int dP_0(\bbeta)\exp\big\{-\mathcal{H}(\bbeta,\Y,\boldsymbol{\Phi})\big\}+1/2
\end{align*}
We can identify the first term on the RHS as $\bar{f}_p$ (appearing in \eqref{free_energy_def}).
It thus follows that
$$i_{p,\epsilon}(0)=f_{p,0}(0)-\frac{c+1}{2}+o(1)=\bar{f}_p-\frac{c}{2}=i_p+o(1).$$
Next, the goal is to compute $i_{p,\epsilon}(0)$ using the following identity 
\begin{align}
i_{p,\epsilon}(0)=i_{p,\epsilon}(1)-\int_{0}^{1}i'_{p,\epsilon}(t)dt.    
\end{align}
We will obtain an asymptotic expression for the right-hand side of the above equation and then, letting $N\to\infty$, obtain the single-letter formula for the mutual information. This formula will not depend on $\boldsymbol{\epsilon}$, since $\boldsymbol{\epsilon}$ lies in the hyper-rectangle [$s_p,2s_p]^{2k}$, whose volume goes to $0$ as $N,p\rightarrow \infty$ (we give rigorous arguments later).

Henceforth, we will denote $i_{p,\epsilon}(1)$ by $i_p(1)$ suppressing the dependence on $\boldsymbol{\epsilon}$.
Considering the form of the Hamiltonian in \eqref{interpol_hamiltonian}, it can be easily verified that $i_p(1)$ can be written in the form
\begin{align*}
    i_p(1)&=\sum_{i=1}^kl_iI(\bbeta_{0}^{(i)};\sqrt{R_{1,i}(1)}\bbeta_{0}^{(i)}+\tilde{\Z}^{(i)})+I(\V;\sum_{i=1}^{k}\sqrt{R_{2,i}(1)l_i}\LLambda^{(i)} \V^{(i)}+\Z)\\
    &:=\sum_{i=1}^kl_iI_1(R_{1,i}(1))+I_2(\mathbf{R}_2(1))
\end{align*}
Further, analogous to equation (8) of \cite{barbier2018mutual} $I_2(\mathbf{R}_2(1))$ is given by $\frac{c}{2}\ex_{X'}\log(1+X'/\sigma^2)$  where $X'$ is drawn from the empirical spectral distribution of the matrix $L^\intercal L$ where
$$L:=(\sqrt{R_{2,1}(1)l_1}\LLambda^{(1)},\ldots,\sqrt{R_{2,k}(1)l_k}\LLambda^{(k)}).$$
Indeed, $I_2(\mathbf{R}_2(1))$ is the normalized mutual information $\frac{1}{p}I(\mathbf{W},\boldeta|L)$ where
$\mathbf{W}=L\boldeta+\bar{\mathbf{Z}}$ for centred Gaussian noise $\bar{\mathbf{Z}}$ with variance $\sigma^2$. On the other hand $I(\mathbf{W},\boldeta|L)=\frac{1}{2}\log \det(L^\intercal L/\sigma^2+I)$ by the celebrated log-det formula for mutual information (see \cite{tulino2004random}) and expressing the same formula in terms of the spectrum of $L^\intercal L$ gives the desired result.
Therefore we only need to analyze the spectrum of $L^\intercal L$, or, equivalently, of $L L^\intercal$
$$L L^\intercal =\sum_{i=1}^k R_{2,i}(1)l_i\LLambda^{(i)} \LLambda^{{(i)}\intercal}=\sum_{i=1}^k l_iR_{2,i}(1)\LLambda_{i,N}.$$
 We can show that (cf. derivation of \eqref{eigen_KMS} from \eqref{eigen_KMS_mid} in Appendix \ref{rep_cal})
\begin{align}\label{eq:I2}
    I_2(\mathbf{R}_2(1))=\frac{c}{2}\ex_{X'}\log(1+X'/\sigma^2)=\frac{c}{2\pi}\int_{0}^{\pi}\ln \Big(\frac1{\sigma^2}\sum_{i=1}^kl_i\delta_i(\theta)R_{2,i}(1)+1\Big)d\theta.
\end{align}
With this we conclude that $I_2(\mathbf{R}_2(1))$ is indeed a function of $\{R_{2,i}(1)\}_{i=1}^k$, or more explicitly, of the $\{r_{2,i}(\cdot)\}_{i=1}^k$.

We now have defined all the necessary set-up for the adaptive interpolation proof. We are going to decompose the proof into steps, starting with the  evaluation of the time derivative $i'_{p,\epsilon}(t)$. In order to do so, define the overlaps 
\begin{align}
Q_i:=\frac{1}{|I_i|}\sum_{j=1}^{|I_i|}\beta_{0,j}^{(i)}\beta^{(i)}_j  \label{overlap_1} 
\end{align}
and an $N$-dimensional vector $\mathbf{u}_t$ with elements
\begin{equation}{\label{utnu}}
    u_{t,\nu}:=\sqrt{\frac{1-t}{p}}\Big(\boldsymbol{\Phi}(\bbeta_0-\bbeta)\Big)_{\nu}+\sum_{i=1}^k\sqrt{R_{2,i}(t)l_i}\Big(\LLambda^{(i)}(\V^{(i)}-\v^{(i)})\Big)_{\nu}+Z_{\nu}.
\end{equation}

\paragraph{Step 1} First, we obtain an asymptotic expression for the derivative of the mutual information (in terms of the interpolating parameters $r_{1,i}(t)$, $r_{2,i}(t)$ and the overlaps $Q_i$):
 
\begin{lemma}[Time derivative]{\label{der_lemma}} We have
\begin{equation}
  i'_{p,\epsilon}(t)=\frac{1}{2}\sum_{i=1}^{k}l_ir_{1,i}(t)(\rho-\ex[\langle Q_i\rangle_{t,\epsilon}])
  +\frac{1}{2\sigma^4p}\sum_{i=1}^{k}l_i\ex[\langle (r_{2,i}(t)-(\rho-Q_i))\Z^\intercal(\LLambda^{(i)})^2 \mathbf{u}_{t}\rangle_{t,\epsilon}]+o(1) \label{info_deriv}
\end{equation}
where $\langle.\rangle_{t,\epsilon}$ is the expectation with respect to the posterior \eqref{eq:post_inter} of the interpolating model.
\end{lemma}










\begin{proof}
First, we note that
$$i'_{p,\epsilon}(t)=\frac{1}{p}\ex[ \mathcal{H}'(t,\bbeta_0,\V,\Y_t,\tilde{\Y}_t,\boldsymbol{\Phi})\ln \mathcal{Z}]+\frac{1}{p}\ex[\langle \mathcal{H}'(t,\bbeta,\v,\Y_t,\tilde{\Y}_t,\boldsymbol{\Phi})\rangle_{t,\epsilon}]$$
where $\mathcal{Z}$ is partition function defined as (\ref{eq:part_inter}) and $\ex$ the joint expectation w.r.t. all
quenched variables (i.e. fixed by the realization of the
problem), namely $(\bbeta_0, \V, \Y_t, \tilde{\Y}_t,\boldsymbol{\Phi})$, or equivalently w.r.t.
$(\bbeta_0, \V, \Z, \tilde{\Z},\boldsymbol{\Phi})$. Here $\mathcal{H}'$ means a time-derivative, given by
\begin{align}
  \mathcal{H}'(t,\bbeta_0,\V,\Y_t,\tilde{\Y}_t,\boldsymbol{\Phi})=\frac{1}{2\sigma^2}\sum_{\mu=1}^NZ_{\mu}\Big(\sqrt{\frac{1}{p(1-t)}}(\boldsymbol{\Phi} \bbeta_0)_{\mu}-\sum_{i=1}^k\frac{r_{2,i}(t)\sqrt{l_i}}{\sqrt{R_{2,i}(t)}}(\LLambda^{(i)}\V^{(i)})_{\mu}\Big)\nonumber\\
  -\frac{1}{2}\sum_{i=1}^{k}\sum_{j=1}^{|I_i|}\tilde{Z}^{(i)}_j\frac{r_{1,i}(t)}{\sqrt{R_{1,i}(t)}}\bbeta^{(i)}_{0,j}.  \label{eq:der_gam}
\end{align}
 Nishimori's identity then yields
$$\ex[\langle \mathcal{H}'(t,\bbeta,\v,\Y_t,\tilde{\Y}_t,\boldsymbol{\Phi})\rangle_{t,\epsilon}]]=\ex[\mathcal{H}'(t,\bbeta_0,\V,\Y_t,\tilde{\Y}_t,\boldsymbol{\Phi})]=0$$
and hence we obtain,
$$i'_{p,\epsilon}(t)=\frac{1}{p}\ex[ \mathcal{H}'(t,\bbeta_0,\V,\Y_t,\tilde{\Y}_t,\boldsymbol{\Phi})\ln \mathcal{Z}].$$
We plug \eqref{eq:der_gam} in the last equation and transform each of the three terms by Gaussian integration by parts. The last term will be integrated by parts  with respect to $\Tilde{Z}_j^{(i)}$
\begin{align*}
&\ex \Big[-\frac{1}{2p}\sum_{i=1}^{k}\sum_{j=1}^{|I_i|}\tilde{Z}^{(i)}_j\frac{r_{1,i}(t)}{\sqrt{R_{1,i}(t)}}\beta^{(i)}_{0,j}\ln \mathcal{Z}\Big]
= \frac{1}{2p}\sum_{i=1}^{k}\sum_{j=1}^{|I_i|}\frac{r_{1,i}(t)}{\sqrt{R_{1,i}(t)}}\ex [\beta^{(i)}_{0,j}\langle\sqrt{R_{1,i}(t)}(\beta^{(i)}_{0,j}-\beta^{(i)}_j)+\tilde{Z}^{(i)}_j\rangle_{t,\epsilon}]\\
&=\frac{1}{2p}\sum_{i=1}^{k}\sum_{j=1}^{|I_i|}r_{1,i}(t)\rho-\frac{1}{2p}\sum_{i=1}^{k}\sum_{j=1}^{|I_i|}r_{1,i}(t)\ex[\langle \beta^{(i)}_{0,j}\beta^{(i)}_j\rangle_{t,\epsilon}]
=\frac{1}{2}\sum_{i=1}^{k}l_ir_{1,i}(t)(\rho-\ex[\langle Q_i\rangle_{t,\epsilon}]).
\end{align*}
Similarly, using Gaussian integration by parts for the second term, this time w.r.t. $V^{(i)}_j$ we obtain 
\begin{align*}&\ex \Big[-\frac{1}{2\sigma^2p}\sum_{\mu=1}^{N}Z_{\mu}\Big(\sum_{i=1}^k\frac{r_{2,i}(t)\sqrt{l_i}}{\sqrt{R_{2,i}(t)}}(\LLambda^{(i)}\V^{(i)})_{\mu}\Big)\ln \mathcal{Z}\Big]\\
&=\sum_{i=1}^{k}\sum_{j,\mu,\nu=1}^{N} \frac{l_i}{2\sigma^4p}\ex \Big[Z_{\mu}r_{2,i}(t)\Lambda^{(i)}_{\mu j}\Lambda^{(i)}_{\nu j}\Big\langle\sqrt{\frac{1-t}{p}}\Big(\boldsymbol{\Phi}(\bbeta_0-\bbeta)\Big)_{\nu}\\
&+\sum_{m=1}^k\sqrt{R_{2,m}(t)l_m}\Big(\LLambda^{(m)}(\V^{(m)}-\v^{(m)})\Big)_{\nu}+Z_{\nu}\Big\rangle_{t,\epsilon}\Big]
=\frac{1}{2\sigma^4p}\sum_{i=1}^{k}\sum_{\nu=1}^N r_{2,i}(t)l_i\ex [(\Z^\intercal\LLambda_{i,N})_{\nu}\langle u_{t,\nu}\rangle_{t,\epsilon}]\\
&=\frac{1}{2\sigma^4p}\sum_{i=1}^{k} r_{2,i}(t)l_i\ex [\langle\Z^\intercal\LLambda_{i,N} \mathbf{u}_{t}\rangle_{t,\epsilon}].
\end{align*}

Finally, another use of Gaussian integration by parts of the first term with respect to  $G^{(i)}_{j\ell}$ yields
\begin{align*}
&\ex \Big[\frac{1}{2\sigma^2p}\sum_{\mu=1}^{N}Z_{\mu}\sqrt{\frac{1}{p(1-t)}}(\boldsymbol{\Phi} \bbeta_0)_{\mu}\ln \mathcal{Z}\Big]
=\frac{1}{2\sigma^2p}\ex \Big[\sum_{\mu=1}^{N}Z_{\mu}\sqrt{\frac{1}{p(1-t)}}\sum_{i=1}^k(\LLambda^{(i)} \mathbf{G}^{(i)} \bbeta_0^{(i)})_{\mu}\ln \mathcal{Z}\Big]\\
&=-\frac{1}{2\sigma^4p^2}\ex \Big[\sum_{i=1}^k\sum_{\mu,\nu,j=1}^{N}\sum_{\ell=1}^{|I_i|}Z_{\mu}\Lambda^{(i)}_{\mu j}\Lambda^{(i)}_{\nu j} \beta^{(i)}_{0,\ell}\Big\langle(\beta^{(i)}_{0,\ell}-\beta^{(i)}_\ell)\Big(\sqrt{\frac{1-t}{p}}\Big(\boldsymbol{\Phi}(\bbeta_0-\bbeta)\Big)_{\nu}\\
&+\sum_{i=m}^k\sqrt{R_{2,m}(t)l_m}\Big(\LLambda^{(m)}(\V^{(m)}-\v^{(m)})\Big)_{\nu}+Z_{\nu}\Big)\Big\rangle_{t,\epsilon}\Big]\\
&=-\frac{1}{2\sigma^4p^2}\ex \Big[\sum_{\nu=1}^{N}\sum_{i=1}^k\sum_{\ell=1}^{|I_i|}(\Z^\intercal\LLambda_{i,N})_{\nu} \langle \beta^{(i)}_{0,\ell}(\beta^{(i)}_{0,\ell}-\beta^{(i)}_\ell)u_{t,\nu}\rangle_{t,\epsilon}\Big]\\
&=-\frac{1}{2\sigma^4p}\ex \Big[\sum_{\nu=1}^{N}\sum_{i=1}^kl_i(\Z^\intercal\LLambda_{i,N})_{\nu} \Big\langle \Big(\frac{1}{|I_i|}\sum_{\ell=1}^{|I_i|}(\beta^{(i)}_{0,\ell})^2-\frac{1}{|I_i|}\sum_{\ell=1}^{|I_i|}\beta^{(i)}_{0,\ell}\beta^{(i)}_\ell\Big)u_{t,\nu}\Big\rangle_{t,\epsilon}\Big]\\
&=-\frac{1}{2\sigma^4p}\ex \Big[\sum_{i=1}^kl_i \langle (\hat{\rho}_i-Q_i)\Z^\intercal\LLambda_{i,N}\mathbf{u}_{t}\rangle_{t,\epsilon}\Big],
\end{align*}
where we define $\hat{\rho}_i$ as a normalised norm of block $\bbeta^{(i)}_{0}$, i. e. $\hat{\rho}_i=\frac{1}{|I_i|}\|\bbeta^{(i)}_{0}\|^2$.
Bringing together all above and adding-subtracting term $(2\sigma^4p)^{-1}\sum_{i=1}^k\ex [l_i\rho \langle\Z^\intercal\LLambda_{i,N}  \mathbf{u}_{t}\rangle_{t,\epsilon}]$,  we have
\begin{align*}
i'_{p,\epsilon}(t)&=\frac{1}{2}\sum_{i=1}^{k}l_ir_{1,i}(t)(\rho-\ex[\langle Q_i\rangle_{t,\epsilon}])+\frac{1}{2\sigma^4p}\sum_{i=1}^{k}l_i\ex[\langle(r_{2,i}(t)-(\rho-Q_i))\Z^\intercal\LLambda_{i,N} \mathbf{u}_{t}\rangle_{t,\epsilon}]\\
& -\frac{1}{2\sigma^4p}\sum_{i=1}^k\ex [l_i(\hat{\rho}_i-\rho)\Z^\intercal\LLambda_{i,N} \langle \mathbf{u}_{t}\rangle_{t,\epsilon}].
\end{align*}
To conclude the proof we show that the third term is $o(1)$. In particular, it suffices to show that 
$$\frac{1}{2\sigma^4p}\ex [(\hat{\rho}_i-\rho)\Z^\intercal\LLambda_{i,N} \langle \mathbf{u}_{t}\rangle_{t,\epsilon}]=o(1).$$
Separating factors with Cauchy-Schwarz and using  Markov inequality, we see immediately that $\ex(\hat{\rho}_i-\rho)^2=O(1/|I_i|)$. Thus, it is enough to show that  $\ex [\langle(2p)^{-2}(\Z^\intercal\LLambda_{i,N}  \mathbf{u}_{t})^2\rangle_{t,\epsilon}]$ is bounded. Note that
$$\ex [\langle(\Z^\intercal\LLambda_{i,N} \mathbf{u}_{t})^2\rangle_{t,\epsilon}]\leq (\ex \|\Z\|^4)^{1/2} \|\LLambda_{i,N}\|^2 (\ex[\langle \|\mathbf{u}_{t}\|^4\rangle_{t,\epsilon}])^{1/2}.$$
Since $(\ex \|\Z\|^4)^{1/2}=O(p)$ and  $\|\LLambda_{i,N}\|^2=O(1)$, we only need to check that $(\ex[\langle \|\mathbf{u}_{t}\|^4\rangle_{t,\epsilon}])^{1/2}=O(p)$.
Recall the definition of $u_{t,\nu}$ from \eqref{utnu}.  We have 
\begin{align*}\ex[\langle \|\mathbf{u}_{t}\|^4\rangle_{t,\epsilon}]&\leq C\ex\Big(\frac{1}{p^2}\Big(\|\boldsymbol{\Phi} \bbeta_0\|^4+\langle\|\boldsymbol{\Phi} \bbeta\|^4\rangle_{t,\epsilon}\Big)+\sum_{i=1}^kR^{2}_{2,i}(t)l^2_i(\|\LLambda^{(i)}\V^{(i)}\|^4+\langle\|\LLambda^{(i)}\v^{(i)}\|^4\rangle_{t,\epsilon})+\|\Z\|^4\Big)\\
&\leq C\ex\Big(\frac{1}{p^2}\|\boldsymbol{\Phi} \bbeta_0\|^4+\sum_{i=1}^kR^{2}_{2,i}(t)l_i\|\LLambda^{(i)}\V^{(i)}\|^4+\|\Z\|^4\Big)
\end{align*}
where the last inequality follows from the Nishimori identity.
One can also show by direct computation that $\ex[\|\boldsymbol{\Phi}\|_F^4]=O(p^2)$. Finally, observing that $\ex[\|\bbeta_0\|^4],\ex[\|\Z\|^4]$ are $O(p^2)$ and $\ex[\|\V^{(i)}\|^4]$ is $O(p^4)$,  we obtain $(\ex[\langle \|\mathbf{u}_{t}\|^4\rangle_{t,\epsilon}])^{1/2}=O(p)$ which was required to be shown.

 \end{proof}

\paragraph{Step 2}:
Next we show that the second term in the right-hand side of \eqref{info_deriv} is $o(1)$ by choosing  $r_{2,i}(t)=\rho-\ex[\langle Q_i\rangle_{t,\epsilon}]$ and then showing that $Q_i$ concentrates around $\ex[\langle Q_i\rangle_{t,\epsilon}]$ in a suitable sense. In particular we show the overlap concentration (concentration of $Q_i$) in the following lemma:

\begin{lemma}[Overlap concentration]\label{lemma:overlap_concentr}
    Assume that for any $t\in(0,1)$ the map $\boldsymbol{\epsilon}=(\epsilon_{1,1},\ldots,\epsilon_{1,k}, \epsilon_{2,1},\ldots$
    $\ldots\epsilon_{2,k})\in[s_p,2s_p]^{2k}\mapsto \mathbf{R}(t,\epsilon)=(R_{1,i}(t,\epsilon),\ldots,R_{2,k}(t,\epsilon))$ is a $\mathcal{C}^1$ diffeomorphism with Jacobian determinant greater or equal to 1. Then we can find sequence $s_p$ going to 0 slowly enough such that there exists constants $C$ and $\gamma$ that depend only on support and moments of prior distribution $P_0$ and $c$ for which:
    \begin{align}
        \frac{1}{s_p^{2k}}\int_{[s_p,2s_p]^{2k}}d\epsilon\int_0^1dt\ex\langle(Q_i-\ex\langle Q_i\rangle_{t,\epsilon})^2\rangle_{t,\epsilon}\leq Cp^{-\gamma},\quad\text{for }i=1,\ldots, k.
    \end{align}
\end{lemma}
    \begin{proof}
        The proof follows in the same vein as the concentration of overlaps proved \cite{barbier2019adaptive,barbier2019optimal}, hence we outline the steps and omit the details. Indeed it is enough to show that for all $\boldsymbol{\epsilon}$

    \begin{align}
        \ex\langle(Q_i-\ex\langle Q_i\rangle_{t,\epsilon})^2\rangle_{t,\epsilon}\leq Cp^{-\gamma},\quad\text{for }i=1,\ldots, k. \label{overlap_conc_interpol}
    \end{align}
   Define $\mathcal{L}_i=\frac{1}{|I_i|}\sum_{j=1}^{|I_i|}\Big(\frac{(\beta^{(i)}_{j})^2}{2}-\beta^{(i)}_j\beta^{(i)}_{0,j}-\frac{\beta^{(i)}_{j}\tilde{Z}^{(i)}_j}{2\sqrt{R_{1,i}}}\Big)$. Using the fact  that the free energy concentrates (see Proposition~\ref{prop:concent_free_energy}), i.e.
\begin{equation}\ex\Big[\Big|\frac{1}{p}\log \mathcal{Z}-\ex\Big[\frac{1}{p}\log \mathcal{Z}\Big]\Big|^2\Big]\leq Cp^{-1}\label{conc_free_energy},\end{equation}
 one can show that the following inequalities hold:
\begin{equation}
    \ex\langle(Q_i-\ex\langle Q_i\rangle_{t,\epsilon})^2\rangle_{t,\epsilon}\leq 4 \ex\langle(\mathcal{L}_i-\ex\langle \mathcal{L}_i\rangle_{t,\epsilon})^2\rangle_{t,\epsilon}\label{overlap_bound}
    \end{equation}

\begin{equation}
     \ex[\langle(\mathcal{L}_i-\langle \mathcal{L}_i\rangle_{t,\epsilon})^2\rangle_{t,\epsilon}]\leq Cp^{-1}\label{ell1_bound}
    \end{equation}

\begin{equation}
    \ex[(\langle\mathcal{L}_i\rangle_{t,\epsilon}-\ex\langle \mathcal{L}_i\rangle_{t,\epsilon})^2] \leq Cp^{-\gamma}\label{ell2_bound}
    \end{equation}

The first inequality follows from the proof of equation (111) of \cite{barbier2019adaptive} and the last two inequalities follow from the proof of lemma 29 and 30 from \cite{barbier2019optimal}, respectively. We note that the same proof works because of the separability of the Hamiltonian in $R_{1,i}(t)$ and $\tilde{\Z}^{(i)}$. In particular, since $R_{1,i}(t)$ and $\tilde{\Z}^{(i)}$ only appears in the Hamiltonian through the term $\|\tilde{\Y}^{(i)}_t-\sqrt{R_{1,i}(t)}\bbeta^{(i)}\|^2=\|\tilde{\Z}^{(i)}+\sqrt{R_{1,i}(t)}(\bbeta^{(i)}_0-\bbeta^{(i)})\|^2$, the  differentiation with respect to $R_{1,i}(t)$ and Gaussian integration by parts with respect to $\tilde{Z}^{(i)}_j$ proceeds along the same lines as in the original proofs given in  \cite{barbier2019adaptive,barbier2019optimal}.


\end{proof}

\paragraph{Step 3}: From Lemmas \ref{der_lemma} and \ref{lemma:overlap_concentr} we obtain the following formula for normalized mutual information in terms of the interpolating functions $r_{1,i}(\cdot), r_{2,i}(\cdot)$.

\begin{lemma}[Fundamental identity]\label{lem:fund_iden}
    We take $\boldsymbol{\epsilon}=(\epsilon_{1,1},\ldots,\epsilon_{1,k}, \epsilon_{2,1},\ldots,\epsilon_{2,k})\in[s_p,2s_p]^{2k}\mapsto \mathbf{R}(t,\boldsymbol{\epsilon})$ where $\mathbf{R}(t,\boldsymbol{\epsilon})=(R_{1,1}(t,\boldsymbol{\epsilon}),\ldots,R_{2,k}(t,\boldsymbol{\epsilon}))$ and the sequence $s_p$  as in  Lemma~\ref{lemma:overlap_concentr}. Assume that for all $t\in [0,1]$ and $\boldsymbol{\epsilon}\in[s_p,2s_p]^{2k}$ we have $r_{2,i}(t,\boldsymbol{\epsilon})=\rho-\ex\langle Q_i\rangle_{t,\epsilon}$. Then
    \begin{multline}
     i_p= \frac{1}{s_{p}^{2k}}\int_{[s_p,2s_p]^{2k}}d\boldsymbol{\epsilon}\Big\{\sum_{i=1}^k l_iI_1\Big(\int_0^1 
 r_{1,i}(t,\boldsymbol{\epsilon})dt\Big)+I_2\big(\int_0^1 
 r_{2,1}(t,\boldsymbol{\epsilon})dt,\ldots,\int_0^1 
 r_{2,k}(t,\boldsymbol{\epsilon})dt\Big)\\
 -\frac{1}{2}\sum_{i=1}^k\int_0^1l_i r_{1,i}(t,\boldsymbol{\epsilon})r_{2,i}(t,\boldsymbol{\epsilon})dt\Big\}+o(1).
    \end{multline}
\end{lemma}

\paragraph{Step 4}: In the last step we obtain the asymptotic expression for the normalized mutual information as a min-max problem. 
\begin{lemma}[Upper bound]{\label{lemma_ub}}
    \begin{align}\label{eq:upper_bound}
        \limsup_{p\rightarrow \infty}i_p\leq \inf_{\mathbf{r}_1\in[0,+\infty)^k}\sup_{\mathbf{r}_2\in[0,\rho]^k} i_{\rm RS}(\mathbf{r}_1,\mathbf{r}_2)
    \end{align}
\end{lemma}
\begin{proof}
    The proof of this Lemma is   very similar to Proposition 2.1 in  \cite{barbier2018mutual}, so we will just go through the main steps. We fix $\mathbf{r}_1(t)=(r_{1,1},\ldots,r_{1,k})$ with non-negative values of $r_{1,i}$ and, for $\boldsymbol{\epsilon}\in[s_p,2s_p]^{2k}$, we take the solution $\mathbf{R}(t,\boldsymbol{\epsilon})$ of first order differential equation: $\partial_t R_{1,i}(t)=r_{1,i}=:F_{1,i}$, $\partial_tR_{2,i}(t)=\rho-\ex\langle Q_i\rangle_{t,\epsilon}=:F_{2,i}(t,\mathbf{R}(t))$. One can easily check that  solution of such equation exists and is unique, since  $\mathbf{F}=(F_{1,1},\ldots,F_{2,k})$ is continuous with continuous non-negative partial derivatives with respect to $R_{a,i}$ for $a=1,2$ and $i=1,\ldots,k$ (see Proposition 6 in \cite{barbier2019optimal}). From this also follows that $\mathbf{R}(t,\boldsymbol{\epsilon})$ is  $\mathcal{C}^1$ in each argument. Moreover, by the Liouville formula, the Jacobian of $\boldsymbol{\epsilon}\rightarrow \mathbf{R}(t,\boldsymbol{\epsilon})$ is $\exp\{\int_0^t\sum_{i=1}^k\partial_{R_{2,i}}\mathbb{F}_2(s,\mathbf{R}(s,\boldsymbol{\epsilon}))\}\geq 1$ which yields that $\mathbf{R}_2$ is a $\mathcal{C}^1$ diffeomorphism. This allows as to apply Lemma~\ref{lem:fund_iden} which, together with (\ref{eq:replica_potential}) and (\ref{eq:I2}), gives
    \begin{align}
        i_p=\frac{1}{s_{p}^{2k}}\int_{[s_p,2s_p]^{2k}} i_{\rm RS}(\mathbf{r}_1,\int_0^1 \mathbf{r}_2(t,\boldsymbol{\epsilon})dt)d\boldsymbol{\epsilon}+o(1)\leq \sup_{\mathbf{r}_2\in[0,\rho]^k}i_{\rm RS}(\mathbf{r}_1, \mathbf{r}_2)+o(1).
    \end{align}

Since it holds for all $r_{1,i}\geq0$ we obtain immediately  (\ref{eq:upper_bound}).
\end{proof}

\begin{lemma}[Lower bound]
    \begin{align}
          \liminf_{p\rightarrow \infty}i_p\geq \inf_{\mathbf{r}_1\in[0,+\infty)^k}\sup_{\mathbf{r}_2\in[0,\rho]^k} i_{\rm RS}(\mathbf{r}_1,\mathbf{r}_2)
    \end{align}
\end{lemma}
\begin{proof}
We continue to follow the line of the proof of Lemma~2.2 in \cite{barbier2018mutual}. For this we fix $\mathbf{R}=(R_{1,1},\ldots,R_{2,k})$, where for each $i=1\ldots,k$  $R_{1,i}(t,\boldsymbol{\epsilon})=\epsilon_{1,i}+\int_0^t r_{1,i}(s,\boldsymbol{\epsilon})ds$ and $R_{2,i}(t,\boldsymbol{\epsilon})=\epsilon_{2,i}+\int_0^t r_{2,i}(s,\boldsymbol{\epsilon})ds$ are the solutions to the Cauchy problem
\begin{align}\label{eq:dif_eq}
\begin{cases}
    \partial R_{1,i}=\mathcal{A}_i(\rho-\ex\langle Q_1\rangle_{t,\boldsymbol{\epsilon}},\ldots,\rho-\ex\langle Q_k\rangle_{t,\boldsymbol{\epsilon}})=:F_{1,i}(t,\mathbf{R}(t))\\
    \partial R_{2,i}=\rho-\ex\langle Q_i\rangle_{t,\boldsymbol{\epsilon}}=:F_{2,i}(t,\mathbf{R}(t))\\
    \mathbf{R}(0)=\boldsymbol{\epsilon}
\end{cases}   
\end{align}
where $\mathcal{A}(\mathbf{z})=(\mathcal{A}_1(\mathbf{z}),\ldots,\mathcal{A}_k(\mathbf{z})):\mathbb{R}^k\rightarrow\mathbb{R}^k$ 
\begin{align}\label{eq:def_mathcal_A}
    \mathcal{A}_i(\mathbf{z})=\frac{c}{\pi}\int_{0}^\pi\frac{\delta_i(\theta)d\theta}{\sum_jz_jl_j\delta_j(\theta)+\sigma^2}.
\end{align}
It is worth to mention that since $\delta_i(\theta)>0$ for all $i$, we have $\mathcal{A}((\mathbb{R}^{+})^k)=(\mathbb{R}^{+})^k$.
One can check (as in Proposition 6 in \cite{barbier2019optimal}) that $\mathbf{F}(t,\mathbf{R})=(\mathbf{F}_1(t,\mathbf{R}),\mathbf{F}_2(t,\mathbf{R}))$ is bounded $\mathcal{C}^1$ function of $\mathbf{R}$, which implies that the unique solution of (\ref{eq:dif_eq}) is also of class $\mathcal{C}^1$. Since $\rho-\ex\langle Q_i\rangle\in [0,\rho]$ we get that for solution $\mathbf{R}$ holds $r_{2,i}\in[0,\rho]$ and $r_{1,i}\geq0$ (as $\mathcal{A}_1:\mathbb{R}^+\rightarrow\mathbb{R}^+$). As before we use Liouville formula for  Jacobian $J_{p,\epsilon}(t)$ of the map $\boldsymbol{\epsilon}\rightarrow \mathbf{R}(t,\boldsymbol{\epsilon})$, which gives 
$$J_{p,\epsilon}(t)=\exp\{\int_0^t\sum_{i=1}^k(\partial_{R_{1,i}}\mathbf{F}_1(s,\mathbf{R}(s,\boldsymbol{\epsilon}))+\partial_{R_{2,i}}\mathbf{F}_2(s,\mathbf{R}(s,\boldsymbol{\epsilon})))\}\geq 1.$$
Last inequality is true since $$\partial_{R_{1,i}}\mathbf{F}_1(s,\mathbf{R}(s,\boldsymbol{\epsilon}))=\frac{c}{\pi k}\int\frac{
\delta_i^2(\theta)}{k^{-1}\sum_j\delta_j(\theta)R_{1,j}+\sigma^2}>0.$$ 
By similar arguments as in Lemma \ref{lemma_ub}, we get that for any $t$ the map  is a diffeomorphism, hence, we can use Lemma~\ref{lem:fund_iden}.
    \begin{multline}
     i_p= \frac{1}{s_{p}^{2k}}\int_{[s_p,2s_p]^{2k}}d\boldsymbol{\epsilon}\Big\{\sum_{i=1}^k l_iI_1(\int_0^1 
 r_{1,i}(t,\boldsymbol{\epsilon})dt)+I_2(\int_0^1 
 r_{2,1}(t,\boldsymbol{\epsilon})dt,\ldots,\int_0^1 
 r_{2,k}(t,\boldsymbol{\epsilon})dt)\\
 -\frac{1}{2}\sum_{i=1}^k\int_0^1 l_ir_{1,i}(t,\boldsymbol{\epsilon})r_{2,i}(t,\boldsymbol{\epsilon})dt\Big\}+o(1).
    \end{multline}
    It is known that $I_1$ is a concave function (e.g.  \cite{barbier2019optimal}), let us show that $I_2(x_1,\ldots,x_k)$ is also concave. Due to (\ref{eq:I2}) we have $I_2(x_1,\ldots,x_k)=\frac{c}{2\pi}\int_0^\pi\ln(\sum_il_i\delta_i(\theta)x_i/\sigma^2+1)d\theta$ with Hessian $\mathbf{H}$ with elements
    \begin{align*}
        H_{ij}=\partial_{x_ix_j}I_2=-\frac{c}{2\pi\sigma^4}\int_0^\pi\frac{l_il_j\delta_i(\theta)\delta_j(\theta)}{(\sum_ml_m\delta_m(\theta)x_m/\sigma^2+1)^2}d\theta
    \end{align*}

    It is sufficient to show that $-\mathbf{H}$ is positive semi-definite, for this we take an arbitrary $k$~--dimensional vector $\mathbf{y}$ and write
    \begin{align*}
        -\mathbf{y}^\intercal \mathbf{H}\mathbf{y}=\sum_{i,j}\frac{c}{2\pi\sigma^4}\int_0^\pi\frac{y_il_il_j\delta_i(\theta)\delta_j(\theta)y_j}{(\sum_ml_m\delta_m(\theta)x_m/\sigma^2+1)^2}d\theta
        =\frac{c}{2\pi\sigma^4}\int_0^\pi\frac{(\sum_{i}y_il_i\delta_i(\theta))^2}{(\sum_ml_m\delta_m(\theta)x_m/\sigma^2+1)^2}d\theta\geq0.
    \end{align*}
 Then Jensen's inequality yields
 \begin{multline*}
     i_p\geq \frac{1}{s_{p}^{2k}}\int_{[s_p,2s_p]^{2k}}d\boldsymbol{\epsilon}\int_{[0,1]^k}dt_1\ldots dt_k\Big\{\sum_{i=1}^k l_iI_1( 
 r_{1,i}(t_i,\boldsymbol{\epsilon}))+I_2(
 r_{2,1}(t_1,\boldsymbol{\epsilon}),\ldots, r_{2,k}(t_k,\boldsymbol{\epsilon}))\\
 -\frac{1}{2}\sum_{i=1}^k l_ir_{1,i}(t_i,\boldsymbol{\epsilon})r_{2,i}(t_i,\boldsymbol{\epsilon})\Big\}+o(1)=\frac{1}{s_{p}^{2k}}\int_{[s_p,2s_p]^{2k}}d\boldsymbol{\epsilon}\int_{[0,1]^k}dt_1\ldots dt_ki_{\rm RS}(r_{1,1}(t_1,\boldsymbol{\epsilon}),\ldots,r_{2,k}(t_k,\boldsymbol{\epsilon}))+o(1).
 \end{multline*}
 We want to show that $i_{\rm RS}(r_{1,1}(t_1,\boldsymbol{\epsilon}),\ldots,r_{2,k}(t_k,\boldsymbol{\epsilon}))=\sup_{\mathbf{r}_2\in[0,\rho]^k}i_{\rm RS}(r_{1,1}(t_1),\ldots,r_{1,k}(t_k),\mathbf{r}_2).$ Indeed, let us  define function  $g: \mathbf{r}_2\mapsto i_{\rm RS}(r_1^1(t_1,\boldsymbol{\epsilon}),\ldots,r_1^k(t_k,\boldsymbol{\epsilon}),\mathbf{r}_2)$, since $I_2(\mathbf{r}_2)$ is concave, $g(\mathbf{r}_2)$ is also concave.  From the definition of the solution $\mathbf{R}(t,\boldsymbol{\epsilon})$ for each $i$ we have $\partial_{r_{2,i}} g(r_{2,i}(t_1,\boldsymbol{\epsilon}),\ldots,r_{2,k}(t_k,\boldsymbol{\epsilon}))=l_i\mathcal{A}_i(\mathbf{r}_2)-l_ir_{1,i}(t_i)=0$. This, together with concavity implies that $g$ reaches it's maximum at $(r_{2,1}(t_1,\boldsymbol{\epsilon}),\ldots,r_{2,k}(t_k,\boldsymbol{\epsilon}))$. From what immediately follows 
 \begin{multline*}
     i_p\geq \frac{1}{s_{p}^{2k}}\int_{[s_p,2s_p]^{2k}}d\boldsymbol{\epsilon}\int_{[0,1]^k}dt_1\ldots dt_k\sup_{\mathbf{r}_2\in[0,\rho]^k}i_{\rm RS}(r_{1,1}(t_1,\boldsymbol{\epsilon}),\ldots,r_{1,k}(t_k,\boldsymbol{\epsilon}),\mathbf{r}_2)+o(1)\\
 \geq \inf_{\mathbf{r}_1\in[0,+\infty)^k}\sup_{\mathbf{r}_2\in[0,\rho]^k}i_{\rm RS}(\mathbf{r}_1,\mathbf{r}_2)+o(1).
 \end{multline*} 

\end{proof}

\subsection{Proof of Theorem~\ref{th:ymmse_mmse}: the measurement MMSE}\label{sec:mmse}

For the convenience we repeat here some of the  notations. We denoted 
\begin{align}
    \text{mmse}^{(i)}=\frac{1}{|I_i|}\ex\|\bbeta^{(i)}_0-\langle\bbeta^{(i)}\rangle\|^2
\end{align}
and
\begin{align}
    \text{ymmse}=\frac{1}{N}\ex\|\frac{1}{\sqrt{p}}\boldsymbol{\Phi}_N\bbeta_0-\langle\frac{1}{\sqrt{p}}\boldsymbol{\Phi}_N\bbeta\rangle\|^2.
\end{align}
The goal is to prove
\begin{align}
        \lim_{p\rightarrow\infty}\text{ymmse}=\frac{\sigma^2}{\pi}\int_0^{\pi}\frac{\sum_{i=1}^{k}l_i\delta_i(\theta)\text{mmse}^{(i)}}{\sum_{i=1}^{k}l_i\delta_i(\theta)\text{mmse}^{(i)}+\sigma^2}d\theta.
    \end{align}
\begin{proof} We define matrix $\mathbf{E}=(\mathbf{I}_N+\sigma^{-2}\sum_{i=1}^kl_i\LLambda_{i,N}\mmse^{(i)})^{-1/2}$. This is symmetric deterministic  matrix of dimension $N\times N$, since $\mmse^{(i)}$ is non-negative, the matrix inside the brackets is positive-defined.
    First we show that 
    \begin{align}\label{eq:first_prove}
        \frac{1}{N}\ex[\|\frac{1}{\sqrt{p}}\mathbf{E}\boldsymbol{\Phi}\bbeta_0-\langle\frac{1}{\sqrt{p}}\mathbf{E}\boldsymbol{\Phi}\bbeta\rangle\|^2]=-\frac{1}{N\sqrt{p}}\ex[\Z^\intercal\mathbf{E}^2\langle\boldsymbol{\Phi}(\bbeta_0-\bbeta)\rangle]
    \end{align}
    Indeed, if we integrate by parts the RHS with respect to noise,  we obtain
    \begin{align}
         &-\frac{1}{N\sqrt{p}}\sum_{\mu=1}^N\ex[Z_\mu\langle(\mathbf{E}^2\boldsymbol{\Phi}(\bbeta_0-\bbeta))_\mu\rangle]\nonumber\\
         = & \frac{1}{N\sqrt{p}}\sum_{\mu=1}^N\ex[\langle(\mathbf{E}^2\boldsymbol{\Phi}(\bbeta_0-\bbeta))_\mu(\frac{1}{\sqrt{p}}\boldsymbol{\Phi}(\bbeta_0-\bbeta)+\Z)_\mu\rangle]\nonumber\\
         &-\frac{1}{N\sqrt{p}}\sum_{\mu=1}^N\ex[\langle(\mathbf{E}^2\boldsymbol{\Phi}(\bbeta_0-\bbeta))_\mu\rangle\langle(\frac{1}{\sqrt{p}}\boldsymbol{\Phi}(\bbeta_0-\bbeta)+\Z)_\mu\rangle]\nonumber\\
         =&\frac{1}{N}\ex[\langle\|\frac{1}{\sqrt{p}}\mathbf{E}\boldsymbol{\Phi}(\bbeta_0-\bbeta)\|^2\rangle]-\frac{1}{N}\ex[\|\langle\frac{1}{\sqrt{p}}\mathbf{E}\boldsymbol{\Phi}(\bbeta_0-\bbeta)\rangle\|^2]
    \end{align}
Due to Nishimori identity

\begin{align}\label{eq:examp_Nish}
    \frac{1}{N}\ex[\langle\|\frac{1}{\sqrt{p}}\mathbf{E}\boldsymbol{\Phi}(\bbeta_0-\bbeta)\|^2\rangle]=2\frac{1}{N}\ex[\|\langle\frac{1}{\sqrt{p}}\mathbf{E}\boldsymbol{\Phi}(\bbeta_0-\bbeta)\rangle\|^2]
\end{align}
Plugging this to the previous expression we get
\begin{align}
    -\frac{1}{N\sqrt{p}}\sum_{\mu=1}^N\ex[Z_\mu\langle(\mathbf{E}^2\boldsymbol{\Phi}(\bbeta_0-\bbeta))_\mu\rangle]
         =\frac{1}{N}\ex[\|\langle\frac{1}{\sqrt{p}}\mathbf{E}\boldsymbol{\Phi}(\bbeta_0-\bbeta)\rangle\|^2]
\end{align}
On the other hand, since $\boldsymbol{\Phi}(\bbeta_0-\bbeta)=\sum_{i=1}^k\LLambda^{(i)}\G^{(i)}(\bbeta_0^{(i)}-\bbeta^{(i)})$ we can integrate by parts with respect to $G^{(i)}_{\nu j}$ LHS of the last expression and get:
\begin{align}
    &\frac{1}{N}\ex[\|\langle\frac{1}{\sqrt{p}}\mathbf{E}\boldsymbol{\Phi}(\bbeta_0-\bbeta)\rangle\|^2]\nonumber\\
    =&\frac{1}{\sigma^2N\sqrt{p}}\sum_{i=1}^k\sum_{\mu,\nu=1}^N\sum_{j\in I_i}\ex[(\mathbf{E}^2\Z)_\mu\langle\Lambda^{(i)}_{\mu\nu}(\bbeta^{(i)}_0-\bbeta^{(i)})_j\sum_{\mu_1}(\frac{1}{\sqrt{p}}\boldsymbol{\Phi}(\bbeta_0-\bbeta)+\Z)_{\mu_1}\frac{1}{\sqrt{p}}\Lambda_{\mu_1\nu}^{(i)}(\bbeta_0^{(i)}-\bbeta^{(i)})_j\rangle]\nonumber\\
    &-\frac{1}{\sigma^2N\sqrt{p}}\sum_{i=1}^k\sum_{\mu,\nu=1}^N\sum_{j\in I_i}\ex[(\mathbf{E}^2\Z)_\mu\langle\Lambda^{(i)}_{\mu\nu}(\bbeta^{(i)}_0-\bbeta^{(i)})_j\rangle\langle\sum_{\mu_1}(\frac{1}{\sqrt{p}}\boldsymbol{\Phi}(\bbeta_0-\bbeta)+\Z)_{\mu_1}\frac{1}{\sqrt{p}}\Lambda_{\mu_1\nu}^{(i)}(\bbeta_0^{(i)}-\bbeta^{(i)})_j\rangle]\label{eq:E_ymmse}
    \end{align}
    Let us first take the closer look at the second term of the RHS After writing it in a vector form we obtain
    \begin{align*}
        &\frac{1}{\sigma^2Np}\sum_{i=1}^k\ex[\Z^\intercal\mathbf{E}^2\LLambda_{i,N}\langle(\frac{1}{\sqrt{p}}\boldsymbol{\Phi}(\bbeta_0-\bbeta)+\Z)(\bbeta_0^{(i)}-\bbeta^{(i)})^\intercal\rangle\langle\bbeta^{(i)}_0-\bbeta^{(i)}\rangle]\\
        =&\frac{1}{\sigma^2Np}\sum_{i=1}^k\ex[\Z^\intercal\mathbf{E}^2\LLambda_{i,N}\langle(\frac{1}{\sqrt{p}}\boldsymbol{\Phi}(\bbeta_0-\bbeta)+\Z)(\bbeta_0^{(i)}-\bbeta^{(i)})^\intercal\rangle\bbeta^{(i)}_0]\\
        &-\frac{1}{\sigma^2Np}\sum_{i=1}^k\ex[\Z^\intercal\mathbf{E}^2\LLambda_{i,N}\langle(\frac{1}{\sqrt{p}}\boldsymbol{\Phi}(\bbeta_0-\bbeta)+\Z)(\bbeta_0^{(i)}-\bbeta^{(i)})^\intercal\rangle\langle\bbeta^{(i)}\rangle].
    \end{align*}
    The last term is zero due to Nishimori identity, indeed
     \begin{align*}
 &\frac{1}{\sigma^2Np}\sum_{i=1}^k\ex[\Z^\intercal\mathbf{E}^2\LLambda_{i,N}\langle(\frac{1}{\sqrt{p}}\boldsymbol{\Phi}(\bbeta_0-\bbeta)+\Z)(\bbeta_0^{(i)}-\bbeta^{(i)})^\intercal\rangle\langle\bbeta^{(i)}\rangle]\\
 =&\frac{1}{\sigma^2Np}\sum_{i=1}^k\ex[\langle(\Y-\frac{1}{\sqrt{p}}\boldsymbol{\Phi}\bbeta_0)^\intercal\mathbf{E}^2\LLambda_{i,N}(\Y-\frac{1}{\sqrt{p}}\boldsymbol{\Phi}\bbeta)\bbeta_0^{(i)\intercal}\rangle\langle\bbeta^{(i)}\rangle]\\
 &-\frac{1}{\sigma^2Np}\sum_{i=1}^k\ex[\langle(\Y-\frac{1}{\sqrt{p}}\boldsymbol{\Phi}\bbeta)^\intercal\mathbf{E}^2\LLambda_{i,N}(\Y-\frac{1}{\sqrt{p}}\boldsymbol{\Phi}\bbeta_0)\bbeta^{(i)\intercal}\rangle\langle\bbeta^{(i)}\rangle]=0
    \end{align*}
    Returning back to (\ref{eq:E_ymmse}) we obtain
    \begin{align}
        &\frac{1}{N}\ex[\|\langle\frac{1}{\sqrt{p}}\mathbf{E}\boldsymbol{\Phi}(\bbeta_0-\bbeta)\rangle\|^2]\label{eq:Y_1-Y-2}\\
        =&\frac{1}{\sigma^2Np}\sum_{i=1}^k\ex[\Z^\intercal\mathbf{E}^2\LLambda_{i,N}\Z\langle\|\bbeta^{(i)}_0-\bbeta^{(i)}\|^2\rangle]+
        \frac{1}{\sigma^2Np^{3/2}}\sum_{i=1}^k\ex[\Z^\intercal\mathbf{E}^2\LLambda_{i,N}\langle\boldsymbol{\Phi}(\bbeta_0-\bbeta)\|\bbeta^{(i)}_0-\bbeta^{(i)}\|^2\rangle]\nonumber\\
        &-
         \frac{1}{\sigma^2Np}\sum_{i=1}^k\ex[\Z^\intercal\mathbf{E}^2\LLambda_{i,N}\Z\langle(\bbeta_0^{(i)}
         -\bbeta^{(i)})^\intercal\bbeta^{(i)}_0\rangle]\nonumber\\&
         -\frac{1}{\sigma^2Np^{3/2}}\sum_{i=1}^k\ex[\Z^\intercal\mathbf{E}^2\LLambda_{i,N}\langle\boldsymbol{\Phi}(\bbeta_0-\bbeta)(\bbeta_0^{(i)}-\bbeta^{(i)})^\intercal\rangle\bbeta^{(i)}_0]\nonumber\\
         =&\underbrace{-\frac{1}{\sigma^2Np}\sum_{i=1}^k\ex[\Z^\intercal\mathbf{E}^2\LLambda_{i,N}\Z\langle(\bbeta^{(i)}_0-\bbeta^{(i)})^\intercal\bbeta^{(i)}\rangle]}_{\mathcal{Y}_1}\nonumber\\
         &-\underbrace{
        \frac{1}{\sigma^2Np^{3/2}}\sum_{i=1}^k\ex[\Z^\intercal\mathbf{E}^2\LLambda_{i,N}\langle\boldsymbol{\Phi}(\bbeta_0-\bbeta)(\bbeta^{(i)}_0-\bbeta^{(i)})^\intercal\bbeta^{(i)}\rangle]}_{\mathcal{Y}_2}
        =\mathcal{Y}_1-\mathcal{Y}_2.\nonumber
    \end{align}
    Using again Nishimori we have
    \begin{align*}
         \mathcal{Y}_1 =&-\frac{1}{\sigma^2Np}\sum_{i=1}^k\ex[\Z^\intercal\mathbf{E}^2\LLambda_{i,N}\Z\langle(\bbeta^{(i)}_0-\bbeta^{(i)})^\intercal\bbeta^{(i)}\rangle]\\
         &=\frac{1}{\sigma^2Np}\sum_{i=1}^k\ex[\Z^\intercal\mathbf{E}^2\LLambda_{i,N}\Z\|\bbeta^{(i)}_0-\langle\bbeta^{(i)}\rangle\|^2]
         \end{align*}
          The random variable $\Z^\intercal\mathbf{E}^2\LLambda_{i,N}\Z/N$ is just an averaged norm of a Gaussian vector with covariance matrix $\mathbf{E}^2\LLambda_{i,N}$, its variance is of order $O(1/N)$, so applying Cauchy-Schwarz inequality we get 
         \begin{align*}
         \mathcal{Y}_1=&\frac{1}{\sigma^2}\ex[\Z^\intercal\mathbf{E}^2\frac{1}{N}\sum_{i=1}^kl_i\LLambda_{i,N}\mmse^{(i)}\Z]+o(1)\\
         &=\frac{1}{N}\mathrm{Tr}\Big((\sum_{i=1}^kl_i\LLambda_{i,N}\mmse^{(i)})(\mathbf{I}_N+\sum_{i=1}^kl_i\LLambda_{i,N}\mmse^{(i)})^{-1}\Big)+o(1)
    \end{align*}
To deal with $\mathcal{Y}_2$ we integrate by parts once again with respect to the elements $Z_\mu$ of the noise 
\begin{align*}
    \mathcal{Y}_2=&-\frac{1}{\sigma^2Np^{3/2}}\sum_{i=1}^k\sum_{\mu=1}^N\ex[\langle(\mathbf{E}^2\LLambda_{i,N}\boldsymbol{\Phi}(\bbeta_0-\bbeta))_\mu(\frac{1}{\sqrt{p}}\boldsymbol{\Phi}(\bbeta_0-\bbeta)+\Z)_\mu(\bbeta^{(i)}_0-\bbeta^{(i)})^\intercal\bbeta^{(i)}\rangle]\\
    &+\frac{1}{\sigma^2Np^{3/2}}\sum_{i=1}^k\sum_{\mu=1}^N\ex[\langle(\mathbf{E}^2\LLambda_{i,N}\boldsymbol{\Phi}(\bbeta_0-\bbeta))_\mu(\bbeta^{(i)}_0-\bbeta^{(i)})^\intercal\bbeta^{(i)}\rangle\langle(\frac{1}{\sqrt{p}}\boldsymbol{\Phi}(\bbeta_0-\bbeta)+\Z)_\mu\rangle]\\
    =&\frac{1}{\sigma^2Np^2}\sum_{i=1}^k\ex[\langle(\boldsymbol{\Phi}\bbeta)^\intercal\mathbf{E}^2\LLambda_{i,N}(\boldsymbol{\Phi}(\bbeta_0-\bbeta))(\bbeta^{(i)}_0-\bbeta^{(i)})^\intercal\bbeta^{(i)}\rangle]\\
    &-
    \frac{1}{\sigma^2Np^2}\sum_{i=1}^k\ex[\langle(\boldsymbol{\Phi}\bbeta^\prime)^\intercal\mathbf{E}^2\LLambda_{i,N}(\boldsymbol{\Phi}(\bbeta_0-\bbeta))(\bbeta^{(i)}_0-\bbeta^{(i)})^\intercal\bbeta^{(i)}\rangle]
\end{align*}
To obtain this we notice that in both terms when unfolding brackets with $\frac{1}{\sqrt{p}}\boldsymbol{\Phi}(\bbeta_0-\bbeta)+\Z$, the part with $\frac{1}{\sqrt{p}}\boldsymbol{\Phi}\bbeta_0+\Z$ can go outside the Gibbs brackets and will be cancelled. To obtain the last term we use Nishimori to replace $\langle\bbeta\rangle$ with its independent copy $\langle\bbeta^\prime\rangle$. We state here two Lemmas, which are necessary to conclude the proof of Theorem~\ref{th:ymmse_mmse}. The proofs of the Lemmas are postponed to the Appendix.
\begin{lemma}{\label{lemma_y2_1}}
   \begin{equation}
    \frac{1}{Np^2}\sum_{i=1}^k\ex[\langle(\boldsymbol{\Phi}\bbeta^\prime)^\intercal\mathbf{E}^2\LLambda_{i,N}(\boldsymbol{\Phi}(\bbeta_0-\bbeta))(\bbeta^{(i)}_0-\bbeta^{(i)})^\intercal\bbeta^{(i)}\rangle=O(p^{-1/2}) \label{err_Y_2}
\end{equation} 
\end{lemma}

\begin{lemma}{\label{lemma_y2_2}}
\begin{align*}&\frac{1}
{\sigma^2Np^2}\sum_{i=1}^k\ex[\langle(\boldsymbol{\Phi}\bbeta)^\intercal\mathbf{E}^2\LLambda_{i,N}(\boldsymbol{\Phi}(\bbeta_0-\bbeta))(\bbeta^{(i)}_0-\bbeta^{(i)})^\intercal\bbeta^{(i)}\rangle]\\
=&\frac{1}{\sigma^2Np^2}\sum_{i=1}^k\ex[\langle(\boldsymbol{\Phi}\bbeta)^\intercal\mathbf{E}^2\LLambda_{i,N}(\boldsymbol{\Phi}(\bbeta_0-\bbeta))\rangle]\ex[\langle(\bbeta^{(i)}_0-\bbeta^{(i)})^\intercal\bbeta^{(i)}\rangle]+o(1).
\end{align*}
\end{lemma}

 Combining these Lemmas and using again Nishimori (similar to \eqref{eq:examp_Nish}) we get
\begin{align}
    \mathcal{Y}_2=\frac{1}{N}\ex\Big[(p^{-1/2}\boldsymbol{\Phi}(\bbeta_0-\langle\bbeta\rangle))^\intercal\mathbf{E}^2\Big(\sigma^{-2}\sum_{i=1}^kl_i\LLambda_{i,N}\mmse^{(i)}\Big)(p^{-1/2}\boldsymbol{\Phi}(\bbeta_0-\langle\bbeta\rangle))\Big]+o(1).
\end{align}
Finally, bringing together  expressions for $\mathcal{Y}_1$ and $\mathcal{Y}_2$ with (\ref{eq:Y_1-Y-2}) we have
\begin{align*}
    &\ex\Big[(p^{-1/2}\boldsymbol{\Phi}(\bbeta_0-\langle\bbeta\rangle))^\intercal\mathbf{E}^2\Big(\underbrace{\mathbf{I}_N+\sigma^{-2}\sum_{i=1}^kl_i\LLambda_{i,N}\mmse^{(i)}}_{\mathbf{E^{-2}}}\Big)(p^{-1/2}\boldsymbol{\Phi}(\bbeta_0-\langle\bbeta\rangle))\Big]\\
    =&\frac{1}{N}\mathrm{Tr}\Big((\sum_{i=1}^kl_i\LLambda_{i,N}\mmse^{(i)})(\mathbf{I}_N+\sum_{i=1}^kl_i\LLambda_{i,N}\mmse^{(i)})^{-1}\Big)+o(1),
\end{align*}
which brings us immediately
\begin{align}
    \ymmse=\frac{1}{N}\mathrm{Tr}\Big((\sum_{i=1}^kl_i\LLambda_{i,N}\mmse^{(i)})(\mathbf{I}_N+\sigma^{-2}\sum_{i=1}^kl_i\LLambda_{i,N}\mmse^{(i)})^{-1}\Big)+o(1)
\end{align}
After taking the limit we obtain the desired expression.
\end{proof}

\subsection{Proof of Theorem~\ref{th:gen_A}: general diagonal \texorpdfstring{$\mathbf{A}_p$}{A}}\label{sec:k_infty}

The main idea is to approximate the general model using a sequence of models  satisfying the assumptions of Theorem~\ref{main_thm} and control the error.
For each $k\le p$ we divide the support $[a,b]$ into $k$ disjoint parts $[a,b]=\bigcup_{i=1}^k[\alpha_{i-1},\alpha_i]$, with $\alpha_i=a+i(b-a)/k$, and define a sequence of  $p\times p$ diagonal matrices $\bar{\mathbf{A}}_{p}(k)={\rm diag}(\bar{\lambda}_{1,p}(k),\ldots,\bar{\lambda}_{p,p}(k))$, where each matrix has exactly $k$ different eigenvalues: $\{\alpha_0,\ldots,\alpha_{k-1}\}$. We remind that the initial sequence of $\{\bA_p\}_{p=1}^\infty$ is of the form $\bA_p={\rm diag}(\lambda_{1,p},\ldots,\lambda_{p,p})$. We desire that $\bar{\bA}_p(k)$ becomes a better approximation to $\bA_p$ as $p,k$ grow. We need first to establish the correct multiplicity of each $\alpha_i$.  Define the set of indices $I_i:=\{j:\lambda_{j,p}\in[\alpha_{i-1},\alpha_i)\}$. It is obvious that $\bigcup_{i=1}^k I_i=\{1,\ldots,p\}$ and $l_i=\lim_{p\rightarrow\infty}|I_i|/p=\eta[\alpha_{i-1},\alpha_i)$, where $\eta$ is the limiting normalised counting measure of $\bA_p$. Then, for $j=1,\ldots,p$ we set $\bar{\lambda}_{j,p}(k)=\alpha_{i-1}$ if $j\in I_i$.  With such construction, we have $|\lambda_{j,p}-\bar{\lambda}_{j,p}(k)|<1/k$ for any $j$ and the limiting normalised counting measure of $\bar{\mathbf{A}}_{p}(k)$ as $p\to\infty$, denoted by  $\bar{\eta}_k$, reads $\bar{\eta}_k=\sum_{i=1}^k l_i\delta_{\alpha_{i-1}}$.
It is straightforward that the Wasserstein-2 distance between the measures $\eta$ and $\bar{\eta}_k$ is of order $1/k$. This implies that $\bar{\eta}_k\Rightarrow^{W_2}\eta$ as $k\rightarrow+\infty$.
 
After fixing $k$, with  sequence $\{\bar{\mathbf{A}}_{p}(k)\}_{p=1}^{\infty}$ we construct $\bar{\Y}_N(k)$ and  $\bar{\boldsymbol{\Phi}}_N(k)$ as in \eqref{eq:def_x_t}-\eqref{eq:mod_vec}. Also, we define the analogues of \eqref{eq:post_inter}, \eqref{eq:mut_inf}, with the mutual information associated with $\bar{\mathbf{A}}_{p}(k)$ denoted by $\bar{i}_p(k)$. For this model, the assumptions of the Theorem~\ref{main_thm} are satisfied. Therefore, we define the associated replica symmetry potential as follows:
\begin{align*}
    &\bar{i}_{\rm RS}(k)(\bar{\mathbf{r}}_{1}(k),\bar{\mathbf{r}}_{2}(k)):=\\
    &\frac{c}{2\pi}\int_{0}^{\pi}\ln \Big(\frac1{\sigma^2 }\sum_{i=1}^kl_i\delta(\theta,\alpha_{i-1})\bar{r}_{2,i}(k)+1\Big)d\theta
   -\frac{1}{2}\sum_{i=1}^kl_i\bar{r}_{1,i}(k)\bar{r}_{2,i}(k)
   +\sum_{i=1}^kl_iI(\beta;\sqrt{\bar{r}_{1,i}(k)}\beta+Z).   
\end{align*}
Now recall the definition of $\delta(\theta,\lambda)=(1-2\lambda\cos\theta+\lambda^2)^{-1}$.
By Theorem \ref{main_thm}, the limit of corresponding mutual information is
\begin{align}\label{eq:i_p_k}
    \lim_{p\rightarrow\infty}\bar{i}_p(k)=\inf_{\bar{\mathbf{r}}_{1}(k)}\sup_{\bar{\mathbf{r}}_2(k)}\bar{i}_{\rm RS}(k)(\bar{\mathbf{r}}_{1}(k),\bar{\mathbf{r}}_{2}(k)),
\end{align}
where we optimize over  $\bar{\mathbf{r}}_{1}(k)\in[0,+\infty)^k$ and $\bar{\mathbf{r}}_2(k)\in [0,\rho]^k$. Equivalently,  according to Remark~\ref{rem:equiv_gamma}, we can write
\begin{align}
    \lim_{p\rightarrow\infty}\bar{i}_p(k)=\inf_{(\bar{\mathbf{r}}_{1}(k),\bar{\mathbf{r}}_2(k))\in\bar{\Gamma}(k)}\bar{i}_{\rm RS}(k)(\bar{\mathbf{r}}_{1}(k),\bar{\mathbf{r}}_{2}(k)).
\end{align}
Here $\bar{\Gamma}(k)$ is defined as the set of solutions of the system of equations
\begin{align}
    \bar{r}_{2,i}&=\text{mmse}(\beta|\sqrt{\bar{r}_{1,i}}\beta+Z),\\
    \bar{r}_{1,i}&=\frac{c}{\pi}\int_{0}^\pi\frac{\delta(\theta,\alpha_{i-1})d\theta}{\sum_jl_j\bar{r}_{2,j}\delta(\theta,\alpha_{j-1})+\sigma^2}.
\end{align}
Since $l_i$, $\bar{r}_{2,j}$, and $\delta(\theta,\alpha_{j-1})$ are obviously positive we  can bound $\bar{r}_{1,i}$ by an universal constant $C$ that depends only on $b$, $c$ and $\sigma$:
\begin{align}
    \bar{r}_{1,i}\leq \frac{c}{\pi}\int_{0}^\pi\frac{\delta(\theta,\alpha_{i-1})d\theta}{\sigma^2}\leq \frac{c}{\pi}\int_{0}^\pi\frac{d\theta}{(1-\alpha_{i-1})^2\sigma^2}\leq \frac{c}{(1-b)^2\sigma^2}=:C<\infty.
\end{align}
This means that in $\eqref{eq:i_p_k}$ we can optimize over set $[0,C]^k\times[0,\rho]^k$.

The remaining goal is to prove that as $k\rightarrow\infty$, the LHS of (\ref{eq:i_p_k}) tends to $\lim_{p\rightarrow\infty}i_p$ and the RHS to $\inf_{\Gamma_\infty}i_{\rm RS}$.

We start with the LHS and show that $\lim_{p\rightarrow\infty}|i_p-\bar{i}_p(k)|=o(k)$. Let us denote 
\begin{align}
H:=\frac{1}{2\sigma^2}\|p^{-1/2}\boldsymbol{\Phi}_p(\bbeta_0-\bbeta)+\Z\|^2   , \qquad \bar{H}:=\frac{1}{2\sigma^2}\|p^{-1/2}\bar{\boldsymbol{\Phi}}_p(k)(\bbeta_0-\bbeta)+\Z\|^2.
\end{align}
Then we have 
\begin{align}
    \bar{i}_p(k)-i_p=\frac{1}{p}\ex\ln\frac{\int dP_0(\bbeta)e^{-H}}{\int dP_0(\bbeta)e^{-\bar{H}}}
    =\frac{1}{p}\ex\ln\Big\langle e^{-H+\bar{H}}\Big\rangle_{\bar{H}}
    \geq \frac{1}{p}\ex\langle \bar{H}-H\rangle_{\bar{H}}
\end{align}
by Jensen's inequality, and the expectation operator $\langle\cdot\rangle_{\bar{H}}$ stands for
\begin{align}
    \langle\cdot\rangle_{\bar{H}}=\frac{\int dP_0(\bbeta)(\cdot) e^{-\bar{H}}}{\int dP_0(\bbeta)e^{-\bar{H}}}.
\end{align}
On the other hand we can prove similarly that $i_p-\bar{i}_p(k)\geq\frac{1}{p}\ex\langle H-\bar{H}\rangle_{H}$, which gives 
\begin{align}
    \bar{i}_p(k)-i_p\leq\frac{1}{p}\ex\langle \bar{H}-H\rangle_{H}.
\end{align}
Now what is left is to bound is $p^{-1}|\bar{H}-H|$:
\begin{align*}
\frac{1}{p}|\bar{H}-H|&=\frac{1}{2\sigma^2p}|p^{-1}(\bbeta_0-\bbeta)^\intercal(\bar{\boldsymbol{\Phi}}-\boldsymbol{\Phi})^\intercal\bar{\boldsymbol{\Phi}}(\bbeta_0-\bbeta)+p^{-1}(\bbeta_0-\bbeta)^\intercal\boldsymbol{\Phi}^\intercal(\bar{\boldsymbol{\Phi}}-\boldsymbol{\Phi})(\bbeta_0-\bbeta)\\
&+2p^{-1/2}(\bbeta_0-\bbeta)^\intercal(\bar{\boldsymbol{\Phi}}-\boldsymbol{\Phi})^\intercal\Z|
\leq\|\bar{\boldsymbol{\Phi}}-\boldsymbol{\Phi}\|\Big(\frac{1}{2\sigma^2p^2}\|\bbeta_0-\bbeta\|^2\|\bar{\boldsymbol{\Phi}}+\boldsymbol{\Phi}\|+\frac{1}{\sigma^2p^{3/2}}\|\bbeta_0-\bbeta\|\|\Z\|\Big)\\
&\leq\max_{i\le k}\|\LLambda_{i,N}-\bar{\LLambda}_i(k)\|\times\|\mathbf{G}\|\Big(\frac{1}{2\sigma^2p^2}\|\bbeta_0-\bbeta\|^2\|\bar{\boldsymbol{\Phi}}+\boldsymbol{\Phi}\|+\frac{1}{\sigma^2p^{3/2}}\|\bbeta_0-\bbeta\|\|\Z\|\Big),
\end{align*}
where $\mathbf{G}=\mathbf{G}_N$ is the $N\times p$ matrix composed of the blocks $\mathbf{G}_N^{(i)}$ with i.i.d.  Gaussian entries in $\boldsymbol{\Phi}$, see \eqref{diagonal_lambda_model}. The last inequality can be proved by verifying that for any vectors $\X$ and $\mathbf{y}$ or respective sizes, we have $\X^\intercal(\bar{\boldsymbol{\Phi}}-\boldsymbol{\Phi})\mathbf{y}\leq\X^\intercal\max_{i\le k}\|\LLambda_{i,N}-\bar{\LLambda}_i(k)\|\mathbf{G}\mathbf{y}$. Applying twice the Cauchy-Schwarz inequality we obtain that the expectation w.r.t. $\mathbb{E}\langle \cdot\rangle_{H\, {\rm or}\,\bar H}$ of the random part is of order one. To bound  $\lim_{p\rightarrow\infty}\max_{i\le k}\|\LLambda_{i,N}-\bar{\LLambda}_i(k)\|$ we recall that asymptotically the eigenspaces of $\boldsymbol{\Lambda}_i$ and $\bar{\boldsymbol{\Lambda}}_i(k)$ are the same (see Appendix~\ref{apx:KMS}) so
 \begin{align}
     \lim_{p\rightarrow\infty}\|\LLambda_{i,N}-\bar{\LLambda}_i(k)\|\le\Big|\frac{1}{(1-\lambda_i)^2}-\frac{1}{(1-\bar{\lambda}_i(k))^2}\Big|\le\frac{|\lambda_i-\bar{\lambda}_i(k)|(2-\lambda_i-\bar{\lambda}_i(k))}{(1-\bar{\lambda}_i(k))^2(1-\lambda_i)^2}\le\frac{C}{k}.
 \end{align} 
 
 This gives us $\lim_{p\rightarrow\infty}|i_p-\bar{i}_p(k)|=O(1/k)$.

Now, to deal with the RHS of \eqref{eq:i_p_k} we notice that for each pair of vectors $(\bar{\mathbf{r}}_1, \bar{\mathbf{r}}_2)\in\bar{\Gamma}(k)$ we can write corresponding 
Lipschitz bounded functions defined on $[a,b]$ such that  $r_1(\alpha_{i-1})=\bar{r}_{1,i}(k)$ and $r_2(\alpha_{i-1})=\bar{r}_{2,i}(k)$. To see this, we return to the definition of $\bar{\Gamma}(k)$ and define the corresponding functions $r_1(\lambda)$ and $r_2(\lambda)$ as
\begin{align}
        r_2(\lambda)&=\text{mmse}(\beta|\sqrt{r_1(\lambda)}\beta+Z),\label{461}\\
    r_1(\lambda)&=\frac{c}{\pi}\int_{0}^\pi\frac{\delta(\theta,\lambda)d\theta}{\int r_2(s)\delta(\theta,s)d\bar{\eta}_k(s)+\sigma^2}.\label{462}
\end{align}
These automatically verify the aforementioned conditions $(r_1(\alpha_{i-1}),r_2(\alpha_{i-1}))=(\bar{r}_{1,i}(k),\bar{r}_{2,i}(k))$. Then we can rewrite the replica potential equivalently as
\begin{align}
    i_{\rm RS}(\bar{\mathbf{r}}_1(k),\bar{\mathbf{r}}_2(k))&=\frac{c}{2\pi}\int_0^{\pi}\ln\Big(\frac{1}{\sigma^2}\int_a^b\delta(\theta,\lambda)r_2(\lambda)d\bar{\eta}_k(\lambda)\Big)d\theta-\frac{1}{2}\int_a^br_1(\lambda)r_2(\lambda)d\bar{\eta}_k(\lambda)\nonumber\\    +&\int_a^bI(\beta;\sqrt{r_1(\lambda)}\beta+Z)d\bar{\eta}_k(\lambda).
\end{align}
Note that the optimal value of the potential written in this form w.r.t. functions $(r_1,r_2)$ instead of simply vectors $(\bar{\mathbf{r}}_1(k),\bar{\mathbf{r}}_2(k))$ enforces the solution(s) to verify \eqref{461}, \eqref{462} at $\lambda\in \{\alpha_{0},\ldots,\alpha_{k-1}\}$, the only points on which the potential depends on. It implies that we can write
\begin{align}\label{eq:vect-func}
      \inf_{[0,C]^k}\sup_{[0,\rho]^k}i_{\rm RS}(\bar{\mathbf{r}}_1(k),\bar{\mathbf{r}}_2(k))=&\inf_{\mathcal{C}_1}\sup_{\mathcal{C}_2}\Big\{\frac{c}{2\pi}\int_0^{\pi}\ln\Big(\frac{1}{\sigma^2}\int_a^b\delta(\theta,\lambda)r_2(\lambda)d\bar{\eta}_k(\lambda)\Big)d\theta\nonumber\\
      &-\frac{1}{2}\int_a^br_1(\lambda)r_2(\lambda)d\bar{\eta}_k(\lambda)    +\int_a^bI(\beta;\sqrt{r_1(\lambda)}\beta+Z)d\bar{\eta}_k(\lambda)\Big\},
\end{align}
where $\mathcal{C}_1$ is the set of Lipschitz functions $f:[a,b]\rightarrow[0,C]$ and $\mathcal{C}_2$ is the set of Lipschitz functions $f:[a,b]\rightarrow[0,\rho]$. We need to find the limit of the RHS as $k\rightarrow\infty$. In order to exchange the limit with the $\inf\sup$, it is sufficient to show that the potential on the RHS of \eqref{eq:vect-func} converges uniformly on the set $\mathcal{C}_1\times\mathcal{C}_2$. But showing uniformity is straightforward: since $\bar{\eta}_k\rightarrow\eta$ in Wasserstein-2 distance, it implies that $\int f(\lambda)d\bar \eta_k\rightarrow\int f(\lambda)d\eta$ uniformly on the set of Lipschitz functions bounded by $\max(C,\rho)$.

\section*{Acknowledgments} The authors warmly thank Pragya Sur for initial discussions leading to this project. J.B. and D.T. were funded by the European Union (ERC, CHORAL, project number 101039794). Views and opinions expressed are however those of the authors only and do not necessarily reflect those of the European Union or the European Research Council. Neither the European Union nor the granting authority can be held responsible for them. 

\bibliographystyle{abbrv}
\bibliography{replica_time}

\appendices

\section{Replica calculation\label{rep_cal}}


In this appendix we derive 
the replica symmetric  potential (\ref{eq:replica_potential}) of the free energy using the so called replica method \cite{Mezard_Parisi_Virasoro} which was first introduced in the context of the spin glasses and became a powerful method  of statistical physics. 
 It is based on the simple identities that is also known as replica trick 
\begin{equation}
    \lim_{N\rightarrow+\infty}\ex f_N=-\lim_{N\rightarrow+\infty}\frac{1}{p}\ex\ln \mathcal{Z}(\Y)= -\lim_{N\rightarrow+\infty}\lim_{u\rightarrow0}\dfrac{1}{pu}\ln \ex \mathcal{Z}(\Y)^u,
\end{equation}
where the replicated partition function is
\begin{align}\label{int:gaussian_v}
\mathcal{Z}(\Y)^u= \int \prod_{a=1}^udP_0(\bbeta_a)e^{-\frac{1}{2\sigma^{2}}\sum_{a=1}^{u}\|p^{-1/2}\boldsymbol{\Phi} (\bbeta_0-\bbeta_a)+\Z\|^2}.
\end{align}
 Here and in what follows for convenience  we omit dimensional indexes $N$ and $p$ in $\Y_N$, $\Z_N$, $\boldsymbol{\Phi}_N$, etc. 
We compute
\begin{align}
    \ex \mathcal{Z}(\Y)^u= \int \prod_{a=0}^udP_0(\bbeta_a)dP_XdP_{Z}e^{-\frac{1}{2\sigma^{2}}\sum_{a=1}^{u}\sum_{\mu=1}^{N}(p^{-1/2}\sum_{j=1}^{p}x_{\mu,j}(\bbeta_0-\bbeta_a)_j+Z_{\mu})^2}.
\end{align}
Our goal is to calculate above expression as a Gaussian integral, and for this we introduce a new variable $v_{\mu.a}=p^{-1/2}\sum_{j=1}^{p}x_{\mu,j}(\bbeta_0-\bbeta_a)_j+Z_{\mu}$. One can easily see, that conditionally on $\bbeta$, vector $\mathbf{v}=(v_{\mu,a})_{\mu=1,a=1}^{N,u}$ is Gaussian, with zero mean and covariance matrix

\begin{multline}\label{eq:covar_elem}
        (\mathbf{R}_{\v})_{\mu\nu,ab}
        =\ex v_{\mu,a}v_{\nu.b}
        =\ex\Big(p^{-1/2}\sum_{j=1}^{p}x_{\mu,j}(\bbeta_0-\bbeta_a)_j+Z_{\mu}\Big)\Big(p^{-1/2}\sum_{l=1}^{p}x_{\nu,l}(\bbeta_0-\bbeta_b)_l+Z_{\nu}\Big)\\
        =\frac{1}{p}\sum_{j=1}^p(\bbeta_0-\bbeta_a)_j\lambda_j^{|\mu-\nu|}(1-\lambda_j^2)^{-1}(\bbeta_0-\bbeta_b)_j+\delta_{\mu,\nu}\sigma^{2},
        \end{multline}
        last equality being true due to (\ref{eq:cor_x_t^i}).
We consider two different cases: when $\bA$ is multiple of identity and when $\bA$ has finite number of different eigenvalues. While, obviously, first case is a part of the second, for the sake of clarity, we prefer to study the first case separately in more details and then transfer some results to \textbf{Case 2}. 
\subsection{Case 1: \texorpdfstring{$\bA=\lambda \mathbf{I}$}{A is multiply of identity}}

\textit{Assumption: Matrix $\bA$ is multiple of identity, i.e. $\lambda_i=\lambda$ for all $i$.}

We introduce $u\times u$ overlap matrix $\mathbf{Q}$ defined as $Q_{ab}=\frac{1}{p}\sum_{j=1}^p(\bbeta_0-\bbeta_a)_j(\bbeta_0-\bbeta_b)_j$. Then the covariance matrix becomes
\begin{align}
    (\mathbf{R}_{\v})_{\mu\nu,ab}=\ex v_{\mu,a}v_{\nu,b}=\lambda^{|\mu-\nu|}(1-\lambda^2)^{-1}Q_{ab}+\delta_{\mu,\nu}\sigma^{2}.
\end{align}
We can see covariance matrix $\mathbf{R}_{\v}$ of vector ${\v}$
as $Nu\times Nu$ block matrix with $N^2$ blocks of size $u\times u$. This means that each block is parametrised with $a$ and $b$ while $\mu$ and $\nu$ show the position of the block. Then  covariance matrix can be written as:
\begin{align}
    \mathbf{R}_{\v}=\begin{pmatrix}
       (1-\lambda^2)^{-1}\mathbf{Q}+\mathbbm{1}_u\sigma^{2} & (1-\lambda^2)^{-1}\lambda \mathbf{Q} & (1-\lambda^2)^{-1}\lambda^2\mathbf{Q} &\ldots & (1-\lambda^2)^{-1}\lambda^{N-1}\mathbf{Q}\\
        (1-\lambda^2)^{-1}\lambda \mathbf{Q}&  (1-\lambda^2)^{-1}\mathbf{Q}+\mathbbm{1}_u\sigma^{2} & (1-\lambda^2)^{-1}\lambda \mathbf{Q} &\ldots & (1-\lambda^2)^{-1}\lambda^{N-2}\mathbf{Q}\\
        \vdots&\ddots & \ddots & \ddots & \vdots\\
        (1-\lambda^2)^{-1}\lambda^{N-1}\mathbf{Q} & (1-\lambda^2)^{-1}\lambda^{N-2}\mathbf{Q} & (1-\lambda^2)^{-1}\lambda^{N-3}\mathbf{Q} &\ldots & (1-\lambda^2)^{-1}\mathbf{Q}+\mathbbm{1}_u\sigma^{2}
    \end{pmatrix}
\end{align}
where $\mathbbm{1}_u$ stands for $u\times u$ matrix with each element equal to $1$.

Then  after calculating the Gaussian integral (\ref{int:gaussian_v}) for fixed $\{\bbeta_a\}_{a=0}^u$ we get:
\begin{align}\label{eq:exp_Z^N_det}
    \ex \mathcal{Z}(\Y)^u=\int \prod_{a=0}^udP_0(\bbeta_a)\Big(\det(\mathbf{I}_{Nu}+\mathbf{R}_{\v}/\sigma^2)\Big)^{-1/2}.
\end{align}

In order to fix overlap $Q_{ab}$ we introduce the  Dirac delta function and rewrite it as the Fourier integral (i.e. $\delta(pQ_{ab}-(\bbeta_0-\bbeta_a)^\intercal(\bbeta_0-\bbeta_b))=\int d\hat{Q}_{ab}\exp\{i\hat{Q}_{ab}(pQ_{ab}-(\bbeta_0-\bbeta_a)^\intercal(\bbeta_0-\bbeta_b))\}$):
\begin{multline}\label{eq:int_Q_hat}
    \ex \mathcal{Z}(\Y)^u
    =\int \prod_{a=0}^udP_0(\bbeta_a)\prod_{a\leq b}dQ_{ab}\Big(\det(\mathbf{I}_{Nu}+\mathbf{R}_{\v}/\sigma^2)\Big)^{-1/2}\delta(pQ_{ab}-(\bbeta_0-\bbeta_a)^\intercal(\bbeta_0-\bbeta_b))\\
    =\int \prod_{a\leq b}dQ_{ab}d\hat{Q}_{ab}\Big(\det(\mathbf{I}_{Nu}+\mathbf{R}_{\v}/\sigma^2)\Big)^{-1/2}\prod_{a=0}^udP_0(\bbeta^a)
    \prod_{a\leq b}\exp\{\hat{Q}_{ab}(\bbeta_0-\bbeta_a)^\intercal(\bbeta_0-\bbeta_b)-p\hat{Q}_{ab}Q_{ab}\}\\
    = \int\prod_{a\leq b}dQ_{ab}d\hat{Q}_{ab}e^{-1/2\ln\det(\mathbf{I}_{Nu}+\mathbf{R}_{\v}/\sigma^2)-p\sum_{a\leq b}\hat{Q}_{ab}Q_{ab}}
    \prod_{a=0}^udP_0(\bbeta^a)e^{\sum_{a\leq b}\hat{Q}_{ab}(\bbeta_0-\bbeta_a)^\intercal(\bbeta_0-\bbeta_b)}.
\end{multline}

Before continuing with the transformation of the last expression we analyze $\det(\mathbf{I}_{uN}+\mathbf{R}_{\v}/\sigma^2)$. Under the RS hypothesis we  assume that all off diagonal and diagonal elements of matrix $Q$ are equal to $Q_{12}$ and $Q_{11}$, respectively. In this case $\sum_{a\leq b}\hat{Q}_{ab}Q_{ab}$ becomes just $u\hat{Q}_{11}Q_{11}+ u(u-1)\hat{Q}_{12}Q_{12}/2$, and $\mathbf{R}_{\v}$ can be rewritten as 
\begin{multline}
    \mathbf{R}_{\v}=\begin{pmatrix}
        (1-\lambda^2)^{-1}+\frac{\sigma^2}{Q_{12}}&(1-\lambda^2)^{-1}\lambda&\ldots&(1-\lambda^2)^{-1}\lambda^{N-1}\\
        (1-\lambda^2)^{-1}\lambda&(1-\lambda^2)^{-1}+\frac{\sigma^2}{Q_{12}}&\ldots&(1-\lambda^2)^{-1}\lambda^{N-2}\\
        \vdots&\ddots&\ddots&\vdots\\
        (1-\lambda^2)^{-1}\lambda^{N-1}&(1-\lambda^2)^{-1}\lambda^{N-2}&\ldots&(1-\lambda^2)^{-1}+\frac{\sigma^2}{Q_{12}}
    \end{pmatrix}\otimes \mathbf{Q}+\sigma^2(1-\frac{Q_{11}}{Q_{12}})\mathbf{I}_{Nu}\\
    =(\LLambda+\frac{\sigma^2}{Q_{12}}\mathbf{I}_N)\otimes\mathbf{Q}+\sigma^2(1-\frac{Q_{11}}{Q_{12}})\mathbf{I}_{Nu},
\end{multline} 
where $\LLambda$ is as in (\ref{eq:def_Lambda}). From this we have
\begin{align}\label{eq:det_v}
    \det(\mathbf{I}_{Nu}+\mathbf{R}_{\v}/\sigma^2)=
    \det\Big((\sigma^{-2}\LLambda+Q_{12}^{-1}\mathbf{I}_N)\otimes \mathbf{Q}
    +(2-\frac{Q_{11}}{Q_{12}})\mathbf{I}_{Nu}\Big).
\end{align}
To calculate this determinant it is sufficient to find the eigenvalues of matrices involved. The eigenvalues of matrix $\mathbf{Q}$ is well known: $(u-1)Q_{12}+Q_{11}$ with multiplicity 1 and $Q_{11}-Q_{12}$ with multiplicity $u-1$.
If we denote $\{\delta_i\}_{i=1,N}$ the eigenvalues of matrix 
\begin{align}\label{eq:def_Lambda1}
   \LLambda=(1-\lambda^2)^{-1}\begin{pmatrix}
    1&\lambda &\ldots&\lambda^{N-1}\\
    \lambda&1&\ldots&\lambda^{N-2}\\
    \vdots&\ddots&\ddots&\vdots\\
    \lambda^{N-1}&\lambda^{N-2}&\ldots&1
\end{pmatrix},
    \end{align}
then the eigenvalues of matrix $\mathbf{I}_{Nu}+\mathbf{R}_{\v}/\sigma^2$ are, for each $i=1,\ldots N$:
\begin{align*}
    (\delta_i/\sigma^2+1/Q_{12})((u-1)Q_{12}+Q_{11})+2-Q_{11}/Q_{12}&=\delta_i((u-1)Q_{12}+Q_{11})/\sigma^2+u+1 \text{ with multiplicity } 1\\
    (\delta_i/\sigma^2+1/Q_{12})(Q_{11}-Q_{12})+2-Q_{11}/Q_{12}&=\delta_i(Q_{11}-Q_{12})/\sigma^2+1 \text{ with multiplicity } u-1.
\end{align*} 
This allows us to find (\ref{eq:det_v}) explicitly:
\begin{align}
  \det(\mathbf{I}_{Nu}+\mathbf{R}_{\v}/\sigma^2)= \prod_{i=1}^N  (\delta_i((u-1)Q_{12}+Q_{11})/\sigma^2+u+1)(\delta_i(Q_{11}-Q_{12})/\sigma^2+1)^{u-1}.
\end{align}
Now we can jump back to the (\ref{eq:int_Q_hat}) and find the order of  $\ln\det$
\begin{multline}
   \lim_{N\rightarrow \infty,u\rightarrow0}\dfrac{1}{up}\ln(\det(\mathbf{I}_{Nu}+\mathbf{R}_{\v}/\sigma^2))^{-1/2}= \\
   =-\dfrac{1}{2}\lim_{N,u}\dfrac{1}{p}\ln\prod_{i=1}^N (\delta_i(Q_{11}-Q_{12})/\sigma^2+1)\Big(1+\dfrac{u\delta_iQ_{12}+u\sigma^2}{\delta_i(Q_{11}-Q_{12})+\sigma^2}\Big)^{1/u}.
\end{multline}
Regarding the limit with respect to $u$ we remark that  $(1+\frac{u\delta_iQ_{12}+u\sigma^2}{\delta_i(Q_{11}-Q_{12})+\sigma^2})^{1/u}\rightarrow e^{\frac{\delta_iQ_{12}+\sigma^2}{\delta_i(Q_{11}-Q_{12})+\sigma^2}}$ when $u\rightarrow0$, then
\begin{multline}
   \lim_{N,u}\dfrac{1}{up}\ln(\det(\mathbf{I}_{Nu}+\mathbf{R}_{\v}))^{-1/2}
   =-\dfrac{c}{2}\lim_{N}\Big[\dfrac{1}{N}\sum_{i=1}^N\ln (\delta_i(Q_{11}-Q_{12})/\sigma^2+1)+\dfrac{1}{N}\sum_{i=1}^N\dfrac{ \delta_iQ_{12}+\sigma^2}{\delta_i(Q_{11}-Q_{12})+\sigma^2}\Big].
\end{multline}
The resulting sums are, in fact, normalized eigenvalue statics of matrix $\LLambda$. In order to find their limits we notice that matrix of such form is well-studied Kac-Murdock-Szeg{\"o} matrix (see  for example \cite{grenandertoeplitz}, \cite{TRENCH2002251}), and the limit of its eigenvalue distribution, denoted by $\mu(\delta)$, is known. Thus we can write
\begin{multline}
   \lim_{N,u}\dfrac{1}{up}\ln(\det(\mathbf{I}_{Nu}+\mathbf{R}_{\v}/\sigma^2))^{-1/2}
   =\\
   -\dfrac{c}{2}\Big(\int_{\mathbb{R}}d\mu(\delta)\ln\left(\delta(Q_{11}-Q_{12})/\sigma^2+1\right)+\int_{\mathbb{R}}\dfrac{(\delta Q_{12}+\sigma^2)d\mu(\delta)}{\delta(Q_{11}-Q_{12})+\sigma^2}\Big).
\end{multline}

Returning to the expression (\ref{eq:int_Q_hat}), we note that we have dealt with the first exponent. To address the second one we first remark that $dP(\bbeta_a)=(dP(\beta_{i,a}))^p$ and conclude
\begin{align}
    \prod_{a=0}^udP_0(\bbeta_a)e^{\sum_{a\leq b}\hat{Q}_{ab}(\bbeta_0-\bbeta_a)^\intercal(\bbeta_0-\bbeta_b)}
    =\Big(    \prod_{a=0}^udP_0(\bbeta_{a})_i e^{\sum_{a\leq b}\hat{Q}_{ab}(\bbeta_{0}-\bbeta_{a})_i(\bbeta_{0}-\bbeta_{b})_i}
\Big)^p:=W^p.
\end{align}
In what follows we omit index $i$  for convenience, and consider $\beta$ and $\beta_0$ as scalars. Using a well known Hubbard-Stratanovich trick, we express the exponent as the Gaussian integral in order to decouple $\beta_a$ and $\beta_b$:
\begin{align}
    e^{1/2\hat{Q}_{12}\sum_{a\neq b}(\beta_0-\beta_a)^\intercal(\beta_0-\beta_b)}
    =\int Dze^{z\sqrt{\hat{Q}_{12}}\sum_{a=1}^u(\beta_0-\beta_a)-1/2\hat{Q}_{12}\sum_{a=1}^u(\beta_0-\beta_a)^2},
\end{align}
here $Dz$ is integration with respect to Gaussian standard measure. Then 
\begin{align}
    W=\int DzdP_0(\beta_0)\Big(dP_0(\beta)e^{z\sqrt{\hat{Q}_{12}}(\beta_0-\beta)-1/2(\hat{Q}_{12}-2\hat{Q}_{11})(\beta_0-\beta)^2}\Big)^u.
\end{align}
Next, we use the fact that $\int Dzf^u(z)\approx 1+u\int Dz\ln f(z)\approx e^{u\int Dz\ln f(z)}$ when $u\rightarrow 0$ and after
plugging this back to (\ref{eq:int_Q_hat}), we have
\begin{align}\label{eq:exp_Phi}
    \ex \mathcal{Z}(\Y)^u=\int dQ_{11}dQ_{12}d\hat{Q}_{11}\hat{Q}_{12}\exp\{-\frac{1}{2}up\mathcal{F}(Q_{11},Q_{12},\hat{Q}_{11},\hat{Q}_{12})+o(up)\},
\end{align}
where
\begin{multline}
   \mathcal{F}(Q_{11},Q_{12},\hat{Q}_{11},\hat{Q}_{12}) =c\int_{\mathbb{R}}d\mu(\delta)\ln\Big(\delta(Q_{11}-Q_{12})/\sigma^2+1\Big)+
   c\int_{\mathbb{R}}\dfrac{(\delta Q_{12}+\sigma^2)d\mu(\delta)}{\delta(Q_{11}-Q_{12})+\sigma^2}\\
   -Q_{12}\hat{Q}_{12}+2Q_{11}\hat{Q}_{11}
   -2\int DzdP_0(\beta_0)\ln \int dP_0(\beta)e^{z\sqrt{\hat{Q}_{12}}(\beta_0-\beta)-1/2(\hat{Q}_{12}-2\hat{Q}_{11})(\beta_0-\beta)^2}.
   \end{multline}
Taking derivatives with respect to parameters $Q_{11}$, $Q_{12}$, $\hat{Q}_{11}$, and $\hat{Q}_{12}$ we obtain the self-consistent equations for stationary points of $\mathcal{F}$
\begin{align}\label{eq:part_eq_lam}
    &\dfrac{\partial \mathcal{F}}{\partial Q_{11}}
    =\int_{\mathbb{R}} d\mu(\delta)\frac{c\delta^2(Q_{11}-2Q_{12})}{((Q_{11}-Q_{12})\delta+\sigma^2)^2}+2\hat{Q}_{11}\\
    &\dfrac{\partial \mathcal{F}}{\partial Q_{12}}=\int_{\mathbb{R}}d\mu(\delta)\frac{c\delta(\delta Q_{12}+\sigma^2)}{((Q_{11}-Q_{12})\delta+\sigma^2)^2}-\hat{Q}_{12}\\\label{eq:part_eq_hat1}
    &\dfrac{\partial\mathcal{F}}{\partial \hat{Q}_{11}}=2Q_{11}-2\int DzdP(\beta_0)\dfrac{\int 
            dP_0(\beta)e^{z\sqrt{\hat{Q}_{12}}(\beta_0-\beta)-1/2(\hat{Q}_{12}-2\hat{Q}_{11})(\beta_0-\beta)^2}(\beta_0-\beta)^2}{\int dP_0(\beta)e^{z\sqrt{\hat{Q}_{12}}(\beta_0-\beta)-1/2(\hat{Q}_{12}-2\hat{Q}_{11})(\beta_0-\beta)^2}} \\
    &\dfrac{\partial \mathcal{F}}{\partial \hat{Q}_{12}}=  -Q_{12}-\\\label{eq:part_eq_hat2}
    &2\int DzdP_0(\beta_0)\dfrac{\int 
            dP_0(\beta)1/2e^{z\sqrt{\hat{Q}_{12}}(\beta_0-\beta)-1/2(\hat{Q}_{12}-2\hat{Q}_{11})(\beta_0-\beta)^2}(z(\beta_0-\beta)(\hat{Q}_{12})^{-1/2}-(\beta_0-\beta)^2)}{\int dP_0(\beta)e^{z\sqrt{\hat{Q}_{12}}(\beta_0-\beta)-1/2(\hat{Q}_{12}-2\hat{Q}_{11})(\beta_0-\beta)^2}}\nonumber.
\end{align}      
We recall that the physical meaning of the order parameters is
   \begin{align}
       Q_{12}&=\frac{1}{p}\ex\langle(\bbeta^0-\bbeta^1)^\intercal(\bbeta^0-\bbeta^2)\rangle\overset{N}{=}\frac{1}{p}\ex(\|\bbeta_0\|^2-\bbeta_0^\intercal\langle\bbeta_1\rangle)\\
       Q_{11}&=\frac{1}{p}\ex\langle(\bbeta^0-\bbeta^1)^\intercal(\bbeta^0-\bbeta^1)\rangle\overset{N}{=}\frac{2}{p}\ex(\|\bbeta_0\|^2-\bbeta_0^\intercal\langle\bbeta_1\rangle).
   \end{align}
The right hand sides were obtained using Nishimori identity and  the assumption that the overlap concentrates, from here  one can easily see that $Q_{11}=2Q_{12}$. It means that in order to calculate stationary points of (\ref{eq:exp_Phi}) we can limit ourselves to the set of points such that $Q_{11}=2Q_{12}$, which will simplify the equations (\ref{eq:part_eq_lam}-\ref{eq:part_eq_hat2}).
One can immediately see that  $\hat{Q}_{11}=0$ and the third equation is not needed. Also after integration by part with respect to $z$ the last equation can be rewritten as
\begin{align}
    Q_{12}=\int DzdP_0(\beta_0)\Big(\dfrac{\int 
            dP_0(\beta)e^{z\sqrt{\hat{Q}_{12}}(\beta_0-\beta)-1/2(\hat{Q}_{12}-2\hat{Q}_{11})(\beta_0-\beta)^2}(\beta_0-\beta)}{\int dP_0(\beta)e^{z\sqrt{\hat{Q}_{12}}(\beta_0-\beta)-1/2(\hat{Q}_{12}-2\hat{Q}_{11})(\beta_0-\beta)^2}}\Big)^2.
\end{align}
It is worth to mention that right hand side of the last equation can be rewritten as $Q_{12}=\ex (\beta_0-\langle\beta\rangle)^2$, where Gibbs brackets correspond to the posterior distribution associated to a scalar channel $Y=\sqrt{\hat{Q}_{12}}\beta_0+Z$, $\beta_0\sim P_0(\cdot)$, $Z\sim\mathcal{N}(0,1)$. In other words, $Q_{12}=\text{mmse}(\beta|\sqrt{\hat{Q}_{12}}\beta+Z)$.

This leaves us with only two simple equation
\begin{align}
  \hat{Q}_{12} &=\int_{\mathbb{R}}d\mu(\delta)\dfrac{c\delta}{Q_{12}\delta+\sigma^2}\label{eq:Qhat_12},\\
Q_{12}&=\text{mmse}(\beta|\sqrt{\hat{Q}_{12}}\beta+Z).\nonumber
\end{align}
To simplify (\ref{eq:Qhat_12}) we use properties of Kac-Murdock-Szeg{\"o} (\ref{eq:def_Lambda1}) matrix which are  described in details in the literature,  we need only to know that for any nice function $F$ we have
\begin{align}
    \int d\mu(\delta)F(\delta)=\frac{1}{\pi}\int_0^{\pi}F\Big(\frac{1}{1-2\lambda\cos \theta+\lambda^2}\Big)d\theta.
\end{align}
If we plug it to (\ref{eq:Qhat_12}) and calculate the integral explicitly, we obtain
\begin{align}
\hat{Q}_{12}=\dfrac{c}{\sqrt{(Q_{12}+(1-\lambda)^2\sigma^2)(Q_{12}+(1+\lambda)^2\sigma^2)}}.
\end{align}
After some simplifications  one can see that if $\lambda=0$ we recover the known formulas for model $\Y_N=\frac{1}{\sqrt{p}}\G_N\bbeta+\boldsymbol{Z}_N$ with $\G_N$ consisting of i.\,i.\,d. standard Gaussian elements (see e.\,g. \cite{tanaka2002statistical}, \cite{Mezard_Parisi_Virasoro}).




\subsection{Case 2: General diagonal \texorpdfstring{$\bA_p$}{A} with finite number of different eigenvalues}
Now we consider the more challenging part, i.e. when $\bA_p$ satisfies Assumption \ref{Assum:A}. As was mentioned before, we  consider matrices  $\bA_p$ with  finite number of different eigenvalues, denoted by $k$.

It is easy to see that nothing changes till (\ref{eq:covar_elem}). Since now we have multiple $\lambda_j$,  we are forced to have more  overlap parameters (and more equations in the future). That is, we introduce $k$ matrices $\{\mathbf{Q}^{(i)}\}_{i=1}^{k}$ of size $u\times u$ which elements are defined as $Q^{(i)}_{ab}=\frac{1}{|I_i|}\sum_{j\in I_i}(\bbeta_0-\bbeta_a)_j(\bbeta_0-\bbeta_b)_j$. Here $\{I_i\}_{i=1}^k$ are the sets of indexes $I_i=\{j:A_{jj}=\lambda_i\}$. Then, covariance matrix $\mathbf{R}_{\v}$ becomes
\begin{align}
    (\mathbf{R}_{\v})_{\mu\nu}^{ab}=\sum_{i=1}^kl_iQ_{ab}^{(i)}\sigma^{-2}\lambda_i^{|\mu-\nu|}(1-\lambda_i^2)^{-1}+\delta_{\mu,\nu}\sigma^2.
\end{align} 
Similarly to Case 1, we can write $\mathbf{R}_{\v}$ as a block matrix
\begin{align}
    \mathbf{R}_{\v}=\sum_{i=1}^k l_i\LLambda_{i,N}\otimes \mathbf{Q}^{(i)}+\sigma^2\mathbf{I}_N\otimes\mathbbm{1}_u,
\end{align} 
where $\LLambda_{i,N}$ are defined as in (\ref{eq:def_Lambda}).
Following the same steps as in the previous section, one can easily find that
\begin{multline}
        \ex \mathcal{Z}(\Y)^u
    = \int\prod_{i,a\leq b}dQ_{ab}^{(i)}d\hat{Q}^{(i)}_{ab}e^{-1/2\ln\det(\mathbf{I}_{Nu}+\mathbf{R}_{\v}/\sigma^2)-\sum_i\sum_{a\leq b}|I_i|\hat{Q}^{(i)}_{ab}Q^{(i)}_{ab}}\times\\
    \prod_{a=0}^udP_0(\bbeta^a)\prod_ie^{\sum_{a\leq b}\hat{Q}^{(i)}_{ab}\sum_{j\in I_i}(\bbeta_0-\bbeta_a)_j(\bbeta_0-\bbeta_b)_j}.
\end{multline}
The second exponent can be handled as in Case 1, so here we focus on finding the limit of $\frac{1}{pu}\ln\det(\mathbf{I}_{Nu}+\mathbf{R}_{\v}\sigma^2)$. As before, we assume replica symmetric ansatz, i.\,e. $Q_{ab}^{(i)}=Q_{12}^{(i)}$ for off diagonal and $Q_{aa}^{(i)}=Q_{11}^{(i)}$ for diagonal. Then all matrices $\mathbf{Q}^{(i)}$ share the same eigenvector space $(\mathbf{U})$ and $\mathbf{Q}^{(i)}=\mathbf{U}\mathbf{D}^{(i)}\mathbf{U}^\intercal$, where $\mathbf{D}^{(i)}=diag\{uQ^{(i)}_{12}+Q^{(i)}_{11}-Q^{(i)}_{12},Q^{(i)}_{11}-Q^{(i)}_{12},\ldots,Q^{(i)}_{11}-Q^{(i)}_{12}\}$. Also, since  matrix $\mathbbm{1}_u$ has the same structure, it can be decomposed in the same fashion $\mathbbm{1}_u=\mathbf{U}u\mathbf{e_1e_1}^\intercal \mathbf{U}^\intercal$, where $\mathbf{e_1}=(1,0,\ldots,0)^\intercal$. Using the basic properties of the Kronecker product we decompose each term as $\LLambda_{i,N}\otimes \mathbf{U}\mathbf{D}^{(i)}\mathbf{U}^\intercal=(\mathbf{I}_N\otimes \mathbf{U})(\LLambda_{i,N}\otimes \mathbf{D}^{(i)})(\mathbf{I}_N\otimes \mathbf{U}^\intercal)$. This gives us the opportunity to rewrite
\begin{multline}
    \mathbf{I}_{Nu}+\mathbf{R}_{\v}/\sigma^2=\sum_{i=1}^k\frac{l_i}{\sigma^2}\LLambda_{i,N}\otimes \mathbf{U}\mathbf{D}^{(i)}\mathbf{U}^\intercal+\mathbf{I}_N\otimes \mathbf{U} (u\mathbf{e_1e_1}^\intercal+\mathbf{I}_u) \mathbf{U}^\intercal=\\
    (\mathbf{I}_N\otimes \mathbf{U})(\sum_{i=1}^k\frac{l_i}{\sigma^2}\LLambda_{i,N}\otimes \mathbf{D}^{(i)} +\mathbf{I}_N\otimes (u\mathbf{e_1e_1}^\intercal+\mathbf{I}_u))(\mathbf{I}_N\otimes \mathbf{U}^\intercal)
\end{multline}
Since we are interested only in determinant of this expression, we can disregard factors $(\mathbf{I}_N\otimes \mathbf{U})$ and $(\mathbf{I}_N\otimes \mathbf{U}^\intercal)$, as their product gives identity matrix. Finally, it is easy to see that  Kronecker product is commutative up to a permutation of rows and columns, which gives us

\begin{multline}
    \det(\mathbf{I}_{Nu}+\mathbf{R}_{\v}/\sigma^2)=\det(\sum_{i=1}^k\frac{l_i}{\sigma^2}\mathbf{D}^{(i)}\otimes \LLambda_{i,N} + (u\mathbf{e_1e_1}^\intercal+\mathbf{I}_u)\otimes \mathbf{I}_N)\\
    =\det(\sum_{i=1}^k\frac{l_i}{\sigma^2}(Q^{(i)}_{11}-Q^{(i)}_{12})\LLambda_{i,N} + \mathbf{I}_N)^u\det(\mathbf{I}_N+(\sum_{i=1}^k\frac{ul_i}{\sigma^2}Q^{(i)}_{12}\LLambda_{i,N} + u\mathbf{I}_N)(\sum_{i=1}^k\frac{l_i}{\sigma^2}(Q^{(i)}_{11}-Q^{(i)}_{12})\LLambda_{i,N} + \mathbf{I}_N)^{-1}),
\end{multline}
the last equality being true because all terms are block-diagonal matrices. Now we can easily find the limit
\begin{multline}
    \lim_{p,u}\frac{1}{pu}\ln(\det(\mathbf{I}_{Nu}+\mathbf{R}_{\v}/\sigma^2))^{-1/2}
   =-\frac{c}{2}\lim_{N}\Big[\frac{1}{N}Tr\ln (\sum_{i=1}^k\frac{l_i}{\sigma^2}\LLambda_{i,N}(Q^{(i)}_{11}-Q^{(i)}_{12})+\mathbf{I}_N)\\
   +\frac{1}{N}Tr(\sum_{i=1}^k\frac{l_i}{\sigma^2}\LLambda_{i,N}Q^{(i)}_{12}+\mathbf{I}_N)( \sum_{i=1}^k\frac{l_i}{\sigma^2}\LLambda_{i,N}(Q^{(i)}_{11}-Q^{(i)}_{12})+\mathbf{I}_N)^{-1}\Big].\label{eigen_KMS_mid}
\end{multline}
Here we have limiting eigenvalue statistics of two matrices, each of them consists of the sums of KMS matrices, we mention that KMS matrices asymptotically share the same eigenvector space (for more details we again refer to \cite{grenandertoeplitz}), so in the limit we have
\begin{multline}
        \lim_{p,u}\frac{1}{pu}\ln(\det(\mathbf{I}_{Nu}+\mathbf{R}_{\v}/\sigma^2))^{-1/2}\\
   =-\frac{c}{2\pi}\Big[\int_{0}^{\pi}\ln (\sum_{i=1}^kl_i\delta_i(\theta)(Q^{(i)}_{11}-Q^{(i)}_{12})/\sigma^2+1)d\theta+\int_{0}^{\pi}\frac{ \sum_{i=1}^kl_i\delta_i(\theta)Q^{(i)}_{12}/\sigma^2+1}{ \sum_{i=1}^kl_i\delta_i(\theta)(Q^{(i)}_{11}-Q^{(i)}_{12})/\sigma^2+1}d\theta\Big],\label{eigen_KMS}
\end{multline}
where we denote
    \begin{align}
        \delta_i(\theta)=\dfrac{1}{1-2\lambda_i\cos \theta+\lambda_i^2}.
    \end{align}
Returning to $\ex \mathcal{Z}(\Y)^u$, we obtain (using the same steps as Case 1)
 \begin{align}
    \ex \mathcal{Z}(\Y)^u=\int \prod_i dQ^{(i)}_{11}dQ^{(i)}_{12}d\hat{Q}^{(i)}_{11}\hat{Q}^{(i)}_{12}\exp\{-\frac{1}{2}up\mathcal{F}(Q^{(i)}_{11},Q^{(i)}_{12},\hat{Q}^{(i)}_{11},\hat{Q}^{(i)}_{12})+o(up)\},
\end{align}
where
\begin{multline}
   \mathcal{F}(Q^{(i)}_{11},Q^{(i)}_{12},\hat{Q}^{(i)}_{11},\hat{Q}^{(i)}_{12}) =\frac{c}{\pi}\int_{0}^{\pi}\ln (\sum_{i=1}^kl_i\delta_i(\theta)(Q^{(i)}_{11}-Q^{(i)}_{12})/\sigma^2+1)d\theta\\
   +\frac{c}{\pi}\int_{0}^{\pi}\frac{ \sum_{i=1}^kl_i\delta_i(\theta)Q^{(i)}_{12}/\sigma^2+1}{ \sum_{i=1}^kl_i\delta_i(\theta)(Q^{(i)}_{11}-Q^{(i)}_{12})/\sigma^2+1}d\theta
   -\sum_{i=1}^kl_iQ^{(i)}_{12}\hat{Q}^{(i)}_{12}+2\sum_{i=1}^kl_iQ^{(i)}_{11}\hat{Q}^{(i)}_{11}\\
   -2\sum_{i=1}^kl_i\int DzdP(\beta_0)\ln \int dP_0(\beta)e^{z\sqrt{\hat{Q}^{(i)}_{12}}(\beta_0-\beta)-1/2(\hat{Q}^{(i)}_{12}-2\hat{Q}^{(i)}_{11})(\beta_0-\beta)^2}.
   \end{multline}
   Similarly to Case 1 we use saddle point method and take derivatives with respect to all parameters. One can easily verify that derivatives with respect to $\hat{Q}^{(i)}_{12}$ and $\hat{Q}^{(i)}_{12}$ stay identical to (\ref{eq:part_eq_hat1}-\ref{eq:part_eq_hat2}), for $Q^{(i)}_{12}$ and $Q^{(i)}_{11}$ we have
   \begin{align}\label{eq:part_der_11_diag}
        &\dfrac{\partial \mathcal{F}}{\partial Q^{(i)}_{11}}
    =\frac{c}{\pi}\int_{0}^{\pi} \frac{l_i\delta_i(\theta)(\sum_{j=1}^kl_j(Q^{(j)}_{11}-2Q^{(j)}_{12})\delta_j(\theta)/\sigma^4)}{(\sum_{j=1}^k(Q^{(j)}_{11}-Q^{(j)}_{12})l_j\delta_j(\theta)/\sigma^2+1)^2}d\theta+2l_i\hat{Q}^{(i)}_{11}\\
    &\dfrac{\partial \mathcal{F}}{\partial Q^{(i)}_{12}}=\frac{c}{\pi}\int_{0}^{\pi}\frac{l_i\delta_i(\theta)/\sigma^2(\sum_{j=1}^kl_j\delta_j(\theta) Q^{(j)}_{12}/\sigma^2+1)}{(\sum_{j=1}^k(Q^{(j)}_{11}-Q^{(j)}_{12})l_j\delta_j(\theta)/\sigma^2+1)^2}d\theta-l_i\hat{Q}^{(i)}_{12}.
   \end{align}
   
  Using again Nishimori identity one can show that sought-after   parameters should satisfy $Q^{(i)}_{11}=2Q^{(i)}_{12}$ for all $i$. As before this allows us to simplify the system for stationary points, from (\ref{eq:part_der_11_diag}) we obtain immediately that $\hat{Q}^{(i)}_{11}=0$ and we can disregard equation for $Q_{11}^{(i)}$. This leaves us with the
self-consistent system of $2k$ equations:
\begin{align}
    \hat{Q}^{(i)}_{12}&
    =\frac{c}{\pi}\int_{0}^\pi\frac{\delta_i(\theta)d\theta}{\sum_jQ^{(j)}_{12}l_j\delta_j(\theta)+\sigma^2}\\\label{eq:fixed_point_k_lam}
    Q^{(i)}_{12}&=  \int DzdP_0(\beta_0)\Big(\dfrac{\int 
            dP_0(\beta)e^{z\sqrt{\hat{Q}^{(i)}_{12}}(\beta_0-\beta)-1/2\hat{Q}^{(i)}_{12}(\beta_0-\beta)^2}(\beta_0-\beta)}{\int dP_0(\beta)e^{z\sqrt{\hat{Q}^{(i)}_{12}}(\beta_0-\beta)-1/2\hat{Q}^{(i)}_{12}(\beta_0-\beta)^2}}\Big)^2,\quad i=1,\ldots,k.\nonumber
\end{align}
We again remark that the right hand side of the last equation is $\text{mmse}(\beta|\sqrt{\hat{Q}^{(i)}_{12}}\beta+Z)$ with $\beta\sim P_0$ and $Z\sim\mathcal{N}(0,1)$.
Wrapping up the calculation we are ready to introduce the final replica formula for the limiting free energy
\begin{multline}
    \bar{f}:=\lim_{p\rightarrow+\infty}\ex f_p=\inf_{\Gamma}\Big[\frac{c}{2\pi}\int_{0}^{\pi}\ln (\sum_{i=1}^kl_i\delta_i(\theta)r_{2,i}/\sigma^2+1)d\theta
   -\frac{1}{2}\sum_{i=1}^kl_ir_{2,i}r_{1,i}\\
   +\sum_{i=1}^kl_iI(\beta;\sqrt{r_{1,i}}\beta+Z)+\frac{c+1}{2}\Big],
\end{multline}
where $\Gamma$ is the set of critical points of the potential: 
\begin{align}
    \Gamma=\{(\mathbf{r}_1, \mathbf{r}_2)\in\mathbb{R}_+^k\times[0,\rho]^k|\forall\, i=1,\ldots,k:\, r_{2,i}=\text{mmse}(\beta|\sqrt{r_{1,i}}\beta+Z),\, \mathbf{r}_1=\mathcal{A}(\mathbf{r}_2)\}.
\end{align}

Finally we state the replica conjecture, for the limit of the mutual information.The  replica symmetric potential is defined as
\begin{align}\label{eq:replica_potential}
    i_{\rm RS}(\mathbf{r}_1,\mathbf{r}_2)=\frac{c}{2\pi}\int_{0}^{\pi}\ln (\sum_{i=1}^kl_i\delta_i(\theta)r_{2,i}/\sigma^2+1)d\theta
   -\frac{1}{2}\sum_{i=1}^kl_ir_{2,i}r_{1,i}
   +\sum_{i=1}^kl_iI(\beta;\sqrt{r_{1,i}}\beta+Z),
\end{align}
 where $\mathbf{r}_1=(r_{1,1},\ldots,r_{1,k})$ and $\mathbf{r}_2=(r_{2,1},\ldots,r_{2,k})$ are $k$~--dimensional vectors of parameters.
We have 
\begin{align}
    \lim_{p\rightarrow+\infty}i_p=\inf_{\mathbf{r}_1\in\mathbb{R}_+^k}\sup_{\mathbf{r}_2\in[0,\rho]^k}i_{\rm RS}(\mathbf{r}_1,\mathbf{r}_2).
    \end{align}

\section{Kac-Murdock-Szeg\"{o} matrices}\label{apx:KMS}
The Kac-Murdock-Szeg\"{o} (KMS) matrix is the special class of  symmetric Toeplitz matrices of the general form 
\begin{align}
\mathbf{ A}_{N,\rho} = \begin{bmatrix}    1          & \rho       & \rho^2 & \dots  & \rho^{N-1} \\    \rho       & 1          & \rho   & \dots  & \rho^{N-2} \\    \rho^2     & \rho       & 1      & \ddots & \vdots     \\    \vdots     & \vdots     & \ddots & \ddots & \rho       \\    \rho^{N-1} & \rho^{N-2} & \dots  & \rho   & 1 \end{bmatrix} \in\mathbb{R}^{N\times N} \qquad(1) 
\end{align}
the term originating from the 1953
paper \cite{Kac1953OnTE} by Kac, Murdock, and Szeg\"{o}. In \cite{Kac1953OnTE} (see also the book \cite{grenandertoeplitz} by Grenander and Szeg\"{o}), a detailed study of
such matrices was an important first step for the development of general theorems on the
asymptotic (large- $n$ ) behaviour of eigenvalues and determinants of Toeplitz matrices. The entries of such matrix come from a Fourier coefficients of the  function
\begin{align}
    F_\rho(\theta)=\frac{1-\rho^2}{1-2\rho\cos\theta+\rho^2}=\sum_{n=-\infty}^{+\infty}\rho^{|n|}e^{in\theta},
\end{align}
and because of this such function also sometimes called  generating function of $\bA_{N,\rho}$. From \cite{grenandertoeplitz}, the eigenvalues of $\lambda_{1}(\bA_{N,\rho})< \lambda_{2}(\bA_{N,\rho})<\ldots<\lambda_{N}(\bA_{N,\rho})$ are
\begin{align}
    \lambda_{i}(\bA_{N,\rho})=F_\rho(\theta_{i,N}),
\end{align}
where $\frac{(i-1)\pi}{N+1}<\theta_{i,N}<\frac{i\pi}{N+1}$, for $i=1,\ldots,N$. As $N\rightarrow\infty$, for any continuous function $f$ we have
\begin{align}
    \lim_{N\rightarrow\infty}\frac{1}{N}\sum_{i=1}^Nf(\lambda_{i}(\bA_{N,\rho}))=\frac{1}{\pi}\int_{0}^\pi f(F_\rho(\theta))d\theta.
\end{align}
And as consequence, the sequence of eigenvalues $\{\lambda_k(\bA_{N,\rho}\}$ is asymptotically equally distributed to  the sequence $\{F_{\rho}(2\pi k/N)\}$.
The same is true for the functions of several KMS matrices. An important property, that we mentioning throughout the proofs is that even thought for fixed $N$ eigenvectors of $\bA_N(\rho)$ are different for different values of $\rho$, in the limit  it does not depend on $\rho$. More precisely, if we take the set of KMS matrix $\{ \bA_N(\rho_1),\bA_N(\rho_1),\ldots,\bA_N(\rho_m)$, then for a continuous function $f$, the distribution of eigenvalues of $f(\bA_{N,\rho_1},\ldots,\bA_{N,\rho_m})$ is asymptotically equivalent to distribution of $\{f(F_{\rho_1}(2\pi k/N),\ldots,F_{\rho_m}(2\pi k/N))\}_{k=1}^N$.

\section{Nishimori identity}\label{apx:Nish}
In this paper we  place ourselves  in Bayesian-optimal setting which allows us to use so called Nishimori identities to simplify significantly some parts of calculations.   They were initially discovered in the context of the gauge theory of spin glasses \cite{nishimori01}.
To introduce them, we consider a generic inference problem where a Bayes-optimal statistician observes $\Y$ that is a random function of some ground truth signal $\bbeta^*$: $\Y\sim P_{y|\beta}(\bbeta^*)$. Then the following holds:
\begin{proposition}[Nishimori identity]
For any bounded function $f$ of the signal $\bbeta^*$, the data $\Y$ and of conditionally i.i.d. samples from the posterior $\bbeta^j\sim P_{\beta| y}(\,\cdot | \Y)$, $j=1,2,\ldots,n$, we have that
\begin{align}
    \ex\langle f(\Y,\bbeta^*,\bbeta^2,\ldots,\bbeta^{n})\rangle=\ex\langle f(\Y,\bbeta^1,\bbeta^2,\ldots,\bbeta^{n})\rangle
\end{align}
where the bracket notation $\langle \,\cdot\,\rangle$ is used for the joint expectation over the posterior samples $(\bbeta^j)_{j\le n}$, $\ex$ is over the signal $\bbeta^*$ and data $\Y$.
\end{proposition}

\begin{proof}
    The proof follows directly from Bayes' rule. An elementary proof can be found in \cite{Lelarge2017FundamentalLO}.
\end{proof}

\section{Concentration of the free energy}\label{apx:concen_z}
\begin{proposition}[Concentration of the free energy]\label{prop:concent_free_energy}
    \begin{equation}\ex\Big[\Big|\frac{1}{p}\log \mathcal{Z}-\ex\Big[\frac{1}{p}\log \mathcal{Z}\Big]\Big|^2\Big]\leq Cp^{-1}\label{free_conc_1}\end{equation}
\end{proposition}

The proof follows by hierarchical concentration results. In particular 
 in the first step we prove concentration inequality when expectation is taken conditional on the variables $V^{(i)}_j$s, $G^{(i)}_{jk}$s, and $\beta^{(i)}_{0,j}$s. In the next step we prove the concentration inequality when the expectation is taken 
 conditional on $\beta^{(i)}_{0.j}$, while in the final step we  prove concentration around the unconditional expectation. Thus Lemma \ref{free_conc_1} follows from  the Lemmas \ref{free_conc_2}, \ref{free_conc_3}, and \ref{free_conc_4}. 
Before diving into the calculations, we recall two classical concentration results, which will be the core of the proofs. Their proofs can be found in \cite{Boucheron2004}, Chapter 3.

\begin{lemma}[Poincar\'e-Nash inequality]\label{lem:P-N}
    Let $g:\mathbb{R}^p\rightarrow \mathbb{R}$ be a continuously differentiable function and $\mathbf{X}=(X_1,\ldots,X_p)$ be a random vector with i.i.d. $\mathcal{N}(0,1)$ entries. Then the following holds:
    $$\mathrm{Var}(g(\mathbf{X}))\leq \ex[\|\nabla g\|^2].$$
\end{lemma}

\begin{lemma}[Bounded difference]
    Let $g:\mathbb{R}^p\rightarrow \mathbb{R}$ be a function that satisfies the bounded difference property:
    $$\sup_{u_1,\ldots, u_p,u'_i}|g(u_1,\ldots,u_i,\ldots,u_p)-g(u_1,\ldots,u'_i,\ldots,u_p)|\leq c_i, \quad 1\leq i\leq p.$$
    If $\mathbf{X}=(X_1,\ldots,X_p)$ be a random vector with i.i.d. $\mathcal{N}(0,1)$ entries. Then the following holds:
    $$\mathrm{Var}(g(\mathbf{X}))\leq \frac{1}{4}\sum_{i=1}^p c^2_i.$$
\end{lemma}

Let $\mathbf{W}_1=(\Z,\tilde{\Z})$. As a first step, we prove the following lemma
\begin{lemma}
    $$\ex\Big[\left|\frac{1}{p}\log \mathcal{Z}-\ex_{\mathbf{W}_1}\Big[\frac{1}{p}\log \mathcal{Z}\right]\Big|^2\Big]\leq Cp^{-1}.$$ \label{free_conc_2}
\end{lemma}

\begin{proof}
    Let $g(\mathbf{W}_1)=\log \mathcal{Z}$. The partial derivatives with respect to $Z_{\mu}$ and $\tilde{Z}^{(i)}_j$s are given by:
        \begin{align*}
        \frac{\partial g}{\partial Z_{\mu}}&= -\Big\langle Z_{\mu}+\sqrt{\frac{1-t}{p}}\Big(\boldsymbol{\Phi}\Big(\bbeta_0-\bbeta\Big)\Big)_{\mu}+\sum_{i=1}^k\sqrt{R_{2,i}(t)l_i}\Big(\LLambda^{(i)}\Big(\V^{(i)}-\v^{(i)}\Big)\Big)_{\mu}\Big\rangle,\\
        \frac{\partial g}{\partial \tilde{Z}^{(i)}_j}&= -\Big\langle \tilde{Z}^{(i)}_j+\sqrt{R_{1,i}(t)}(\bbeta_{0}^{(i)}-\bbeta^{(i)})_j\Big\rangle. \\
    \end{align*}
Using Jensen inequality we obtain
\begin{align*}
        \ex\Big[\Big(\frac{\partial g}{\partial Z_{\mu}}\Big)^2\Big]&\leq  C\ex\Big[\Big\langle Z^2_{\mu}+\frac{1-t}{p}\Big(\boldsymbol{\Phi}\Big(\bbeta_0-\bbeta\Big)\Big)^2_{\mu}+\sum_{i=1}^kR_{2,i}(t)l_i\Big(\LLambda^{(i)}\Big(\V^{(i)}-\v^{(i)}\Big)\Big)^2_{\mu}\Big\rangle\Big],\\
        \ex\Big[\Big(\frac{\partial g}{\partial \tilde{Z}^{(i)}_j}\Big)^2\Big]&\leq C\ex\Big[\langle (\tilde{Z}^{(i)}_j)^2+R_{1,i}(t)(\beta_{0,j}^{(i)}-\beta_j^{(i)})^2\rangle\Big]. \\
    \end{align*}
    Combining all above  and applying  Nishimori identity, we have 
\begin{align*}
    \ex\|\nabla g\|^2&=\ex\Big[\sum_{\mu=1}^{N}\Big(\frac{\partial g}{\partial Z_\mu}\Big)^2+\sum_{i=1}^{k}\sum_{j=1}^{|I_i|}\Big(\frac{\partial g}{\partial \tilde{Z}^{(i)}_j}\Big)^2\Big]\\
    & \leq C \ex \Big[\|\Z\|^2+\frac{1}{p}\|\boldsymbol{\Phi}\|^2\|\bbeta_0\|^2+\sum_{i=1}^kR_{2,i}(t)l_i\|\LLambda^{(i)}\|^2\|\V^{(i)}\|^2+\sum_{i=1}^k(\|\tilde{\Z}^{(i)}\|^2+R^{2}_{1,i}(t)\|\bbeta_0^{(i)}\|^2)\Big]\\
    &\leq Cp.
\end{align*}
The last inequality holds since all the quantities $\ex[\|\Z\|^2], \ex[\|\tilde{\Z}^{(i)}\|^2],\ex[\|\V^{(i)}\|^2]$, and $\ex[\|\bbeta^{(i)}_0\|^2]$ are $O(p)$. The operator norm of $\boldsymbol{\Phi}$ also can be bounded us
$$ \ex\|\boldsymbol{\Phi}\|^2=\ex\|\boldsymbol{\Phi}\boldsymbol{\Phi}^\intercal\|\leq \sum_{i=1}^k\|\LLambda_{i,N}\|\ex[\|\G_i\G_i^\intercal\|]=O(p).$$
From Gaussian Poincar\'e inequality (Lemma~\ref{lem:P-N}) it follows that 
$$\ex\Big[\Big|\frac{1}{p}\log \mathcal{Z}-\ex_{\mathbf{W_1}}\Big[\frac{1}{p}\log \mathcal{Z}\Big]\Big|^2\Big]=p^{-2}\mathrm{Var}(g)\leq p^{-2}\ex\|\nabla g\|^2=Cp^{-1}.$$
\end{proof}

Let $\mathbf{W_2}=(\Z,\tilde{\Z},\G,\V)$ where $\G$ is the ensemble of Gaussian random variables in $\boldsymbol{\Phi}$.  In the next step we prove that 
\begin{lemma}\label{free_conc_3}
    $$\ex\Big[\Big|\ex_{\mathbf{W_1}}\Big[\frac{1}{p}\log \mathcal{Z}\Big]-\ex_{\mathbf{W_2}}\Big[\frac{1}{p}\log \mathcal{Z}\Big]\Big|^2\Big]\leq Cp^{-1}.$$
\end{lemma}

\begin{proof}
    Let $h(\G,\V)=\ex_{\mathbf{W_1}}\Big[\log \mathcal{Z}\Big ]$. We have
    \begin{align*}
        \frac{\partial h}{\partial V^{(j)}_m}= & -\ex_{\mathbf{W_1}}\Big[\Big\langle \sum_{\mu=1}^N\Big(Z_{\mu}+\sqrt{\frac{1-t}{p}}\Big(\boldsymbol{\Phi}(\bbeta_0-\bbeta)\Big)_{\mu}+\sum_{i=1}^k\sqrt{R_{2,i}(t)l_i}\Big(\LLambda^{(i)}(\V^{(i)}-\v^{(i)})\Big)_{\mu}\Big)\sqrt{R_{2,j}(t)l_j}\Lambda^{(j)}_{\mu m}\Big\rangle\Big]\\
        \frac{\partial h}{\partial G^{(j)}_{sm}}=\\
        -\ex_{\mathbf{W_1}}\Big[\Big\langle&  \sum_{\mu=1}^N\Big(Z_{\mu}+\sqrt{\frac{1-t}{p}}\Big(\boldsymbol{\Phi}(\bbeta_0-\bbeta)\Big)_{\mu}+\sum_{i=1}^k\sqrt{R_{2,i}(t)l_i}\Big(\LLambda^{(i)}(\V^{(i)}-\v^{(i)})\Big)_{\mu}\Big)\sqrt{\frac{1-t}{p}}\Lambda^{(j)}_{\mu s}(\bbeta_0^{(j)}-\bbeta^{(j)})_m\Big\rangle\Big]\\
        :=& I^1_{j,sm}+I^2_{j,sm}+I^3_{j,sm},\\
    \end{align*}
where
\begin{align*}
    I^1_{j,sm}&= \sqrt{\frac{1-t}{p}}\ex_{\mathbf{W_1}}\Big[\langle (\Z^\intercal \LLambda^{(j)})_{s}(\bbeta_0^{(j)}-\bbeta^{(j)})_m\rangle\Big],\\
    I^2_{j,sm}&= \frac{1-t}{p}\ex_{\mathbf{W_1}}\Big[\Big\langle \Big((\bbeta_0-\bbeta)^\intercal\boldsymbol{\Phi}^\intercal \LLambda^{(j)}\Big)_{s}(\bbeta_0^{(j)}-\bbeta^{(j)})_m\Big\rangle\Big],\\
    I^3_{j,sm}&= \ex_{\mathbf{W_1}}\Big[\sqrt{\frac{1-t}{p}}\sum_{i=1}^k\sqrt{R_{2,i}(t)l_i}\Big\langle \Big((\V^{(i)}-\v^{(i)})^\intercal\LLambda^{(i)}\LLambda^{(j)}\Big)_{s}(\bbeta_0^{(j)}-\bbeta^{(j)})_m\Big\rangle\Big].
\end{align*}
It follows that
\begin{align*}I^1_j&:=\sum_{s,m}\ex[(I^1_{j,sm})^2]\\
&\leq \sum_{s,m}\ex\Big[\frac{1-t}{p}\ex_{\mathbf{W_1}}\langle(\Z^\intercal\LLambda^{(j)})^2_s(\bbeta_0^{(j)}-\bbeta^{(j)})^2_m\rangle\Big]\\
&\leq p^{-1} \ex\Big[\ex_{\mathbf{W_1}}\langle(\|\LLambda^{(j)} \Z\|^2\|\bbeta_0^{(j)}-\bbeta^{(j)}\|^2\rangle\Big]\\
&\leq Cp^{-1}\sqrt{\ex\langle\|\Z\|^4\rangle\ex\langle\|\bbeta_0^{(j)}-\bbeta^{(j)}\|^4\rangle}\\
&\leq C\sqrt{\ex\|\bbeta_0^{(j)}\|^4+\ex\langle\|\bbeta^{(j)}\|^4\rangle}\\
&\leq C\sqrt{\ex\|\bbeta_0^{(j)}\|^4}=  O(p)\end{align*}
where the third inequality follows from the boundedness of the norm of $\LLambda^{(j)}$, the fourth inequality follows using the fact that $\ex\|\Z\|^4=O(p^2)$ and the last inequality follows from the Nishimori identity.

Similarly we obtain 
\begin{align*}I^2_j&=\sum_{s,m}\ex[(I^2_{j,sm})^2]\\
&\leq \sum_{s,m}\ex\Big[\frac{(1-t)^2}{p^2}\ex_{\mathbf{W_1}}\langle((\bbeta_0-\bbeta)^\intercal\boldsymbol{\Phi}^\intercal\LLambda^{(j)})^2_s(\bbeta_0^{(j)}-\bbeta^{(j)})^2_m\rangle\Big]\\
&\leq p^{-2} \ex\Big[\ex_{\mathbf{W_1}}\langle(\|\LLambda^{(j)} \boldsymbol{\Phi} (\bbeta_0-\bbeta)\|^2\|\bbeta_0^{(j)}-\bbeta^{(j)}\|^2\rangle\Big]\\
&\leq Cp^{-2}\sqrt{\ex\langle\|\boldsymbol{\Phi} (\bbeta_0-\bbeta)\|^4\rangle\ex\langle\|\bbeta_0^{(j)}-\bbeta^{(j)}\|^4\rangle}\\
&\leq  C p\end{align*}
where we have used the fact that
$$\ex\langle\|\boldsymbol{\Phi} (\bbeta_0-\bbeta)\|^4\rangle \leq (\ex\|\boldsymbol{\Phi}\|^8 \ex \langle\|\bbeta_0-\bbeta\|^8\rangle)^{1/2}\leq Cp^4.$$

For the last term we observe that
\begin{align*}I^3_j&=\sum_{s,m}\ex[(I^3_{j,sm})^2]\\
&\leq \sum_{s,m}\ex\Big[\frac{(1-t)}{p}\sum_{i=1}^k R_{2,i}(t)l_i\ex_{\mathbf{W_1}}\Big\langle((\V^{(i)}-\v^{(i)})^\intercal\LLambda^{i}\LLambda^{(j)})^2_s(\bbeta_0^{(j)}-\bbeta^{(j)})^2_m\Big\rangle\Big]\\
&\leq \ex\Big[\frac{(1-t)}{p}\sum_{i=1}^k R_{2,i}(t)l_i\ex_{\mathbf{W_1}}\Big\langle\|\LLambda^{(j)} \LLambda^{(i)} (\V^{(i)}-\v^{(i)})\|^2\|\bbeta_0^{(j)}-\bbeta^{(j)}\|^2\Big\rangle\Big]\\
&\leq Cp^{-1}\sum_{i=1}^k\sqrt{\ex\langle\|\V^{(i)}-\v^{(i)}\|^4\rangle\ex\langle\|\bbeta_0^{(j)}-\bbeta^{(j)}\|^4\rangle}\\
&\leq  C p\end{align*}
where we have used $\ex\langle \|\V^{(i)}-\v^{(i)}\|^4\rangle\leq C \ex\|\V^{(i)}\|^4\leq Cp^{2}$. Thus we have
$$\sum_{sm}\ex\Big[\frac{1}{p}\frac{\partial h}{\partial G^{(j)}_{sm}}\Big]^2\leq Cp^{-1}$$
and a similar argument gives $\sum_{m}\ex\Big[\frac{1}{p}\frac{\partial h}{\partial V^{(j)}_m}\Big]^2\leq Cp^{-1}.$
By Gaussian Poincare inequality the claim of the lemma follows.

\end{proof}





\begin{lemma}\label{free_conc_4}
    $$\ex\Big[\Big|\ex_{\mathbf{W_2}}\Big[\frac{1}{p}\log \mathcal{Z}\Big]-\ex\Big[\frac{1}{p}\log \mathcal{Z}\Big]\Big|^2\Big]\leq Cp^{-1}$$
\end{lemma}

\begin{proof}
This follows from the following bounded difference inequality. Indeed, define $h(\bbeta)=\ex_{\mathbf{W_2}}\Big[\frac{1}{p}\log \mathcal{Z}\Big]$. Next, let $\bbeta^{\rightarrow j}$ be the vector such that $\beta^{\rightarrow j}_i=\beta_i$ for $i\neq j$ and $\beta_j^{\rightarrow j}=\tilde{\beta}_j$ (i.e. replaced by an i.i.d. copy).

Define $\psi(s)=h(s\bbeta+(1-s)\bbeta^{\rightarrow j})$ so that $\psi(1)=h(\bbeta)$ and $\psi(0)=h(\bbeta^{\rightarrow j})$. Showing that $\psi'(s)\leq Cp^{-1}$ would imply
$$h(\bbeta)-h(\bbeta^{\rightarrow j})\leq Cp^{-1}.$$ Bounded difference inequality would then imply the statement of the lemma.
\end{proof}

\section{Proof of  Lemmas  \ref{lemma_y2_1} and  \ref{lemma_y2_2}}\label{apx:lemmas}

\begin{proof}[Proof of lemma \ref{lemma_y2_1}]
Define the following quantities:
\begin{equation}Q_i(\bbeta,\bbeta')=\frac{1}{p}\bbeta^{(i)\intercal}\bbeta^{'(i)},\quad \tilde{Q}_i(\bbeta',\bbeta)=\frac{1}{Np}(\boldsymbol{\Phi}\bbeta^\prime)^\intercal\mathbf{E}^2\LLambda_{i,N}\boldsymbol{\Phi}\bbeta.\label{overlap_defn}\end{equation}
Note that we have
\begin{align*}
&-\frac{1}{Np^2}\ex[\langle(\boldsymbol{\Phi}\bbeta^\prime)^\intercal\mathbf{E}^2\LLambda_{i,N}(\boldsymbol{\Phi}(\bbeta_0-\bbeta))(\bbeta^{(i)}_0-\bbeta^{(i)})^\intercal\bbeta^{(i)}\rangle\\
=&\ex[\langle Q_i(\bbeta,\bbeta)\tilde{Q}_i(\bbeta',\bbeta_0)\rangle]-\ex[\langle Q_i(\bbeta,\bbeta)\tilde{Q}_i(\bbeta',\bbeta)\rangle]+\ex[\langle Q_i(\bbeta_0,\bbeta)\tilde{Q}_i(\bbeta',\bbeta)\rangle]-\ex[\langle Q_i(\bbeta_0,\bbeta)\tilde{Q}_i(\bbeta',\bbeta_0)\rangle]
\end{align*}

The third and fourth terms of the RHS cancel each other due to the Nishimori identity and the first two terms can be decomposed as follows:

\begin{align*}
&|\ex[\langle Q_i(\bbeta,\bbeta)\tilde{Q}_i(\bbeta',\bbeta_0)\rangle]-\ex[\langle Q_i(\bbeta,\bbeta)\tilde{Q}_i(\bbeta',\bbeta)\rangle]|\\
=&\Big|\ex\Big[\Big\langle (Q_i(\bbeta,\bbeta)- \ex[\langle Q_i(\bbeta,\bbeta)\rangle])\tilde{Q}_i(\bbeta',\bbeta_0)\Big\rangle\Big]-\ex\Big[\Big\langle (Q_i(\bbeta,\bbeta)- \ex[\langle Q_i(\bbeta,\bbeta)\rangle])\tilde{Q}_i(\bbeta',\bbeta)\Big\rangle\Big]\Big|\\
\leq &  \Big(\ex[\langle (Q_i(\bbeta,\bbeta)- \ex[\langle Q_i(\bbeta,\bbeta)\rangle])^2\rangle]\Big)^{1/2}\Big((\ex[\langle \tilde{Q}^2_i(\bbeta',\bbeta_0)\rangle])^{1/2}+(\ex[\langle \tilde{Q}^2_i(\bbeta',\bbeta)\rangle])^{1/2}\Big)
\end{align*}
where the second equality follows since $$\ex[\langle Q_i(\bbeta,\bbeta)\rangle]\ex[\langle\tilde{Q}_i(\bbeta',\bbeta_0)\rangle\rangle]=\ex[\langle Q_i(\bbeta,\bbeta)\rangle]\ex[\langle\tilde{Q}_i(\bbeta',\bbeta)\rangle]$$ by Nishimori identity and the third inequality is just Cauchy-Schwarz inequality. Thus to show \eqref{err_Y_2} it is enough to show that

\begin{align}\ex[\langle \tilde{Q}^2_i(\bbeta',\bbeta)\rangle])^{1/2}&=O(1)\label{pred_ovlp_bnd}\\
\ex[\langle ((Q_i(\bbeta,\bbeta)- \ex[\langle Q_i(\bbeta,\bbeta)\rangle])^2\rangle]&=O(p^{-1})\label{self_ovlp_bnd}
\end{align}

The bound \eqref{pred_ovlp_bnd} can be shown as follows:
\begin{align}\ex[\langle \tilde{Q}^2_i(\bbeta',\bbeta)\rangle])^{1/2}&
\leq \frac{C}{Np} \ex[\langle\|\Phi\|^2\|\bbeta'\|\|\bbeta\|\rangle]\nonumber\\
&\leq \frac{C}{Np}(\ex[\|\Phi\|^4])^{1/2}(\ex[\langle\|\bbeta'\|^4\rangle])^{1/4}(\ex[\langle\|\bbeta\|^4\rangle])^{1/4}\nonumber\\
&= \frac{C}{Np}(\ex[\|\Phi\|^4])^{1/2}(\ex[\|\bbeta_0\|^4])^{1/2}\label{tildeQ_l2_bound}
\end{align}
where the last equality uses the Nishimori identity. Noting that $\ex[\|\Phi\|^4=O(p^2)$ and $\ex[\|\bbeta_0\|^4]=O(p^2)$ the desired conclusion follows.

On the other hand by Nishimori identity the LHS of \eqref{self_ovlp_bnd} is equal to $\ex[ (Q_i(\bbeta_0,\bbeta_0)- \ex[Q_i(\bbeta_0,\bbeta_0)])^2]$. Noting that the components of the signal $\bbeta_0$ are i.i.d. random variables with bounded support an application of Chebyschev's inequality gives the the bound in \eqref{self_ovlp_bnd}.

\end{proof}

\begin{proof}[Proof of lemma \ref{lemma_y2_2}]

 Using Cauchy Schwarz inequality we obtain
\begin{align*}
&\frac{1}{Np^2}\ex[\langle(\boldsymbol{\Phi}\bbeta)^\intercal\mathbf{E}^2\LLambda_{i,N}\boldsymbol{\Phi}(\bbeta_0-\bbeta)(\bbeta^{(i)}_0-\bbeta^{(i)})^\intercal\bbeta^{(i)}\rangle]\\
=&\frac{1}{Np^2}\ex[\langle(\boldsymbol{\Phi}\bbeta)^\intercal\mathbf{E}^2\LLambda_{i,N}\boldsymbol{\Phi}(\bbeta_0-\bbeta)\rangle]\ex[\langle(\bbeta^{(i)}_0-\bbeta^{(i)})^\intercal\bbeta^{(i)}\rangle]\\
&+ \frac{1}{Np^2}\ex[\langle((\boldsymbol{\Phi}\bbeta)^\intercal\mathbf{E}^2\LLambda_{i,N}\boldsymbol{\Phi}(\bbeta_0-\bbeta)-\ex[\langle(\boldsymbol{\Phi}\bbeta)^\intercal\mathbf{E}^2\LLambda_{i,N}\boldsymbol{\Phi}(\bbeta_0-\bbeta)\rangle])(\bbeta^{(i)}_0-\bbeta^{(i)})^\intercal\bbeta^{(i)}\rangle]\\
=&\frac{1}{Np^2}\ex[\langle(\boldsymbol{\Phi}\bbeta)^\intercal\mathbf{E}^2\LLambda_{i,N}\boldsymbol{\Phi}(\bbeta_0-\bbeta)\rangle]\ex[\langle(\bbeta^{(i)}_0-\bbeta^{(i)})^\intercal\bbeta^{(i)}\rangle]\\
&+ \frac{1}{Np^2}\Big(\ex\Big[\langle((\boldsymbol{\Phi}\bbeta)^\intercal\mathbf{E}^2\LLambda_{i,N}\boldsymbol{\Phi}(\bbeta_0-\bbeta)-\ex[\langle(\boldsymbol{\Phi}\bbeta)^\intercal\mathbf{E}^2\LLambda_{i,N}\boldsymbol{\Phi}(\bbeta_0-\bbeta)\rangle])^2\rangle\Big]\Big)^{1/2}\Big(\ex[\langle((\bbeta^{(i)}_0-\bbeta^{(i)})^\intercal\bbeta^{(i)})^2\rangle]\Big)^{1/2}\\
\end{align*}
Since $p^{-1}(\bbeta^{(i)}_0-\bbeta^{(i)})^\intercal\bbeta^{(i)}$ is of order $O(1)$, to prove the lemma, it is sufficient to show that
$$\frac{1}{Np}\Big(\ex\Big[\langle((\boldsymbol{\Phi}\bbeta)^\intercal\mathbf{E}^2\LLambda_{i,N}\boldsymbol{\Phi}(\bbeta_0-\bbeta)-\ex[\langle(\boldsymbol{\Phi}\bbeta)^\intercal\mathbf{E}^2\LLambda_{i,N}\boldsymbol{\Phi}(\bbeta_0-\bbeta)\rangle])^2\rangle\Big]\Big)^{1/2}=o(1).$$
To proceed we first show that in the expression above we can replace $\mathbf{E}^2\LLambda_{i}$ by $\LLambda^{(i)}\mathbf{E}^2\LLambda^{(i)}$, i. e.
\begin{align}\label{eq:com_lamb}
    \frac{1}{(Np)^2}\ex\langle|(\boldsymbol{\Phi}\bbeta)^\intercal\mathbf{E}^2\LLambda_{i,N}\boldsymbol{\Phi}(\bbeta_0-\bbeta)-(\boldsymbol{\Phi}\bbeta)^\intercal\LLambda^{(i)}\mathbf{E}^2\LLambda^{(i)}\boldsymbol{\Phi}(\bbeta_0-\bbeta)|^2\rangle=o(1).
\end{align}
Since $\boldsymbol{\Phi}\bbeta=\sum_{j=1}^k\LLambda^{(j)}\G^{(j)}\bbeta^{(j)}$,  applying Jensen inequality, \eqref{eq:com_lamb} becomes true if we show 
\begin{align}
     I:=\frac{1}{(Np)^2}\ex\langle|\bbeta^{{(j)}\intercal}\G^{(j)\intercal}(\LLambda^{(j)}\mathbf{E}^2\LLambda_{i,N}-\LLambda^{(i)}\mathbf{E}^2\LLambda^{(i)}\LLambda^{(j)})\G^{(j)}(\bbeta_0^{(j)}-\bbeta^{(j)})|^2\rangle=o(1).
\end{align}
We denote matrix $\boldsymbol{\Theta}=\LLambda^{(j)}\mathbf{E}^2\LLambda_{i,N}-\LLambda^{(i)}\mathbf{E}^2\LLambda^{(i)}\LLambda^{(j)}$. Two terms of this matrix are  continuous functions of  of KMS matrices (see Appendix~\ref{apx:KMS}),  which means that asymptotically, they share the same eigenspace. For us it means that eigenvalue distribution of matrix $\boldsymbol{\Theta}$ is concentrated at 0, with a finite number (zero measure) outliers. Since the operator norm of $\boldsymbol{\Theta}$ is bounded by some constant $C$ that do not depend on $N$, the outliers also will be bounded by some general  constant $C$.  With singular value decomposition we localize the bigger singular values  $\boldsymbol{\Theta}=\boldsymbol{\Theta}_{\epsilon}+\sigma_1\mathbf{u}_1\mathbf{v}_1^\intercal+\ldots$, the number of such eigenvalues is fixed and do not grow with $N$, so we focus only on one of them, the presence of others is taken into account by the general constant $C$. Here, the operator norm  $\|\boldsymbol{\Theta}_{\epsilon}\|=o(1)$. Then we have
\begin{align*}
     &I\leq\frac{C}{(Np)^2}\ex\langle|\bbeta^{{(j)}\intercal}\G^{(j)\intercal}\mathbf{u}\mathbf{v}^\intercal\G^{(j)}(\bbeta_0^{(j)}-\bbeta^{(j)})|^2\rangle+o(1)\\
     &\leq \frac{C}{(Np)^2}\Big(\ex\langle|\bbeta^{{(j)}\intercal}\G^{(j)\intercal}\mathbf{u}|^4\rangle\ex\langle|\mathbf{v}^\intercal\G^{(j)}(\bbeta_0^{(j)}-\bbeta^{(j)})|^4\rangle\Big)^{1/2}+o(1)
     \leq \frac{C}{Np}\Big(\ex|\bbeta_0^{{(j)}\intercal}\G^{(j)\intercal}\mathbf{u}|^4\Big)^{1/2}+o(1)
\end{align*}
For obtain the last inequality we bounded the second expectation  by $(Np)^2$ and applied Nishimori for the first expectation. Now we can present vector $\mathbf{u}$ as $\mathbf{U}\mathbf{e}_1$, where matrix $\mathbf{U}$ is unitary and  vector $\mathbf{e}_1$ is an $N$~- dimensional standard basis  vector with  the first element equal to $1$ and others equal to $0$. Since Gaussian matrix is right-rotationally invariant we obtain
\begin{align}
    I\leq \frac{C}{Np}\Big(\ex|\bbeta_0^{{(j)}\intercal}\G^{(j)\intercal}\mathbf{e}_1|^4\Big)^{1/2}+o(1)
    =\frac{C}{Np}\Big(\ex|\bbeta_0^{{(j)}\intercal}\G^{(j)}_1|^4\Big)^{1/2}+o(1)
\end{align}
Conditionally on $\bbeta_0$, the expression $\bbeta_0^{{(j)}\intercal}\G^{(j)}_1$ is a Gaussian variable with zero mean and variance $\|\bbeta_0^{(j)}\|^2$, then its fourth moment is $3\|\bbeta_0^{(j)}\|^4=p^2O(1)$. This finally gives us the desired bound:
\begin{align}
    I\leq \frac{C}{Np}\Big(p^2\Big)^{1/2}+o(1)=o(1).
\end{align}
For brevity we write  $\LLambda^{(i)}\mathbf{E}^2\LLambda^{(i)}=:\bB^\intercal\bB$.
Define $\tPhi=\frac{\bB\Phi}{\sqrt{p}}$, so that  it is enough to show that $\q:=\frac{1}{N}(\boldsymbol{\tPhi}\bbeta)^\intercal\boldsymbol{\tPhi}(\bbeta_0-\bbeta)$ concentrates around its expected value. The proof now follows in the same vein as in \cite{barbier2020mutual}. One defines $\M$ as

$$\M:=\frac{1}{p}\sum_{\mu=1}^N\frac{(\tPhi \beta)^2_{\mu}}{2}-(\tPhi \beta)_{\mu}(\tPhi \beta_0)_{\mu}-\frac{\sigma}{2}(\tPhi \beta)_{\mu}Z_{\mu}.$$

The following equations hold:
\begin{align}
&\int_\epsilon^a d\sigma^{-1} \ex [\langle \ (\M-\langle \M \rangle)^2 \rangle]=o(1),\label{m_1}\\
& \int_\epsilon^a d\sigma^{-1} \ex [\ (\ex \langle\M\rangle -\langle \M \rangle)^2 ]=o(1),\label{m_2}\\
&  \int_\epsilon^a d\sigma^{-1} \ex [\langle \ (\q-\ex \langle \q \rangle)^2 \rangle]=o(1).\label{q_1}
\end{align}
We observe that \eqref{m_1} follows same proof as Lemma 9.1 of \cite{barbier2020mutual}. We  define the Hamiltonian with $\tPhi$ replaced by $\Phi$ and show that:
$$\frac{d^2i_p}{(d\sigma^{-1})^2}=-p\ex[\langle\mathcal{M}^2-\langle \mathcal{M}\rangle^2\rangle]+\frac{\sigma^{3/2}}{4p}\sum_{\mu=1}^N\ex[\langle(\tPhi \beta)_\mu Z_{\mu}\rangle]$$
\eqref{m_1} then follows by integrating w.r.t. $\sigma^{-1}$ and then showing that the derivatives of the mutual information is bounded and the second term in the RHS is $O(1)$.

The second equation is equivalent to Lemma 9.2 of \cite{barbier2020mutual} . It uses the concavity of the free energy coupled with the fact that $\mathrm{Var}[a^2]=O(p^{-1})$, where,
\begin{equation}a:=\frac{1}{p}(\|\tPhi \beta_0\|^2-\ex[\|\tPhi \beta_0\|^2]).\label{norm_conc}\end{equation}

Now while the above relation was shown in \cite{barbier2020mutual} using the independence of the entries of $\Phi$, similar condition fails to hold in the present setup. Instead we directly show that variance of $\frac{1}{p}(\|\tPhi \beta_0\|^2)$ is $O(p^{-1})$. Indeed conditionally on $\bbeta_0$, vector $\tPhi\bbeta_0$ is Gaussian with covariance matrix
\begin{align}
    \mathbf{R}=\ex_{\G}(\tPhi\bbeta_0\bbeta_0^\intercal\tPhi^\intercal)=\frac{1}{p}\sum_{j=1}^k\ex_{\G}(\bB\LLambda^{(j)}\G^{(j)}\bbeta_0^{(j)}\bbeta_0^{(j)\intercal}\G^{(j)\intercal}\LLambda^{(j)}\bB^\intercal)=\frac{1}{p}\sum_{j=1}^k\|\bbeta_0^{(j)}\|^2\bB\LLambda_j\bB^\intercal.
\end{align}
For a Gaussian vector $\mathbf{x}\sim \mathcal{N}(0,\mathbf{R})$, we have  $\mathrm{Var}\|\mathbf{x}\|^2=2\Tr \mathbf{R}\mathbf{R}^\intercal$. Since  distribution $P_0$ of $\bbeta_0$ has bounded support and $\|\bB\LLambda^{(j)}\|^2\leq 1/(1-\lambda_j)^2$, we can obtain
\begin{equation}
    \mathrm{Var}(p^{-1}\|\tPhi \beta_0\|^2)\le \frac{C}{p}\|\mathbf{R}\|^2\le\sum_{j=1}^k\frac{C\|\bbeta_0^{(j)}\|^2}{p^2}\le\frac{C}p. \label{var_a_bound}
\end{equation}

Finally we observe that \eqref{q_1} is the equivalent of  
 Lemma 9.3 of \cite{barbier2020mutual}. First one shows that
 \begin{equation}\int_\epsilon^a d\sigma^{-1} \ex [\langle \ (\q \mathcal{M})\rangle-\ex[\langle \q \rangle ]\ex[\langle \mathcal{M} ]\rangle]\leq (\int_\epsilon^a d\sigma^{-1} \ex [\langle \q^2\rangle])^{1/2}(\int_\epsilon^a d\sigma^{-1} \ex [\langle \mathcal{M}-\ex[\mathcal{M}]\rangle]^2)^{1/2}=o(1)\label{q11}\end{equation}
 which follows from the boundedness of $\ex [\langle \q^2\rangle]$ and the combination of \eqref{m_1} and \eqref{m_2}.

 In the next step one shows that
 \begin{equation}\ex [\langle \ (\q \mathcal{M})\rangle-\ex[\langle \q \rangle ]\ex[\langle \mathcal{M} ]\rangle]\geq \frac{1}{2}\ex [\langle \ (\q-\ex \langle \q \rangle)^2 \rangle]+ \frac{C}{p}(\ex[\langle\|\tPhi \beta\|^2-\ex[\langle\|\tPhi \beta\|^2\rangle]\rangle^2)^{1/2}\label{q12}\end{equation}
 By Nishimori identity we have
\begin{equation}\frac{1}{p}(\ex[\langle\|\tPhi \beta\|^2-\ex[\langle\|\tPhi \beta\|^2\rangle]\rangle^2)^{1/2}=\ex[a^2]^{1/2}=O(p^{-1/2})\label{q13}\end{equation}
which follows from \eqref{var_a_bound}. Finally combining \eqref{q11}, \eqref{q12}, and \eqref{q13}, we obtain \eqref{q_1}.
 

\end{proof}

\end{document}